\documentclass[12pt]{article}
\usepackage{amssymb}
\usepackage{times}
\usepackage{epsfig}
\usepackage{xr}
\newtheorem{theorem}{Theorem}[section]

\newtheorem{lemma}[theorem]{Lemma}
\newtheorem{proposition}[theorem]{Proposition}

\newcommand{\proofendsign}{\hfill\hbox{\vrule width 7pt depth 0pt height 7pt} \par\vspace{10pt}}
\newenvironment{proof}{\par\noindent {\it Proof:} \hspace{7pt}}%
{\proofendsign}

\newcommand{\bbbone}{{\mathchoice {\rm 1\mskip-4mu l} {\rm 1\mskip-4mu l}    {\rm 1\mskip-4.5mu l} {\rm 1\mskip-5mu l}}}

\newcommand{\bC}{{\mathbb C}}

\newcommand{\bM}{{\mathbb M}}
\newcommand{\bN}{{\mathbb N}}

\newcommand{\bR}{{\mathbb R}}

\newcommand{\bX}{{\mathbb X}}

\newcommand{\bZ}{{\mathbb Z}}
\newcommand{\al}{\alpha}
\newcommand{\be}{\beta}

\newcommand{\de}{\delta}
\newcommand{\ep}{\epsilon}
\newcommand{\et}{\eta}
\newcommand{\ze}{\zeta}
\renewcommand{\th}{\theta} 

\newcommand{\ka}{\kappa}
\newcommand{\la}{\lambda}
\newcommand{\rh}{\rho}
\newcommand{\si}{\sigma}

\newcommand{\ps}{\psi}
\newcommand{\om}{\omega}

\newcommand{\vth}{\vartheta}
\newcommand{\veps}{\varepsilon}
\newcommand{\Th}{\Theta}
\newcommand{\La}{\Lambda}
\newcommand{\Si}{\Sigma}

\newcommand{\bGa}{{{\rm I}\kern-.16em \Gamma}}
\newcommand{\psq}{{\bar\psi}}

\newcommand{\cA}{{\cal A}}
\newcommand{\cB}{{\cal B}}

\newcommand{\cF}{{\cal F}}

\newcommand{\cI}{{\cal I}}
\newcommand{\cJ}{{\cal J}}

\newcommand{\cM}{{\cal M}}

\newcommand{\cO}{{\cal O}}

\newcommand{\del}{\partial}
\newcommand{\gtoas}[1]{{\quad\mathop{\longrightarrow}\limits_{#1}\quad}}
\newcommand{\ve}[1]{{\bf #1}}

\newcommand{\abs}[1]{{\left\vert #1 \right\vert}}
\newcommand{\norm}[1]{{\left\Vert #1 \right\Vert}}

\newcommand{\Ref}[1]{$(\ref{#1})$}
\newcommand{\True}[1]{\; \bbbone\left( #1 \right) \;}
\newcommand{\sfrac}[2]{{\textstyle \frac{#1}{#2}}}

\newcommand{\vol}{{\hbox{ vol }}}

\newcommand{\const}{\hbox{ \rm const }}

\newcommand{\half}{\frac{1}{2}}
\newif\ifintrmk
\newcommand{\sst}{\scriptstyle}

\newcommand{\tst}{\textstyle}

\def\suffix{ps}
\newcount\system
\global\system=3   

\def\ifundefined#1{\expandafter\ifx\csname#1\endcsname\relax}
\ifundefined{figdir}\def\figdir{}\fi
\newcount\firstline
\newdimen\pswidth  \newdimen\xleft
\newdimen\psheight \newdimen\ytop \newdimen\ybot
\newcount\justx \newcount\justy
\global\justx=0 \global\justy=0
\newdimen\vpos \newtoks\labeL 
\newread\labeLfile \newdimen\xcoord \newdimen\ycoord
\newif\ifdoit 
\newbox\labox
\newdimen\xdvikwid 
\newdimen\xdvikht
\newdimen\pspoints
\newdimen\rwi
\newcount\temp
\def\readdim#1{\global\read\labeLfile to \temp
\global #1=\temp pt}
\def\figcrop#1{\par
\openin\labeLfile=\figdir#1.lbl                                              
\global\read\labeLfile to\firstline\message{#1}               
\global\read\labeLfile to\temp
\readdim{\ybot}
\readdim{\xleft}
\readdim{\ytop}
\global\read\labeLfile to\justx
\global\read\labeLfile to\justy
\global\read\labeLfile to\labeL
\readdim{\pswidth}
\global\advance\pswidth by -\xleft
\readdim{\psheight}
\global\advance\ybot by -\psheight
\global\advance\psheight by -\ytop
\global\read\labeLfile to\justx
\global\read\labeLfile to\justy
\global\read\labeLfile to\labeL
\vbox to\psheight{\vfill
\ifnum\system=1
\fi 
\ifnum\system=2
\fi 
\ifnum\system=3
                                                 \fi         
\ifnum\system=4
\fi         
\ifnum\system=1
\hbox to \pswidth{\kern-\xleft\special{postscriptfile \figdir#1.\suffix }\hfil}\fi
\ifnum\system=2
\hbox to \pswidth{\kern-\xleft\special{ps: plotfile \figdir#1.\suffix }\hfil}\fi
\ifnum\system=3
\hbox to \pswidth{\kern-\xleft\includegraphics{\figdir#1.\suffix}\hfil}\fi
\ifnum\system=4
\hbox to \pswidth{\kern-\xleft\includegraphics{\figdir#1.\suffix}\hfil}\fi
\ifnum\system=5
\hbox to \pswidth{\kern-\xleft\includegraphics{\figdir#1.\suffix}\hfil}\fi 
\ifnum\system=6
   \xdvikwid=\pswidth
   \xdvikht=\psheight
   {\global\divide\xdvikwid by \pspoints}
   {\global\divide\xdvikht by \pspoints}
   \rwi=\xdvikwid
    {\global\multiply\rwi by 10}
\hbox to \pswidth{\kern-\xleft\includegraphics{\figdir#1.\suffix\space}\hfil}\fi                   
\vskip -\baselineskip
\vskip -\ybot 
\vskip-\psheight %
\hbox to\pswidth  {\hss}%
\parindent=0pt\offinterlineskip                                       
\vpos=0 pt%
\loop\readdim{\xcoord}                                 
\ifdim \xcoord < -999pt \doitfalse\else\doittrue\fi                        
\ifdoit \advance \xcoord by -\xleft
\readdim{\ycoord}
\advance \ycoord by -\ytop                              
\global\read\labeLfile to\justx                                       
\global\read\labeLfile to\justy                                       
\global\read\labeLfile to\labeL
\global\setbox\labox=\hbox{\labeL\hskip-0.3em}%
\advance\vpos by-\ycoord                                              
\vskip-\vpos \vpos=\ycoord                                         
\hbox to\pswidth{\hskip\xcoord %
\hbox to 0pt{\ifnum\justx>0\hss\fi%
\vbox to0pt{%
\ifnum\justy<2\vss\fi%
\copy\labox\kern0pt%
\ifnum\justy>0\vss\fi}%
\ifnum\justx<2\hss\fi}%
\hss}%
\repeat%
\advance\vpos by-\psheight%
\vskip-\vpos %
}\closein\labeLfile}
\def\figplace#1#2#3{
\openin\labeLfile=\figdir#1.lbl
\ifeof \labeLfile
       \immediate\write16{***Can't find \figdir#1.lbl; Skipping it.***}
\else  \closein\labeLfile
       \null\hskip#2\raise #3 \hbox{\figcrop{#1}}
\fi
}
\def\figput#1{
\openin\labeLfile=\figdir#1.lbl
\ifeof \labeLfile
       \immediate\write16{***Can't find \figdir#1.lbl; Skipping it.***}
\else  \closein\labeLfile
       \hbox{\figcrop{#1}}
\fi
}

\def\figdir{fig/}

\newcommand{\set}[2]{\big\{ \ #1\ \big|\ #2\ \big\}}

\newcommand{\E}{{\rm e}}
\newcommand{\I}{{\rm i}}
\newcommand{\dd}{{\rm d}}

\newcommand{\Feb}{f_\beta}
\newcommand{\Beb}{b_\beta}
\newcommand{\Deb}{\delta_\beta}
\newcommand{\rx}{{\rm x}}
\newcommand{\ry}{{\rm y}}

\newcommand{\Ha}{{\bf H1}}
\newcommand{\Hb}{{\bf H2}}
\newcommand{\Hc}{{\bf H3}}
\newcommand{\Hd}{{\bf H4}}
\newcommand{\He}{{\bf H5}}
\newcommand{\Hf}{{\bf H6}}

\renewcommand{\const}{\mbox{ \rm const }}
\newcommand{\bx}{{\bf x}}
\newcommand{\bk}{{\bf k}}
\newcommand{\bzer}{{\bf 0}}

\newcommand{\bp}{{\bf p}}
\newcommand{\bq}{{\bf q}}

\newcommand{\cst}[2]{{\rm const}^{#1}_{#2}}

\newcommand{\bnabla}{\nabla}
\renewcommand{\vol}{\mbox{ \rm  Vol }}

\newcommand{\tze}{\tilde\zeta}

\newcommand{\mM}{\bM}

\newcommand{\rv}{{\rm v}}

\newcommand{\dimp}{{\bar\veps}}

\newcommand{\pmb}[1]{\setbox0=\hbox{#1}       
     \kern-.025em\copy0\kern-\wd0
     \kern.05em\copy0\kern-\wd0
     \kern-.025em\box0}             

\begin{document}

\title{Singular Fermi Surfaces II. \\ The Two--Dimensional Case}

\author{Joel Feldman$^1$\footnote{feldman@math.ubc.ca; supported by NSERC of Canada}
 and Manfred Salmhofer$^{1,2}$\footnote{salmhofer@itp.uni-leipzig.de; supported by DFG grant Sa-1362, an ESI senior research fellowship, and NSERC of Canada} \\
\small  $^1$ Mathematics Department, The University of British Columbia \\
\small 1984 Mathematics Road, Vancouver, B.C., Canada V6T 1Z2
\\
\small $^2$ Theoretische Physik, Universit\"{a}t Leipzig \\
\small Postfach 100920, 04009 Leipzig, Germany
}

\date{\today}

\maketitle

\begin{abstract}
\noindent
We consider many--fermion systems with singular Fermi surfaces, 
which contain {\em Van Hove points} 
where the gradient of the band function $\bk \mapsto e(\bk)$ vanishes. 
In a previous paper, we have treated the case of spatial dimension $d \ge 3$. 
In this paper, we focus on the more singular case $d=2$
and establish properties of the fermionic self--energy 
to all orders in perturbation theory. 
We show that there is an asymmetry 
between the spatial and frequency derivatives of the self--energy. 
The derivative with respect to the Matsubara frequency 
diverges at the Van Hove points, but, surprisingly, 
the self--energy is $C^1$  in the spatial momentum 
to all orders in perturbation theory, 
provided the Fermi surface is curved away from the Van Hove points.
In a prototypical example, 
the second spatial derivative behaves similarly to the 
first frequency derivative. 
We discuss the physical significance of these findings.
\end{abstract}

\tableofcontents
\goodbreak

\section{Introduction}

In this paper, we continue our analysis of (all--order) perturbative 
properties of the self--energy and the correlation functions 
in fermionic systems with a fixed non--nested singular Fermi surface.
That is, the Fermi surface contains {\em Van Hove points}, 
where the gradient of the dispersion function vanishes, 
but satisfies a no--nesting condition away from these points. 
In a previous paper \cite{SFS1}, we treated the case of spatial dimensions $d \ge 3$. 
Here we focus on the two--dimensional case, where the effects
of the Van Hove points are strongest. 

We have given a general introduction to the problem and some 
of the main questions in \cite{SFS1}. 
As discussed in \cite{SFS1}, the no--nesting hypothesis is 
natural from a theoretical point of view, because it 
separates effects coming from the saddle points and nesting 
effects. Moreover, generically, nesting and Van Hove effects do not 
occur at the same Fermi level. 
In the following, we discuss those aspects of the problem that 
are specific to two dimensions. 
As already discussed in \cite{SFS1}, the effects caused by saddle
points of the dispersion function lying on the Fermi  surface are 
believed to be strongest in two dimensions (we follow the usual 
jargon of calling the level set the Fermi ``surface'' even though 
it is a curve in $d=2$). 
Certainly, the Van Hove singularities in the density of states of the noninteracting
system are strongest in $d=2$.  As concerns many--body properties, 
we have shown in \cite{SFS1} that for $d \ge 3$, the overlapping loop 
estimates of \cite{FST1} carry over essentially unchanged, which implies
differentiability of the self--energy and hence a 
quasiparticle weight ($Z$--factor) close to 1
to all orders in renormalized perturbation theory.  

In this paper, we show that for $d=2$, there are more drastic changes.
Namely, there is an asymmetry between the derivatives of the self--energy
$\Sigma (q_0,\bq) $ with respect to the frequency variable $q_0$ 
and the spatial momentum $\bq$. 
We prove that the spatial gradient $\nabla \Sigma$ is a bounded function
to all orders in perturbation theory if the Fermi surface satisfies a
no--nesting condition. By explicit calculation, we show that for a standard
saddle point singularity, even the second--order contribution 
$\partial_0\Sigma_2 (q_0,\bq_s)$ diverges as $(\log |q_0|)^2$ at any 
Van Hove point 
$\bq_s$ (if that point is on the Fermi surface).

This asymmetric behaviour is unlike the behaviour in all other cases
that are under mathematical control: in one dimension, both $\del_0 \Sigma_2$
and $\del_1 \Sigma_2$ diverge like $\log |q_0|$ at the Fermi point. 
This is the first indication for vanishing of the $Z$--factor and the 
occurrence of anomalous decay exponents in this model. The point is, 
however, that once a suitable $Z$--factor is extracted, $ Z \del_1 \Sigma_2$
remains of order 1
in one dimension, while for two--dimensional singular Fermi surfaces,
the $\bp$--dependent function $Z(\bp) \nabla \Sigma (\bp)$ vanishes 
at the Van Hove points.
In higher dimensions $d \ge 2$, and with a regular
Fermi surface fulfilling a no-nesting condition very similar to 
that required here, $\Sigma$ is continuously differentiable 
both in $q_0$ and in $\bq$. Thus it is really the Van Hove points 
on the Fermi surface that are responsible for the asymmetry. In the last section 
of this paper, we point out some possible (but as yet unproven) consequences 
of this behaviour. 

Our analysis is partly motivated by the two--dimensional Hubbard model,
a lattice fermion model with a local interaction and a dispersion relation 
$\bk \mapsto e(\bk)$ which, in suitable energy units, reads
\begin{equation}\label{hubdisp}
e(\bk) =
-\cos k_1 - \cos k_2 
+ \theta (1+ \cos k_1 \cos k_2) - \mu .
\end{equation}
The parameter $\mu $ is the chemical potential, used to adjust the particle density,
and $\theta$ is a ratio of hopping parameters.  As we shall explain now, 
the most interesting parameter range is $\mu \approx 0$ and 
$$
0 < \theta < 1 .
$$
The zeroes of the gradient of $e$ are at $(0,0)$, $(\pi,\pi)$ and at $(\pi,0)$, $(0,\pi)$.
The first two are extrema, and the last two are the saddle points relevant for van Hove
singularities (VHS). For $\mu =0$, both saddle points are on the Fermi surface. 
For $\theta = 1$ the Fermi surface degenerates to the pair of lines
$\{k_1 =0\} \cup \{ k_2 = 0\}$, so we assume that $\theta < 1$. 
For $\theta =0$ and $\mu=0$, the Fermi surface becomes the so--called Umklapp surface
$
U = \{ k: k_1 \pm k_2 = \pm\pi\},
$ 
which is nested since it has flat sides.
For $ 0 < \theta < 1$  the Fermi surface
at $\mu =0$ has nonzero curvature away from the Van Hove points $(\pi,0)$ and
$(0,\pi)$. Viewed from the point $(\pi,\pi)$, it encloses a 
strictly convex region (as a subset of $\bR^2$). 
There is ample evidence that in the Hubbard  model, 
it is the parameter range $\theta > 0$ and  electron density near to the van Hove density
($\mu \approx 0$) that is relevant  for high--$T_c$ superconductivity 
(see, e.g.\ \cite{Markiewicz,BorKor,FRS,HM,HSFR}). 
In this parameter region, an important kinematic property is that
the two saddle points at $(\pi,0)$ and $(0,\pi)$ are connected by 
the vector $Q=(\pi,\pi)$, which has the property that $2Q=0$ mod $2\pi \bZ^2$.
This modifies the leading order flow of the 
four--point function strongly (Umklapp scattering, \cite{FRS, HM,HSFR}).
The bounds we discuss here hold both in presence and absence of Umklapp 
scattering. 

The interaction of the fermions is given by $\lambda \hat v$, 
where $\lambda$ is the coupling constant and $\hat v$ is the 
Fourier transform of the two--body potential defining a density--density interaction. 
For the special case
of the Hubbard model, two fermions interact only if they are at the same 
lattice point, so that $\hat v (\bk ) =1$.
Despite the simplicity of the Hamiltonian, little is known rigorously about the 
low--temperature phase diagram of the Hubbard model, even for small $|\la|$. 
In this paper, we do perturbation theory to all orders, 
i.e.\ we treat $\la $ as a formal expansion parameter.
For a discussion of the relation of perturbation theory to all orders
to renormalization group flows obtained from truncations of the 
RG hierarchies, see the Introduction of \cite{SFS1}.

Although our analysis is motivated by the Hubbard model, 
it applies to a much more general class of models.
In this paper, we shall need only that the band function $e$ 
has enough derivatives, as stated below, 
and a similar condition on the interaction. 
In fact, the interaction is allowed to be more general than just 
a density--density interaction: it may depend on frequencies, 
as well as the spin of the particles. See \cite{FST1,FST2,FST3,FST4} for details.
As far as the singular points of $e$ are concerned, we require that they are nondegenerate. 
The precise assumptions on $e$ will be stated in detail below. 

We add a few remarks to put these assumptions into perspective.
No matter if we start with a lattice model or a 
periodic Schr\" odinger operator describing Bloch electrons
in a crystal potential, the band function given by the 
Hamiltonian for the one--body problem is, under very mild conditions, 
a smooth, even analytic function. 
In such a class of functions, the occurrence 
of degenerate critical points is nongeneric, i.e.\ measure zero. 
In other words, if 
$$
e(\bk_s +  R \bk) 
=
- \veps_1 k_1^2 + \veps_2 k_2^2 + \ldots
$$
around a Van Hove point $\bk_s$ (here $R$ is a rotation that diagonalizes the Hessian 
at $\bk_s$),  getting even one of the two prefactors $\veps _i$
to vanish in a Taylor expansion requires a fine--tuning of the hopping parameters, 
in addition to the condition that the VH points are on $S$. 
Thus, in a one--body theory, an {\em extended VHS}, where
the critical point becomes degenerate because, say, $\veps_1$ vanishes,
is nongeneric. On the other hand, experiments suggest
\cite{BorKor} that $\veps_1$ is very small in some materials, 
which seem to be modeled well by Hubbard--type band functions.
On the theoretical side, in a  renormalized expansion with counterterms,
it is not the dispersion relation of the noninteracting system, but that
of the interacting system, which appears in all fermionic covariances. 
It is thus an important theoretical question to decide what 
effects the interaction has on the dispersion relation and
in particular whether an extended VHS can be caused by the interaction. 
We shall discuss this question further in Section \ref{discussion}.

\section{Main Results} 
In this section we state our hypotheses on the dispersion function 
and the Fermi surface, and then state our main result. 

\subsection{Hypotheses on the dispersion relation}
We make the following hypotheses on the dispersion relation $e$ 
and its Fermi surface $\cF=\set{\bk}{e(\bk)=0}$ in $d=2$.

\begin{itemize}

\item[\Ha] $\set{\bk}{|e(\bk)|\le 1}$ is compact. 

\item[\Hb] $e(\bk)$ is $C^r$ with $r\ge 7$.

\item[\Hc] $e(\tilde\bk)=0$ and $\bnabla e(\tilde\bk)=\bzer$ simultaneously
      only for finitely many $\tilde\bk$'s, called {\em Van Hove points} or {\em singular points}.

\item[\Hd] If $\tilde\bk$ is a singular point then         
            $\big[\sfrac{\partial^2\hfill}{\partial\bk_i\partial\bk_j} e(\tilde\bk)
                \big]_{1\le i,j\le d}$ is nonsingular and has  
                one positive eigenvalue and one negative eigenvalue.

\item[\He] There is at worst polynomial flatness. This means the following.
               Let $\tilde\bk\in\cF$. Suppose that 
               $k_2-\tilde k_2=f(k_1-\tilde k_1)$ is a $C^{r-2}$
               curve contained in $\cF$ in a neighbourhood of $\tilde\bk$.
               (If $\tilde\bk$ is a singular point, there can be two such
               curves.)
               Then some derivative of $f(x)$ at $x=0$ of order at least 
               two and at most $r-2$ does not vanish. Similarly if the
               roles of the first and second coordinates are exchanged.

\item[\Hf] There is at worst polynomial nesting. This means the following.
               Let $\tilde\bk\in\cF$ and $\tilde\bp\in\cF$ with 
               $\tilde\bk\ne\tilde\bp$. Suppose that 
               $k_2-\tilde k_2=f(k_1-\tilde k_1)$ is a $C^{r-2}$
               curve contained in $\cF$ in a neighbourhood of $\tilde\bk$
               and $k_2-\tilde p_2=g(k_1-\tilde p_1)$ is a $C^{r-2}$
               curve contained in $\cF$ in a neighbourhood of $\tilde\bp$.
               Then some derivative of $f(x)-g(x)$ at $x=0$ of order 
               at most $r-2$ does not vanish. Similarly if the roles of the
               first and second coordinates are exchanged.
               If $e(\bk)$ is not even, we further assume a similar
               nonvanishing when $f$ gives a curve in $\cF$ in a neighbourhood
               of any $\tilde\bk\in\cF$ and $g$ gives a curve in $-\cF$ in
               a neighbourhood of any $\tilde\bp\in-\cF$.
\end{itemize}

\noindent
We denote by $n_0$ the largest nonflatness or nonnesting order plus one,
and assume that 
$$
r \ge 2n_0+1
$$
The Fermi surface for the Hubbard model with $0<\theta<1$ and $\mu=0$, 
when viewed from $(\pi,\pi)$, encloses a convex region. 
See the figure below. It has nonzero
curvature except at the singular points. If one writes the equation of 
(one branch of) the Fermi surface near the singular point $(0,\pi)$ in the form 
$k_2-\pi=f(k_1)$, then $f^{(3)}(0)\ne 0$.  So this Fermi surface
satisfies the nonflatness and no--nesting conditions with $n_0=4$.
$$
\figput{hubbardpipi}
$$

\subsection{Main theorem}

In the following, we state our main results about the fermionic self--energy.
A discussion will be given at the end of the paper, in Section \ref{discussion}.

\begin{theorem}\label{th:si}
Let $\cB = \bR^2/2\pi\bZ^2$ and 
$e \in C^7(\cB,\bR)$. Assume that the Fermi surface 
$S = \{ \bk \in \cB : e(\bk ) =0\}$ contains points where $\nabla e(\bk) =0$, 
and that the Hessian of $e$ at these points is nonsingular. 
Moreover, assume that away from these points, the Fermi surface
can have at most finite--order tangencies with its (possibly reflected)
translates and is at most polynomially flat. [These hypotheses have 
been spelled out in detail in \Ha--\Hf\  above.] 
As well, the interaction $v$ is assumed to be short--range, so that its Fourier
transform $\hat v$ is $C^2$. 

Then there is a counterterm function $K \in C^1 (\cB,\bR)$, 
given as a formal power series $K = \sum_{r\ge 1} K_r \la^r$
in the coupling constant $\la$, such that the renormalized expansion 
for all Green functions, at temperature zero, is finite to all orders in $\la$. 

\begin{enumerate}

\item\label{erschtns}
The self--energy is given as a formal power series 
$\Si = \sum_{r \ge 1} \Si_r \la^r$, 
where for all $r \in \bN$ and all $\omega \in \bR$, the function
$\bk \mapsto \Si_r (\omega ,\bk)\in C^1(\cB, \bC)$. 
Specifically, we have 
\begin{eqnarray}
\norm{\Si_r}_\infty 
&\le& 
\const
\nonumber\\
\norm{\nabla \Si_r}_\infty
&\le&
\const
\nonumber
\end{eqnarray}
with the constants depending on $r$. 
Moreover, the function 
$\omega \mapsto \Si_r (\omega ,\bk)$ is $C^1$ in $\omega$ for all 
$k \in \cB \setminus \overline{V}$, where $V $ denotes the integer lattice 
generated by all van Hove points, and the bar means the closure in $\cB$. 

\item\label{zweitns}
For $e$ given by the normal form $e(\bk) = k_1 k_2$, which has a van Hove point 
at $\bk ={\bf 0}$, the second order contribution $\Sigma_2$ to the self--energy obeys 
\begin{eqnarray*}
\mbox{ Im } \del_\om \Si_2(\om,0)
&=&
- a_1\; (\log |\om|)^2 
+ O (|\log|\om||)\\
\mbox{\rm  Re }\del_{k_1}\del_{k_2}\Si_2(\om,0)
&=&
\hphantom{-}
a_2 \; (\log |\om|)^2 + O (|\log|\om||)
\end{eqnarray*}
$\del_{k_1}^2\Si_2(\om,0)$ and $\del_{k_2}^2\Si_2(\om,0)$
grow at most linearly in $\log |\om|$.
The explicit values of $a_1> 0$ and $a_2 > 0$ are given
in Lemmas \ref{lowerbound} and \ref{xietderiv}, below. 

\end{enumerate}
\end{theorem}

\medskip\noindent
Theorem \ref{th:si} is a statement about the zero temperature limit
of $\Sigma$. That is, $\Sigma_r (\om,\bk)$ and its derivatives are computed 
at a positive temperature $T=\beta^{-1}$, where they are $C^2$ in $\omega$ 
and $\bq$, and then the limit $\beta \to \infty$ is taken. 
(Because only one--particle--irreducible graphs contribute to $\Sigma_r$,
it is indeed a regular function of $\omega$ for all $\omega \in \bR$ at any
inverse temperature $\beta < \infty$.) 

The bounds to all orders stated in item \ref{erschtns} of Theorem \ref{th:si} 
generalize to low positive temperatures in an obvious way: 
the length--of--overlap estimates and the singularity analysis 
done below
only use the spatial geometry of the Fermi surface for $e$, 
which is unaffected by the temperature. The other changes 
are merely to replace some derivatives with respect to frequency 
by finite differences, which only leads to trivial changes. 

Our explicit computation of the asymptotics in the model case of  
item \ref{zweitns} of Theorem \ref{th:si} uses that  several contributions
to these derivatives vanish in the limit $\beta \to \infty$, 
and that certain cancellations occur in the remaining terms. 
For this reason, the result stated in item \ref{zweitns} is a result at zero temperature.
(In particular, the coefficients in the $O(\log |\omega|)$ terms are just 
numbers.)
However, we do not expect any significant change in the asymptotics at low 
temperature and small $\omega$ to occur. 
That is, we expect the low--temperature asymptotics to contain only terms 
whose supremum over $|\omega| \ge \pi/\beta$ is at most of order 
$(\log \beta)^2$, and the square of the logarithm to be present. 

\subsection{Heuristic explanation of the asymmetry}
\newcommand{\pmap}{P}

We refer to the different behaviour of $\nabla \Sigma$ (which is bounded to 
all orders)
and $\del_\om \Sigma$ (which is $\log^2$--divergent in second order) 
as the asymmetry in the derivatives of $\Sigma$. 
In the case of a regular Fermi surface, no--nesting implies that $\Sigma_r$ is in $C^{1+\delta}$ 
with a H\" older exponent $\delta $ that depends on the no--nesting
assumption. A similar bound was shown in \cite{SFS1} for Fermi surfaces 
with singularities in $d \ge 3$ dimensions. In \cite{SFS1}, we formulated
a slight generalization of the no-nesting hypothesis of \cite{FST1}, 
and again proved a volume improvement estimate, which implies the 
above--mentioned H\" older continuity of the first derivatives. 

In the more special case of a regular Fermi surface with strictly positive 
curvature, we have given, in  \cite{FST2,FST3}, bounds on certain second 
derivatives  of the self--energy with respect to momentum. We briefly review 
that discussion for the second--order contribution, to motivate why there is a 
difference between the spatial and the frequency derivatives. 

For simplicity, we assume a local interaction, 
and consider the infrared part of the two--loop contribution
$$
I(q_0,\bq) = 
\big\langle 
C(\omega_1,e(\bp_1)) C(\omega_2,e(\bp_2)) 
C(\omega, \tilde e) 
\big\rangle
$$
where $\omega=q_0 + v_1 \omega_1 + v_2 \omega_2$, $v_i= \pm 1$ and 
$$
\tilde e = e( \bq+v_1 \bp_1 + v_2 \bp_2)
$$
The angular brackets denote integration of
$\bp_1$ and $\bp_2$ over the two-dimensional Brillouin zone
and Matsubara summation of $\omega_1$ and $\omega_2$, 
over the set $\frac{\pi}{\beta} (2 \bZ +1)$. By infrared part we mean 
that the fermion propagators are of the form
$$
C(\om,E) 
=
\frac{U(\om^2+E^2)}{\I \om - E}
$$
where $U $ is a suitable cutoff function that is supported in a small,
fixed neighbourhood of zero.

The third denominator depends on the external momentum $(q_0,\bq)$ 
and derivatives with respect to the external momentum 
increase the power of that  denominator, which
may lead to bad behaviour as $\beta\rightarrow\infty$. 
The main idea why some derivatives behave better than expected
by simple counting of powers (see \cite{FST2}) is that in dimension two and 
higher, there are, in principle, enough integrations to make a change
of variables so that $\tilde e$, $e(\bp_1)$ and $e(\bp_2)$ all become
integration variables. This puts all dependence on the 
external variable $\bq$ into the Jacobian $J$ of this change 
of variables. If $J$ were $C^k$ with uniform bounds, 
$I(q_0,\bq)$ would be $C^k$ in the spatial momentum $\bq$, and
(by integration by parts) also in $q_0$. 
However, $J$ always has singularities, and the leading contributions to 
the derivatives of  $I(q_0,\bq)$ come from the vicinity of these singularities. 
It was proven in \cite{FST2,FST3} that if the Fermi surface is regular
and has strictly positive curvature, these
singularities of the Jacobian are harmless, provided derivatives are
taken tangential to the Fermi surface. 

To explain the change of variables, we first show it for the case
without van Hove singularities and then discuss the changes required
when van Hove singularities are present. 

In a neighborhood of the Fermi surface, we introduce coordinates 
$\rho $ and $\theta$, so that $\bp = \pmap (\rho,\theta)$. The
coordinates are chosen such that $\rho = e(\bp)$ and that 
$\pmap (\rho,\theta+\pi)$ is the antipode of $\pmap (\rho,\theta)$. 
(We are assuming that the Fermi surface is strictly convex -- see \cite{FST2}.) 
Doing this for $\bp_1$ and $\bp_2$, with corresponding Jacobian 
$J = \det \pmap'$, we have 
$$
I(q_0,\bq) = \left\lbrack
\int d \theta_1\; d\theta_2\; 
J(\rho_1,\theta_1) J(\rho_2,\theta_2)\;
C(\omega, \tilde e)
\right\rbrack_{1,2}
$$
where $\lbrack F \rbrack_{1,2}$ now denotes multiplying $F$ by 
$C(\omega_1,\rho_1) C(\omega_2,\rho_2)$ and integrating
over $\rho_1$ and $\rho_2$ and summing over the frequencies. 
To remove the $\bq$--dependence from $C(\om,\tilde e)$, 
one now wants to change variables from $\theta_1$ or $\theta_2$ to $\tilde e$. 
This works except near points where 
$$
\frac{\partial \tilde e}{\partial \theta_1}
=
\frac{\partial \tilde e}{\partial \theta_2}
= 0.
$$
These equations determine the singularities of the Jacobian.
The detailed analysis of their solutions is in \cite{FST2}. 
Essentially, if one requires that the momenta $\bp_1$ and $\bp_2$ are
on the FS, i.e.\ $\rho_1= \rho_2=0$, that $\bq=\pmap (0,\theta_\bq)$ is on 
the FS {\em and} that the sum 
$\bq+v_1 \pmap(0,\theta_1) + v_2 \pmap(0,\theta_2)$ is on the FS,
then the only solutions are 
$\theta_1 \in \{\theta_\bq,\theta_\bq+\pi\}$ and 
$\theta_2 \in \{\theta_\bq,\theta_\bq+\pi\}$.
(The general case of momenta near to the Fermi surface is then 
treated by a deformation argument which requires that 
$\nabla e \ne 0$ and that the curvature be nonzero.)
A detailed analysis of the singularity  in $J$, in which strict convexity
enters again, then implies that the self--energy is regular.

The conditions needed for the above argument fail at the Van Hove points.
But, introducing a partition of unity on the Fermi surface, they still 
hold away from the singular points. So the only contributions that may fail to 
have derivatives come from  $\bq+v_1 \bp_1 + v_2 \bp_2$ 
in a small neighbourhood of the singular points. 
When a derivative with respect to $q_0$ is taken, the integrand contains
a factor of $-\I (\I \om - \tilde e)^{-2}$.
When a derivative with respect to $\bq$ is taken, the integrand contains 
a factor $\nabla e ( \bq+v_1 \bp_1 + v_2 \bp_2) (\I \om - \tilde e)^{-2}$. 
Because we are in a small neighbourhood of the singular point, 
the numerator, $\nabla e ( \bq+v_1 \bp_1 + v_2 \bp_2)$, in the latter expression is small, and vanishes at the 
singular point. This suggests  that the first derivative with respect to $\bq$ 
may be better behaved than the first derivative with respect to 
$q_0$, as is indeed the case -- see item \ref{erschtns} 
of Theorem \ref{th:si}.  A second derivative with respect to
$\bq$ may act on the numerator, $\nabla e ( \bq+v_1 \bp_1 + v_2 \bp_2)$,
and eliminate its zero. This suggests that the second derivative with respect 
to $\bq$ behaves like the first derivative with respect to $q_0$, 
as is indeed the case -- see item \ref{zweitns} of Theorem \ref{th:si}.

The above heuristic  discussion is only to provide a motivation as to why
the asymmetry in the derivatives of $\Sigma_2$ occurs. The proof does
not make use of the idea of the change of variables to $\tilde e$, 
but rather of length--of--overlap estimates, which partially replace
the overlapping loop estimates, away from the singular points. 
This allows us to show the convergence of the first $\bq$--derivative under
conditions \Ha--\Hf, which are significantly weaker than
strict convexity, and it also allows us to treat the 
situation with Umklapp scattering, which had to be excluded in second
order in \cite{FST2}, and which is the reason for the restriction on the 
density in \cite{FST2}.

\section{Fermi Surface} 

In this section we prove bounds on the size
of the overlap of the Fermi surface with translates of a tubular neighbourhood
of the Fermi surface. These bounds make precise the geometrical 
idea that for non--nested surfaces (here: curves), the non--flatness
condition \He\ strongly restricts such lengths of overlap. 

\subsection{Normal form for $e(\bk)$ near a singular point}\label{ssecnormform}

\begin{lemma} \label{le:normalform}
Let $d=2$ and assume \Hb--\He\ with $r\ge n_0+1$. Assume that $\tilde \bk = \bzer$ is a singular
point of $e$. 
Then there are
\begin{itemize}
\item[$\circ$] integers $2\le \nu_1,\nu_2 < n_0$,
\item[$\circ$] a constant, nonsingular matrix $A$ and 
\item[$\circ$] $\!C^{r-2-\max\{\nu_1,\nu_2\}}$ functions $a(\bk)$, 
$b(\bk)$ and $c(\bk)$ that are bounded and bounded away from zero
\end{itemize}
such that in a neighbourhood of the origin
\begin{equation}\label{normalform1}
e(A\bk)=a(\bk)\big(k_1-k_2^{\nu_1}b(\bk)\big)\big(k_2-k_1^{\nu_2} c(\bk)\big) .
\end{equation}
\end{lemma}

\begin{proof}
Let $\la_1$ and $\la_2$ be the eigenvalues of 
$\big[\sfrac{\partial^2\hfill}{\partial\bk_i\partial\bk_j} e(\bzer)
                \big]_{1\le i,j\le 2}$ and set $\tilde\la=\sqrt{|\la_2/\la_1|}$.
By the Morse lemma \cite[Lemma 1.1 of Chapter 6]{MorselemmaH}, 
there is a $C^{r-2}$ diffeomorphism $\bx(\bk)$ 
with $\bx(\bzer)=\bzer$
such that
\begin{eqnarray}
e(\bk)
&=&
\la_1 x_1(\bk)^2+\la_2 x_2(\bk)^2
\nonumber\\
&=&
\la_1\big(\,x_1(\bk)^2-\tilde\la^2 x_2(\bk)^2\,\big)
\nonumber\\
&=&
\la_1\big(\,x_1(\bk)-\tilde\la x_2(\bk)\,\big)
                       \ \big(\,x_1(\bk)+\tilde\la x_2(\bk)\,\big)
\nonumber\\
&=&
\la_1\big(\,a_1k_1+a_2k_2-\tilde x_1(\bk)\,\big)\ 
                       \big(\,b_1k_1+b_2k_2-\tilde x_2(\bk)\,\big)
\nonumber
\end{eqnarray}
with $\tilde x_1(\bk)$ and $\tilde x_2(\bk)$ vanishing to at least order two 
at $\bk=0$. Here $a_1k_1+a_2k_2$ and $b_1k_1+b_2k_2$ are the degree one
parts of the Taylor expansions of $x_1(\bk)-\tilde\la x_2(\bk)$
and $x_1(\bk)+\tilde\la x_2(\bk)$ respectively.
Since the Jacobian $\det D\bx(\bzer)$ of the diffeomorphism at the
origin is nonzero and  
$\det\left[\begin{array}{rr}1&-\tilde\la\\ 1&\tilde\la\end{array}\right]=2\tilde\la\ne 0$, we have
$$
\det\left[\begin{array}{rr}a_1&a_2\\ b_1&b_2\\\end{array}\right]
=\det\left\{\left[\begin{array}{rr}1&-\tilde\la\\ 1&\tilde\la\end{array}\right] 
D\bx(\bzer)\right\}
\ne 0
$$
Setting
$$
A=\left[\begin{array}{rr}a_1&a_2\\ b_1&b_2\end{array}\right]^{-1}
$$
we have
$$
e(A\bk)=\la_1\big(k_1-\tilde x_1(A\bk)\big)
                       \big(k_2-\tilde x_2(A\bk)\big)
$$
Write
$$
k_1-\tilde x_1(A\bk)
=k_1-k_1f_1(\bk)-k_2^{\nu_1}g_1(k_2)
$$
with
\begin{eqnarray}
k_1f_1(\bk)
&=&
\tilde x_1(A\bk)-\tilde x_1(A\bk)\big|_{k_1=0}
\nonumber\\
k_2^{\nu_1}g_1(k_2)
&=&
\tilde x_1(A\bk)\big|_{k_1=0}
\nonumber
\end{eqnarray}
Since $\tilde x_1(A\bk)$ vanishes to order at least two  at  $\bk=0$,
$$
f_1(\bk)=\int_0^1
\big[\sfrac{\partial\hfill}{\partial k_1}\tilde x_1(A\bk)\big]_{\bk=(tk_1,k_2)}
\ dt
$$
is $C^{r-3}$ and vanishes to order at least one at $\bk=0$ 
and, in particular, $|f_1(\bk)|\le\half $ for all $\bk$ in 
a neighbourhood of the origin. We choose $\nu_1$ to be the power of the
first nonvanishing term in the Taylor expansion of 
$\tilde x_1(A\bk)\big|_{k_1=0}$. Since this function must vanish to order
at least two in $k_2$, we have that $\nu_1\ge 2$. By the nonflatness condition,
applied to the curve implicitly determined by $k_1=\tilde x_1(A\bk)$,
$\nu_1\le r-2$.
So $g_1(k_2)$ is $C^{r-2-\nu_1}$ and is bounded and bounded away from zero in 
a neighbourhood of $k_2=0$.

In a similar fashion, write
$$
k_2-\tilde x_2(A\bk)
=k_2-k_2f_2(\bk)-k_1^{\nu_2}g_2(k_1)
$$
with
\begin{eqnarray}
k_2f_2(\bk)
&=&
\tilde x_2(A\bk)-\tilde x_2(A\bk)\big|_{k_2=0}
\nonumber\\
k_1^{\nu_2}g_2(k_1)
&=&
\tilde x_2(A\bk)\big|_{k_2=0}
\nonumber
\end{eqnarray}
Then we have the desired decomposition \Ref{normalform1} with
\begin{eqnarray}
a(\bk)
&=&
\la_1\big(1-f_1(\bk)\big)\big(1-f_2(\bk)\big)\nonumber\\
b(\bk)
&=&
\frac{g_1(k_2)}{1-f_1(\bk)}\; ,
\qquad
c(\bk)
=
\frac{g_2(k_1)}{1-f_2(\bk)} .
\nonumber
\end{eqnarray}
\end{proof}\goodbreak
We remark that it is possible to impose weaker regularity hypotheses
by exploiting that $k_2^{\nu_1}b(\bk)$, resp. $k_1^{\nu_2}c(\bk)$,
is a $C^{r-3}$ function whose $k_2$, resp. $k_1$, derivatives of order 
strictly less than $\nu_1$, resp. $\nu_2$,  vanish at $k_2=0$, resp. $k_1=0$.

\subsection{Length of overlap estimates}
It follows from the normal form derived in Lemma \ref{le:normalform} that 
under the hypotheses \Hb--\He\ 
the curvature of the Fermi surface may vanish as one approaches 
the singular points. Thus, even if the Fermi surface is curved 
away from these points, there is no uniform lower bound
on the curvature. Curvature effects are very important in the 
analysis of regularity estimates, and in a situation without 
uniform bounds these curvature effects improve power counting
only at scales lower than a scale set by the rate at which the curvature
vanishes. 
Thus it becomes natural to define, at a given scale, 
scale--dependent neighbourhoods of the singular points,
outside of which curvature improvements hold. 
The estimates for the length of overlaps that we prove in this section
allow us to make this idea precise. They hold under much more 
general conditions than a nonvanishing curvature, namely 
the nonnesting and nonflatness assumptions \He\ and \Hf\
suffice. We first discuss the special case corresponding 
to the normal form in the vicinity of a singular point, and then
deal with the general case. 

\subsubsection{Length of overlap -- special case}
\begin{lemma}\label{le:4.2}
Let $\nu_1\ge 2$ and $\nu_2\ge 2$ be integers and
$$
e(x,y)= \big(x-y^{\nu_1}b(x,y)\big)\big(y-x^{\nu_2}c(x,y)\big)
$$
with $b$ and $c$ bounded and bounded away from zero and with 
$b, c\in C^{\nu_2+1}$. Let $u(x)$ obey
$$
u(x)=x^{\nu_2} c\big(x,u(x)\big)
$$
for all $x$ in a neighbourhood of $0$. That is, $y=u(x)$ lies on the
Fermi curve $e(x,y)=0$. There are constants $C$ and $D>0$ such that
for all $\veps>0$ and $0<\de\le |(X,Y)|\le D $
$$
\vol\set{x\in\bR}{|x|\le D,\ \big|e\big(X+x, Y+u(x)\big)\big|\le\veps}
\le C\big(\sfrac{\veps}{\de}\big)^{1/\nu_2}
$$
\end{lemma}
\begin{proof}
Write
\begin{equation}\label{profo}
e\big(X+x, Y+u(x)\big)=F(x,X,Y)G(x,X,Y)
\end{equation}
with 
\begin{eqnarray}
F(x,X,Y)
&=&
X+x-(Y+u(x))^{\nu_1}b\big(X+x,Y+u(x)\big)
\nonumber\\
G(x,X,Y)
&=&
Y+u(x)-(X+x)^{\nu_2}c\big(X+x,Y+u(x)\big)
\nonumber\\
&=&
Y-\big\{(X+x)^{\nu_2}-x^{\nu_2}\big\}c\big(X+x,Y+u(x)\big)
\nonumber\\
&&
\quad - x^{\nu_2}\big\{c\big(X+x,Y+u(x)\big)-c\big(x,u(x)\big)\big\}
\nonumber
\end{eqnarray}
Observe that, for all allowed $x$, $X$ and $Y$,
$$
|F(x,X,Y)|\le\sfrac{1}{100}\qquad
\big|\sfrac{\partial\hfill}{\partial x}F(x,X,Y)\big|\ge\sfrac{99}{100}
$$
since $x$, $X$, $Y$ and $u(x)$ all have to be $O(D)$ small.
For our analysis of $G(x,X,Y)$ we consider two separate cases.

\noindent{\it Case 1: $|Y|\ge\ka |X|$} with $\ka$ a constant to be chosen
shortly. Since $c\big(X+x,Y+u(x)\big)-c\big(x,u(x)\big)$ vanishes to first
order in $(X,Y)$, for all $x$
\begin{eqnarray}
\big|x^{\nu_2}\big\{c\big(X+x,Y+u(x)\big)-c\big(x,u(x)\big)\big\}\big|
&\le&
\sfrac{1}{100}\big[|X|+|Y|\big]
\nonumber\\
\big|\sfrac{\partial\hfill}{\partial x}
\big[x^{\nu_2}\big\{c\big(X+x,Y+u(x)\big)-c\big(x,u(x)\big)\big\}\big]\big|
&\le&
\sfrac{1}{100}\big[|X|+|Y|\big]
\nonumber
\end{eqnarray}
Since $(X+x)^{\nu_2}-x^{\nu_2}$ vanishes to first
order in $X$, for all $x$
\begin{eqnarray}
\big|\big\{(X+x)^{\nu_2}-x^{\nu_2}\big\}c\big(X+x,Y+u(x)\big)\big|
&\le&
\tilde\ka |X|
\nonumber\\
\big|\sfrac{\partial\hfill}{\partial x}
\big[\big\{(X+x)^{\nu_2}-x^{\nu_2}\big\}c\big(X+x,Y+u(x)\big)\big]\big|
&\le&
\tilde\ka |X|
\nonumber
\end{eqnarray}
We choose $\ka=\max\{2,200\tilde\ka\}$. Then
$$
|G(x,X,Y)|\ge\sfrac{98}{100}|Y|\qquad
\big|\sfrac{\partial\hfill}{\partial x}G(x,X,Y)\big|\le\sfrac{2}{100}|Y|
$$
Thus, by \Ref{profo} and the product rule,
\begin{eqnarray}
\big|\sfrac{\partial\hfill}{\partial x}e\big(X+x, Y+u(x)\big)\big|
&\ge&
 \big(\sfrac{99}{100}\sfrac{98}{100}-\sfrac{1}{100}\sfrac{2}{100}\big)|Y|
\ge
\sfrac12 |Y|
\nonumber
\end{eqnarray}
and, by  Lemma \ref{le:interval},
$$
\vol\set{x\in\bR}{|x|\le D,\ \big|e\big(X+x, Y+u(x)\big)\big|\le\veps}
\le 4\sfrac{\veps}{|Y|/2}
\le 16\sfrac{\veps}{\de} .
$$

\noindent{\it Case 2: $|Y|\le\ka |X|$.}
In this case we bound the $\nu_2^{\rm th}$ $x$--derivative 
away from zero. We claim that the dominant term comes from one derivative
acting on $F$ and $\nu_2-1$ derivatives acting on $G$. Observe that for
$|X|,\ |Y|,\ |x|\le D$ with $D$ sufficiently small
\begin{eqnarray}
\abs{\sfrac{d^m\hfill}{d x^m}u(x)}
&\le&
 O(D)
\qquad\hbox{ for }0\le m<\nu_2
\nonumber
\end{eqnarray}
since $u(x)=x^{\nu_2}c\big(x,u(x)\big)$ and consequently
\begin{eqnarray}
|F(x,X,Y)|
&\le&
 O(D)
\nonumber\\
\big|\sfrac{\partial\hfill}{\partial x}F(x,X,Y)\big|
&\ge&
 1-O(D)
\nonumber\\
\big|\sfrac{\partial^m\hfill}{\partial x^m}F(x,X,Y)\big|
&\le&
 O(D)
\qquad\hbox{ for }1< m\le\nu_2
\nonumber
\end{eqnarray}
Furthermore,
since $c\big(X+x,Y+u(x)\big)-c\big(x,u(x)\big)$ vanishes to first
order in $(X,Y)$, for all $x$,
$$
\big|\sfrac{\partial^m\hfill}{\partial x^m}
\big[x^{\nu_2}\big\{c\big(X+x,Y+u(x)\big)-c\big(x,u(x)\big)\big\}\big]\big|
\le
 O(D)\big[|X|+|Y|\big]
$$
for $0\le  m<\nu_2$ and
$$
\big|\sfrac{\partial^m\hfill}{\partial x^m}
\big[x^{\nu_2}\big\{c\big(X+x,Y+u(x)\big)-c\big(x,u(x)\big)\big\}\big]\big|
\le
 O(1)\big[|X|+|Y|\big]
$$
for $m=\nu_2$.
Since $(X+x)^{\nu_2}-\nu_2Xx^{\nu_2-1}-x^{\nu_2}$ vanishes to second
order in $X$, for all $x$,
\begin{eqnarray}
\big|\sfrac{\partial^m\hfill}{\partial x^m}
\big[\big\{(X+x)^{\nu_2}-\nu_2Xx^{\nu_2-1}-x^{\nu_2}\big\}
c\big(X+x,Y+u(x)\big)\big]\big|
&\le& O(|X|^2) \nonumber\\
&\le& O(D)|X|
\nonumber
\end{eqnarray}
for all $0\le m\le \nu_2$. Finally
\begin{eqnarray}
\big|\sfrac{\partial^m\hfill}{\partial x^m}
\big[\nu_2Xx^{\nu_2-1}c\big(X+x,Y+u(x)\big)\big]\big|
&\le&
 O(D)|X|
\qquad\hbox{ for }0\le  m<\nu_2-1
\nonumber\\
\big|\sfrac{\partial^m\hfill}{\partial x^m}
\big[\nu_2Xx^{\nu_2-1}c\big(X+x,Y+u(x)\big)\big]\big|
&\ge&
 \ka'|X|-O(D)|X|
\quad\hbox{ for }m=\nu_2-1
\nonumber\\
\big|\sfrac{\partial^m\hfill}{\partial x^m}
\big[\nu_2Xx^{\nu_2-1}c\big(X+x,Y+u(x)\big)\big]\big|
&\le&
 O(1)|X|
\qquad\hbox{ for }m=\nu_2
\nonumber
\end{eqnarray}
with $\ka'=\nu_2!\inf |c(x,y)|>0$.
Consequently,
\begin{eqnarray}
\big|\sfrac{\partial^{\nu_2}\hfill}{\partial x^{\nu_2}}
&& \hskip-24pt
e\big(X+x, Y+u(x)\big)\big|
\nonumber\\
&=&
\big|\nu_2\sfrac{\partial F}{\partial x}
      \sfrac{\partial^{\nu_2-1} G}{\partial x^{\nu_2-1}}
+\sum_{m=0\atop m\ne 1}^{\nu_2}
{\tst{\nu_2 \choose m}} \;
\sfrac{\partial^m F}{\partial x^m}
      \sfrac{\partial^{\nu_2-m} G}{\partial x^{\nu_2-m}}\big|
\nonumber\\
&\ge&
  \big(1-O(D)\big)\big(\ka'-O(D)\big)|X|-O(D)\big[|X|+|Y|\big]
\nonumber\\
&\ge&
  \big(1-O(D)\big)\big(\ka'-O(D)\big)|X|-O(D)(1+\ka)|X|
\nonumber\\
&\ge&
\sfrac{\ka'}{2}|X|
\nonumber
\end{eqnarray}
if $D$ is small enough. Hence, by Lemma \ref{le:interval},
\begin{eqnarray}
&&\vol\set{x\in\bR}{|x|\le D,\ \big|e\big(X+x, Y+u(x)\big)\big|\le\veps}
\nonumber\\
&&\hskip70pt\le
 2^{\nu_2+1}\big(\sfrac{\veps}{\ka'|X|/2}\big)^{1/\nu_2}
\le
 2^{\nu_2+1}
   \big(\sfrac{2\sqrt{1+\ka^2}}{\ka'}\sfrac{\veps}{\de}\big)^{1/\nu_2}
\nonumber
\end{eqnarray}
\end{proof}

\subsubsection{Length of overlap -- general case}\label{lenovgen}

\begin{proposition}\label{le:length2}
Assume  \Ha--\Hf\ with $r\ge 2n_0+1$. There is a constant $D>0$ 
such that for all $0<\de<1$ and each sign $\pm$ the measure of the set of $\bp\in\bR^2$ such that
$$
\ell\Big(\set{\bk\in\cF}{\big|e\big(\bp\pm\bk\big)\big|\le M^j}\Big)
\ge \big(\sfrac{M^j}{\de}\big)^{1/n_0}\qquad\hbox{for some }j<0
$$
is at most $D\de^2$. Here $\ell$ is the Euclidean measure (length) on $\cF$.
Recall that $n_0$ is the largest nonflatness or nonnesting order plus one. 
\end{proposition}

\begin{lemma}\label{le:curve2}
Let $r\ge 2n_0+1$. 
For each $\tilde\bp\in\bR^2$ and $\tilde \bk\in\cF$, there are
constants $d,D'>0$ (possibly depending on $\tilde\bp$ and $\tilde\bk$)
such that for each sign $\pm$, all $j<0$, and all $\bp\in\bR^2$ obeying 
$|\bp-\tilde\bp|\le d$ 
$$
\ell\Big(\set{\bk\in\cF}{|\bk-\tilde\bk|\le d,\ \big|e\big(\bp\pm\bk\big)\big|
      \le M^j}\Big)
\le D'\big(\sfrac{M^j}{|\bp-\tilde\bp|}\big)^{1/n_0}
$$
\end{lemma}

\noindent
{\em Proof of  Proposition \ref{le:length2}, assuming Lemma \ref{le:curve2}. }
For each $\tilde\bp\in\bR^2$ and $\tilde \bk\in\cF$, 
let $d_{\tilde\bp,\tilde\bk},D'_{\tilde\bp,\tilde\bk}$ be
the constants of the Lemma and set
$$
\cO_{\tilde\bp,\tilde\bk}=\set{(\bp,\bk)\in\bR^2\times \cF}{
           |\bp-\tilde\bp|< d_{\tilde\bp,\tilde\bk},\ 
           |\bk-\tilde\bk|< d_{\tilde\bp,\tilde\bk}}
$$
Since $\cF=\set{\bk}{e(\bk)=0}$ and $\set{\bk}{|e(\bk)|\le 1}$ are compact,
there is an $R>0$ such that if $|\bp|>R$, then
$\set{\bk\in\cF}{\big|e\big(\bp\pm\bk\big)\big| \le M^j}$ is empty for all $j<0$.
Since $\set{\bp\in\bR^2}{|\bp|\le R}\times\cF$ is compact, there are
$(\tilde\bp_1,\tilde\bk_1)$, $\ldots$, $(\tilde\bp_N,\tilde\bk_N)$ such that
$$
\set{\bp\in\bR^2}{|\bp|\le R}\times\cF\subset
\bigcup_{i=1}^N \cO_{\tilde\bp_i,\tilde\bk_i}
$$
Fix any $0<\de<1$ and set, for each $1\le i\le N$, 
$\de_i=(N D'_{\tilde\bp_i,\tilde\bk_i})^{n_0}\de$.
If $|\bp-\tilde\bp_i|> \de_i$ for all $1\le i\le N$, then
for all $j<0$
\begin{eqnarray}
&&\hskip-24pt
\ell\big(\set{\bk\in\cF}{|e(\bp\pm\bk)|\le M^j}\big)
\nonumber\\
&\le&
\ell\Big(\bigcup_{1\le i\le n
            \atop |\bp-\tilde\bp_i|\le d_{\tilde\bp_i,\tilde\bk_i} }
 \set{\bk\in\cF}{
   |\bk-\tilde\bk_{\tilde\bp_i,\tilde\bk_i}|\le d_{\tilde\bp_i,\tilde\bk_i},
   \ \big|e\big(\bp\pm\bk\big)\big| \le M^j}\Big)
\nonumber\\
&\le&
\sum_{i=1}^N D'_{\tilde\bp_i,\tilde\bk_i}
\Big(\sfrac{M^j}{|\bp-\tilde\bp_i|}\Big)^{1/n_0}
<\sum_{i=1}^N D'_{\tilde\bp_i,\tilde\bk_i}
\Big(\sfrac{M^j}{\de_i}\Big)^{1/n_0}
\nonumber\\
&\le&\sum_{i=1}^N \sfrac{1}{N}
\Big(\sfrac{M^j}{\de}\Big)^{1/n_0}
=\Big(\sfrac{M^j}{\de}\Big)^{1/n_0}
\nonumber
\end{eqnarray}
Consequently the measure of the set of $\bp\in\bR^2$ for which
$$
\ell\Big(\set{\bk\in\cF}{\big|e\big(\bp\pm\bk\big)\big|\le M^j}\Big)
\ge \big(\sfrac{M^j}{\de}\big)^{1/n_0}\qquad\hbox{for some }j<0
$$
is at most 
$$
\sum_{i=1}^N \pi\de_i^2
\le D\de^2
\qquad\hbox{where}\quad
D=\sum_{i=1}^N \pi(N D'_{\tilde\bp_i,\tilde\bk_i})^{2n_0} .
$$
\proofendsign

\noindent
{\em Proof of Lemma \ref{le:curve2}. }
We give the proof for $\bp+\bk$. The other case is similar.
In the event that $\tilde\bp+\tilde\bk\notin\cF$, there is a $d>0$ and an integer $j_0<0$ such 
that $\set{\bk\in\cF}{|\bk-\tilde\bk|\le d,\ \big|e\big(\bp+\bk\big)\big|\le M^j}$ 
is empty for all $|\bp-\tilde\bp|\le d$ and $j<j_0$. So we may assume that
$\tilde\bp+\tilde\bk\in\cF$. 

\noindent
{\it Case 1: $\tilde\bp+\tilde\bk$ is not a singular point.\ \ \ }
By a rotation and translation of the $\bk$ plane, we may assume that
$\tilde\bp+\tilde\bk=\bzer$ and that the tangent line to $\cF$ at $\tilde\bp+\tilde\bk$
is $k_2=0$. Then, as in Section \ref{ssecnormform},  there are
\begin{itemize}

\item[$\circ$] $\nu\in\bN$ with $2\le\nu < n_0$ and 

\item[$\circ$] $C^{r-\nu}$ functions $a(\bq)$ and $b(\bq)$ that are bounded
and bounded away from zero 
\end{itemize}
such that
$$
e(\bq)= a(\bq)\big(q_2-q_1^\nu b(\bq)\big)
$$
in a neighbourhood of $\bzer$. (Choose $q_2a(\bq)=e(\bq)-e(\bq)\big|_{q_2=0}$
and $q_1^\nu a(\bq)b(\bq)= e(\bq)\big|_{k_q=0}$.) If the tangent line
to $\cF$ at $\tilde\bk$ (when $\tilde\bk$ is a singular point the tangent
line of one branch) is not parallel to $k_1=0$ and if $d$ is small enough, 
we can write the equation
of $\cF$ for $\bk$ within a distance $d$ of $\tilde \bk$ 
(when $\tilde\bk$ is a singular point, the equation of the branch under 
consideration) as $k_2-\tilde k_2=(k_1-\tilde k_1)^{\nu'}c(k_1-\tilde k_1)$ 
for some  $1\le \nu'< n_0$ and some $C^{r-n_0-1}$ function $c$ that 
is bounded and bounded away from zero (when $\tilde\bk$ is a singular point, by
Lemma \ref{le:normalform}).
Then, for $\bk\in\cF$, we have, writing $\bp-\tilde\bp=(X,Y)$ and 
$k_1-\tilde k_1=x$,
\begin{eqnarray}
e(\bp+\bk)
&=&
e(\bp-\tilde\bp+\bk-\tilde\bk)
=e\big((X,Y)+(x,x^{\nu'}c(x))\big)
\nonumber\\
&=&
A(x,X,Y)\big(Y+x^{\nu'}c(x)-(X+x)^\nu B(x,X,Y)
\big)
\nonumber
\end{eqnarray}
where
\begin{eqnarray}
A(x,X,Y)
&=&
a\big(X+x,Y+x^{\nu'}c(x)\big)
\nonumber\\
B(x,X,Y)
&=&
 b\big(X+x,Y+x^{\nu'}c(x)\big)
\nonumber
\end{eqnarray}
and $c(x)$ are bounded and bounded away from zero. 

Observe that $y=x^{\nu'} c(x)$ is the equation of a fragment of $\cF$ 
translated so as to move $\tilde\bk$ to $\bzer$ and $y=x^\nu b(x,y)$ is 
the equation of a fragment of $\cF$ translated so as to 
move $\tilde\bp+\tilde\bk$ to $\bzer$. If $\tilde\bp=\bzer$ these two fragments 
may be identical. That is $x^{\nu'} c(x)\equiv x^\nu b\big(x,x^{\nu'}c(x)\big)$.
(Of course, in this case $\nu=\nu'$.) If $\tilde\bp\ne\bzer$, the nonnesting
condition says that there is an $n\in\bN$ such that if $y=x^\nu b(x,y)$
is rewritten in the form $y=x^\nu C(x)$, then the $n^{\rm th}$ derivative
of  $x^\nu C(x)-x^{\nu'} c(x)$ must not vanish at $x=0$. Let $n < n_0$ be the
smallest such natural number. Since derivatives  of $x^\nu C(x)$ at $x=0$ 
of order strictly lower than $n$ agree with the corresponding derivatives of 
$x^{\nu'} c(x)$, the $n^{\rm th}$ derivatives at $x=0$ of 
$x^\nu b\big(x,x^\nu C(x)\big)$  and $x^\nu b\big(x,x^{\nu'}c(x)\big)$ coincide.
Since $x^\nu C(x)\equiv x^\nu b\big(x,x^\nu C(x)\big)$, the $n^{\rm th}$ 
derivative of  $x^{\nu'} c(x)-x^\nu b\big(x,x^{\nu'} c(x)\big)
    =x^{\nu'}c(x)-x^\nu B(x,0,0)$ must not vanish at $x=0$.

\begin{itemize}

\item[$\circ$] If $\nu'<\nu$,
$$
\sfrac{d^{\nu'}\hfill}{dx^{\nu'}}\big(Y+x^{\nu'}c(x)-(X+x)^\nu B(x,X,Y)\big)
=\nu'!\,c(x)+O(d)=\nu'!\,c(0)+O(d)
$$
is uniformly bounded away from zero, if $d$ is small enough.

\item[$\circ$] If $\nu'>\nu$,
$$
\sfrac{d^{\nu}\hfill}{dx^{\nu}}\big(Y+x^{\nu'}c(x)-(X+x)^\nu B(x,X,Y)\big)
=-\nu!\,B(0,0,0)+O(d)
$$
is uniformly bounded away from zero, if $d$ is small enough.

\item[$\circ$] If $\nu'=\nu$ and $x^{\nu'}c(x)\not\equiv x^\nu B(x,0,0)$,
then, as the function $x^\nu B(x,0,0)-(X+x)^\nu B(x,X,Y)$ vanishes for all 
$x$ if $X=Y=0$, 
\begin{eqnarray}
&&\hskip-0.5in \sfrac{d^n\hfill}{dx^n}
        \Big[Y+x^{\nu'}c(x)-(X+x)^\nu B(x,X,Y)\Big]
\nonumber\\
&=&\sfrac{d^n\hfill}{dx^n}\Big[Y+x^{\nu'}c(x)-x^\nu B(x,0,0)
\nonumber\\
&&\hskip0.5in+
\big[x^\nu B(x,0,0)-(X+x)^\nu B(x,X,Y)\big]\Big]
\nonumber\\
&=&
\sfrac{d^n\hfill}{dx^n}\big(x^{\nu'}c(x)-x^\nu B(x,0,0)\big) +O(|X|+|Y|)
\nonumber\\
&=&
\sfrac{d^n\hfill}{dx^n}\big(x^{\nu'}c(x)-x^\nu B(x,0,0)\big)\big|_{x=0}
+O(d) +O(|X|+|Y|)
\nonumber
\end{eqnarray}
is uniformly bounded away from zero, if $d$ is small enough.

\item[$\circ$] If $\nu'=\nu$ and $x^{\nu'}c(x)\equiv x^\nu B(x,0,0)$
and $|Y|\le |X|$
\begin{eqnarray}
Y&+&
x^{\nu'}c(x)-(X+x)^\nu B(x,X,Y)
\nonumber\\
&=&Y-\nu Xx^{\nu-1}B(x,X,Y)
\nonumber\\
&&
-\{(X+x)^\nu-\nu Xx^{\nu-1}- x^\nu\} B(x,X,Y)
\nonumber\\
&&
-x^\nu \{B(x,X,Y)-B(x,0,0)\}
\nonumber
\end{eqnarray}
so that
\begin{eqnarray}
&&\hskip-24pt
\sfrac{d^{\nu-1}\hfill}{dx^{\nu-1}}\big(Y+x^{\nu'}c(x)-(X+x)^\nu B(x,X,Y)\big)
\nonumber\\
&=&
-\nu!\,X\big[B(0,0,0)+O(d)\big]+O(d)O\big(|X|+|Y|\big)
\nonumber
\end{eqnarray}
is bounded away from zero by $\half \nu!\,\big|X B(0,0,0)\big|$, 
if $d$ is small enough.

\end{itemize}
In all of the above cases, by Lemma \ref{le:interval},
$$
\ell\Big(\set{\bk\in\cF}{|\bk-\tilde\bk|\le d,\ \big|e\big(\bp+\bk\big)\big|
      \le M^j}\Big)
\le \sqrt{1+c_1^2}\ 2^{\nu_0+1}\big(\sfrac{c_0M^j}{\rho}\big)^{1/\nu_0}
$$
where $\nu_0$ is one of $\nu'$, $\nu$, $n$ or $\nu-1$, the constant $c_0$
is the inverse of the infimum of $a(\bk)$, the constant $c_1$ is the maximum
slope of $\cF$ within a distance $d$ of $\tilde\bk$  
and $\rho$ is either a constant
or a constant times $X$ with $X$ at least a constant times $|\bp-\tilde\bp|$.
There are two remaining possibilities. One is that the tangent line
to $\cF$ at $\tilde\bk$ (when $\tilde\bk$ is a singular point the tangent
line of one branch) is parallel to $k_1=0$. This case is easy to handle
because the two fragments of $\cF$ are almost perpendicular, so that
$$
\ell\Big(\set{\bk\in\cF}{|\bk-\tilde\bk|\le d,\ \big|e\big(\bp+\bk\big)\big|
      \le M^j}\Big)\le\const M^j
$$
The final possibility is
\begin{itemize}
\item[$\circ$] If $\nu'=\nu$ and $x^{\nu'}c(x)\equiv x^\nu B(x,0,0)$
and $|Y|\ge |X|$ 
\begin{eqnarray}
Y&+&x^{\nu'}c(x)-(X+x)^\nu B(x,X,Y)
\nonumber\\
&=&Y-\{(X+x)^\nu- x^\nu\} B(x,X,Y)-x^\nu \{B(x,X,Y)-B(x,0,0)\}
\nonumber
\end{eqnarray}
so that
$$
\big|Y+x^{\nu'}c(x)-(X+x)^\nu B(x,X,Y)\big|
\ge |Y| -O(d)O(|X|+|Y|)
\ge \sfrac{1}{2} |Y|
$$
if $d$ is small enough. As a result 
$\set{\bk\in\cF}{|\bk-\tilde\bk|\le d,\ \big|e\big(\bp+\bk\big)\big|\le M^j}$
is empty if $|Y|$ is larger than some constant times $M^j$. 
On the other hand, if  $|Y|$ is smaller than a constant times $M^j$, then
$\sfrac{M^j}{|\bp-\tilde\bp|}$ is larger than some constant.
\end{itemize}

\goodbreak\noindent
{\it Case 2: $\tilde\bp+\tilde\bk$ is a singular point.\ \ \ }
By Lemma \ref{le:normalform},
$$
e(\tilde\bp+\tilde\bk+\mM\bq)=a(\bq)\big(q_1-q_2^{\nu_1}b(\bq)\big)\big(q_2-q_1^{\nu_2} c(\bq)\big)
$$
where $2\le \nu_1,\nu_2< n_0$ are integers, $\mM$ is a constant, 
nonsingular matrix and $a(\bk)$, $b(\bk)$ and $c(\bk)$ are 
$C^{r-2-\max\{\nu_1,\nu_2\}}$ functions  that are bounded and bounded 
away from zero. 

Suppose that the tangent line to $\mM^{-1}\cF$ at $\tilde\bk$ (when $\tilde\bk$ 
is a singular point, the tangent line of one branch) makes an angle of at most
$45^\circ$ with the $x$--axis. Otherwise exchange the roles of the $q_1$
and $q_2$ coordinates. If $d$ is small enough, we can write the equation 
of $\cF$ for $\bk$ within a distance $d$ of $\tilde \bk$ 
(when $\tilde\bk$ is a singular point, the equation of the branch 
under consideration) as 
$$
\big(\mM^{-1}(\bk-\tilde\bk)\big)_2
=\big(\mM^{-1}(\bk-\tilde\bk)\big)_1^{\nu'}
v\big(\big(\mM^{-1}(\bk-\tilde\bk)\big)_1\big)
$$ 
for some $1\le \nu'< n_0$ and some $C^{r-n_0-1}$ function $v$ that 
is bounded and bounded away from zero. 
Then, writing $\mM^{-1}(\bp-\tilde\bp)=(X,Y)$ and
$\big(\mM^{-1}(\bk-\tilde\bk)\big)_1=x$ and assuming that $\bk\in\cF$,
\begin{eqnarray}
e(\bp+\bk)
&=&
e\big(\tilde\bp+\tilde\bk+\mM\mM^{-1}(\bp-\tilde\bp+\bk-\tilde\bk)\big)
\nonumber\\
&=&
e\big(\tilde\bp+\tilde\bk+\mM(X+x,Y+x^{\nu'}v(x))\big)
\nonumber\\
&=&
A(x,X,Y)F(x,X,Y)G(x,X,Y)
\nonumber
\end{eqnarray}
where
\begin{eqnarray}
A(x,X,Y)
&=&
a\big(X+x,Y+x^{\nu'}v(x)\big)
\nonumber\\
F(x,X,Y)
&=&
X+x-(Y+x^{\nu'}v(x))^{\nu_1} B(x,X,Y)
\nonumber\\
G(x,X,Y)
&=&
Y+x^{\nu'}v(x)-(X+x)^{\nu_2} C(x,X,Y)
\nonumber\\
B(x,X,Y)
&=&
 b\big(X+x,Y+x^{\nu'}v(x)\big)
\nonumber\\
C(x,X,Y)
&=&
 c\big(X+x,Y+x^{\nu'}v(x)\big)
\nonumber
\end{eqnarray}
The functions $A(x,X,Y)$, $B(x,X,Y)$, $C(x,X,Y)$ and $v(x)$ are all 
$C^{r-n_0-1}$ and bounded and bounded away from zero. 

As in case 1, $y=x^{\nu'} v(x)$ is the equation of a fragment of $\mM^{-1}\cF$ 
translated so as to move $\tilde\bk$ to $\bzer$ and $y=x^{\nu_2} c(x,y)$ is 
the equation of a fragment of $\mM^{-1}\cF$ translated so as to 
move $\tilde\bp+\tilde\bk$ to $\bzer$. If $\tilde\bp=\bzer$, these two fragments may be 
identical in which case 
$x^{\nu'} v(x)\equiv x^{\nu_2} c\big(x,x^{\nu'} v(x)\big)$ and $\nu_2=\nu'$.
This case has already been dealt with in Lemma \ref{le:4.2}. 
Otherwise, the nonnesting condition says that there is an $n\in\bN$ such 
that if $y=x^{\nu_2} c(x,y)$ is rewritten in the form $y=x^{\nu_2} V(x)$, 
then the $n^{\rm th}$ derivative of  $x^{\nu_2} V(x)-x^{\nu'} v(x)$ must not 
vanish at $x=0$. Let $n < n_0$ be the smallest such natural number. Since 
$x^{\nu_2} V(x)\equiv x^{\nu_2} c\big(x,x^{\nu_2} V(x)\big)$ and 
since derivatives at $0$ of $x^{\nu_2} V(x)$ of order lower than $n$ agree with 
the corresponding derivatives of $x^{\nu'} v(x)$, the $n^{\rm th}$ derivative
of  $x^{\nu'} v(x)-x^{\nu_2} c\big(x,x^{\nu'} v(x)\big)
    =x^{\nu'}v(x)-x^{\nu_2} C(x,0,0)$ must not vanish at $x=0$.
So the remaining cases are:
\begin{itemize}

\item[$\circ$] If $\nu'<\nu_2$, then
$$
\sfrac{d^{\nu'}\hfill}{dx^{\nu'}}G(x,X,Y)
=\nu'!\,v(0)+O(d)
$$
Since
$$
F(x,X,Y)=O(d)\qquad
\sfrac{d\hfill}{dx}F(x,X,Y)=1+O(d)
$$
and applying zero to $\nu'-1$ $x$--derivatives
to $G(x,X,Y)$ gives $O(d)$, we have
$$
\sfrac{d^{\nu'+1}\hfill}{dx^{\nu'+1}}F(x,X,Y)G(x,X,Y)
=(\nu'+1)!\,v(0)+O(d)
$$
uniformly bounded away from zero, if $d$ is small enough.

\item[$\circ$] If $\nu'>\nu_2$,
$$
\sfrac{d^{\nu_2}\hfill}{dx^{\nu_2}}G(x,X,Y)
=-\nu_2!\,C(0,0,0)+O(d)
$$
Again
$$
F(x,X,Y)=O(d)\qquad
\sfrac{d\hfill}{dx}F(x,X,Y)=1+O(d)
$$
and applying zero to $\nu_2-1$ $x$--derivatives
to $G(x,X,Y)$ gives $O(d)$, so that
$$
\sfrac{d^{\nu_2+1}\hfill}{dx^{\nu_2+1}}F(x,X,Y)G(x,X,Y)
=-(\nu_2+1)!\,C(0,0,0)+O(d)
$$
is uniformly bounded away from zero, if $d$ is small enough.

\item[$\circ$] If $\nu'=\nu_2$ and $x^{\nu'}v(x)\not\equiv x^{\nu_2} C(x,0,0)$
\begin{eqnarray}
&&\sfrac{d^n\hfill}{dx^n}G(x,X,Y)
=
\sfrac{d^n\hfill}{dx^n}\Big[Y+x^{\nu'}v(x)-x^{\nu_2} C(x,0,0)
\nonumber\\
&& \hskip1.7in+\big[x^{\nu_2} C(x,0,0)-(X+x)^{\nu_2} C(x,X,Y)\big]\Big]
\nonumber\\
&&\hskip0.2in=
\sfrac{d^n\hfill}{dx^n}\big(x^{\nu'}v(x)-x^{\nu_2} C(x,0,0)\big)\big|_{x=0}
+O(d) +O(|X|+|Y|)
\nonumber
\end{eqnarray}
and applying strictly fewer than $n$ derivatives gives $O(d)$.
As in the last two cases
$$
\sfrac{d^{n+1}\hfill}{dx^{n+1}}F(x,X,Y)G(x,X,Y)
=n\sfrac{d^n\hfill}{dx^n}\big(x^{\nu'}v(x)-x^{\nu_2} C(x,0,0)\big)\big|_{x=0}
+O(d)
$$
is uniformly bounded away from zero, if $d$ is small enough.
\end{itemize}
The lemma now follows by Lemma \ref{le:interval}, as in Case 1.
\proofendsign

\newpage
\section{Regularity} 
\subsection{The gradient of the self--energy}

\subsubsection{The second order contribution}
Let
$$
C(k)=\sfrac{U(k)}{ik_0-e(\bk)}
$$
where the ultraviolet cutoff $U(k)$ is a smooth compactly supported
function that is identically one for all $k$ with  $|ik_0-e(\bk)|$ 
sufficiently small. We consider the value
\begin{eqnarray}
F(q)&=&\int d^3k^{(1)} d^3k^{(2)} d^3k^{(3)}\ 
          \de(k^{(1)}+k^{(2)}-k^{(3)}-q)\ C(k^{(1)})C(k^{(2)})C(k^{(3)})
\nonumber\\
      &&\hskip3.2in V(k^{(1)},k^{(2)},k^{(3)},q)\,
\nonumber
\end{eqnarray}
of the diagram
$$
\figplace{secondOrderB}{0 in}{0.2in}
$$
The function $V$ is a second order polynomial in the interaction function 
$\hat v$. For details, as well as the generalization to frequency--dependent
interactions, see \cite{FST2}. For the purposes of the present discussion, 
all we need is a simple regularity assumption on $V$.

\begin{lemma}\label{le:secondorder}
Assume \Ha--\Hf. If $V(k^{(1)},k^{(2)},k^{(3)},q)$ 
is $C^1$, then $F(q)$ is $C^1$ in the spatial coordinates $\bq$.
\end{lemma}

\begin{proof}
Introduce our standard partition of unity of a neighbourhood of the Fermi
surface \cite[\S 2.1]{FST1}
$$
U(k)=\sum_{j< 0}f(M^{-2j}|ik_0-e(\bk)|^2)
$$
where $f(M^{-2j}|ik_0-e(\bk)|^2)$ vanishes unless 
$M^{j-2}\le |ik_0-e(\bk)| \le M^j$. We have
$$
C(k)=\sum_{j< 0}C_j(k)
\qquad\hbox{where}\qquad
C_j(k)=\sfrac{f(M^{-2j}|ik_0-e(\bk)|^2)}{ik_0-e(\bk)}
$$
and
\begin{eqnarray}
F(q)
&=&
\sum_{j_1,j_2,j_3< 0}\int d^3k^{(1)} d^3k^{(2)} d^3k^{(3)}
       \ \ \de(k^{(1)}+k^{(2)}-k^{(3)}-q)
\nonumber\\
&& \hskip2.5in C_{j_1}(k^{(1)})C_{j_2}(k^{(2)})C_{j_3}(k^{(3)})\, V
\nonumber
\end{eqnarray}
Route the external momentum $q$ through the line with smallest $|j_i|$
(i.e. use the delta function to evaluate the integral over the $k^{(i)}$
corresponding to the smallest $|j_i|$) and apply $\nabla_\bq$. Rename the
remaining integration variables $k^{(i)}$ to $k$ and $p$.
Permute the indices so that $j_1\le j_2\le j_3$. If the $\nabla_\bq$
acts on $V$, the estimate is easy. For each fixed $j_1$, $j_2$ and $j_3$,
\begin{itemize}
\item[$\circ$] the volume of the 
domain of integration is bounded  by a constant times $|j_1|M^{2j_1}
|j_2|M^{2j_2}$, by Lemma \ref{FS1-le:2.4} of \cite{SFS1}. (The $k_0$ 
and $p_0$ components contribute $M^{j_1}M^{j_2}$ to this bound.)
\item[$\circ$] and the integrand is bounded by $\const M^{-j_1}M^{-j_2}M^{-j_3}$. 
\end{itemize}
so that
\begin{eqnarray*}
&&\sum_{j_1\le j_2\le j_3<0}\int d^3k d^3p\ 
  \big|C_{j_1}(k)C_{j_2}(p)C_{j_3}(\pm k\pm p\pm q)\ \nabla_\bq V\big|\\
&&\hskip0.5in\le \const \sum_{j_1\le j_2\le j_3<0}|j_1|M^{j_1}|j_2|M^{j_2}M^{-j_3}
\le \const \sum_{j_1\le j_2<0}|j_1|M^{j_1}|j_2|\\
&&\hskip0.5in\le \const \sum_{j_1<0}|j_1|^3M^{j_1}
\end{eqnarray*}
is uniformly bounded. So assume that the $\nabla_\bq$
acts on $C_{j_3}(\pm k\pm p\pm q)$. The terms of interest are now of the form
$$
\int d^3k d^3p\ V\ C_{j_1}(k)C_{j_2}(p)\nabla_\bq C_{j_3}(\pm k\pm p\pm q)
$$ 
with $j_1\le j_2\le j_3<0$ and
\begin{eqnarray}\label{eq:diffprop}
\nabla_\bq C_{j_3}(\pm k\pm p\pm q)
&=&
\pm\Big[\sfrac{\nabla e(\tilde\bk)}{[i\tilde k_0-e(\tilde\bk)]^2}
f({\sst M^{-2j_3}|i\tilde k_0-e(\tilde\bk)|^2})
\nonumber\\
&&\hskip0.2in
+\sfrac{2M^{-2j_3}e(\tilde\bk)
        \nabla e(\tilde\bk)}{i\tilde k_0-e(\tilde\bk)}
   f'({\sst M^{-2j_3}|i\tilde k_0-e(\tilde \bk)|^2})
   \Big]_{\tilde k=\pm k\pm p\pm q}
\end{eqnarray}
Observe that $\big|\nabla_\bq C_{j_3}(\pm k\pm p\pm q)\big|\le\const M^{-2j_3}$ 
since $|i\tilde k_0-e(\tilde\bk)|\ge \const M^{j_3}$ and $|e(\tilde\bk)|
\le\hbox{\rm const}'\  M^{j_3}$ on the support of 
$f(M^{-2j_3}|i\tilde k_0-e(\tilde\bk)|^2)$.
Choose three small constants $\et,\tilde\et,\veps>0$ such that $0<\veps\le\sfrac{1}{2n_0}$,
$0<\et<\sfrac{2n_0-1}{2n_0+2}\veps$ and
$\tilde\et\le\sfrac{2\et+\veps}{2n_0+2\et+\veps}$. 
Here $n_0$ is the integer in  Proposition \ref{le:length2}.
For example, if $n_0=3$, we can choose $\veps=\sfrac{1}{6}$, 
$\et=\sfrac{1}{10}<\sfrac{5}{48}$ and $\tilde\et=\sfrac{1}{20}$.

\bigskip\noindent{\bf  Reduction 1:} 
For any $\tilde\et>0$, it suffices to consider $j_1\le j_2\le j_3\le 
(1-\tilde\et)j_1$. For the remaining terms, we simply bound
\begin{eqnarray}
&&\bigg|\int d^3k d^3p\ V\ C_{j_1}(k)C_{j_2}(p)\nabla_\bq C_{j_3}(\pm k\pm p\pm q)
\bigg|
\nonumber\\
&&
\hskip0.5in\le\const \int d^3k d^3p\ f(\sfrac{|ik_0-e(\bk)|^2}{ M^{2j_1}})M^{-j_1}\  
                          f(\sfrac{|ip_0-e(\bp)|^2}{M^{2j_2}})M^{-j_2}\ 
                          M^{-2j_3}
\nonumber\\
&&
\hskip0.5in\le\const |j_1|M^{j_1}\  |j_2|M^{j_2}\ M^{-2j_3}
\nonumber
\end{eqnarray}
and
\begin{eqnarray}
&&
\sum_{j_1\le j_2\le j_3<0\atop j_3\ge (1-\tilde\et)j_1}
|j_1|M^{j_1}\  |j_2|M^{j_2}\ M^{-2j_3}
\nonumber\\
&&\hskip0.5in\le
\const \sum_{j_1\le j_2\le j_3<0}
|j_1|M^{j_1}\  |j_2|M^{j_2}\ M^{-j_3-(1-\tilde\et)j_1}
\nonumber\\
&&\hskip0.5in\le
\const \sum_{j_1\le j_2<0}
|j_1|\,|j_2|\,M^{j_1}\ M^{-(1-\tilde\et)j_1}
\nonumber\\  
&&\hskip0.5in=
\const \sum_{j_1<0}
|j_1|^3M^{\tilde\et j_1}<\infty
\nonumber
\end{eqnarray}
\goodbreak
\bigskip\noindent{\bf  Reduction 2:} For any $\et>0$, it suffices to consider $(\bk,\bp)$ with 
$|\pm \bk\pm \bp\pm \bq-\tilde \bq|\ge M^{\et j_3}$ for all singular points 
$\tilde \bq$. Let $\Xi_{j_3}(\bk)$ be the characteristic
function of the set 
$$
\set{\bk\in\bR^2}{|\bk-\tilde \bq|\ge M^{\et j_3}\hbox{ for all singular points
} \tilde \bq}
$$
If $|\pm \bk\pm \bp\pm \bq-\tilde\bq|\le M^{\et j_3}$ for some singular point
$\tilde \bq$,
then $\big|\nabla e(\pm \bk\pm \bp\pm \bq)\big|\le \const M^{\et j_3}$
and we may bound
\begin{eqnarray}
&&\bigg|\int d^3k d^3p\ V\ C_{j_1}(k)C_{j_2}(p)
\nabla_\bq C_{j_3}(\pm k\pm p\pm q)
\ \big(1-\Xi_{j_3}(\pm \bk\pm \bp\pm \bq)\big)\bigg|
\nonumber\\
&&\hskip0.5in
\le\const \int d^3k d^3p\ f(\sfrac{|ik_0-e(\bk)|^2}{ M^{2j_1}})M^{-j_1}\  
                          f(\sfrac{|ip_0-e(\bp)|^2}{M^{2j_2}})M^{-j_2}\ 
                          M^{-(2-\et)j_3}
\nonumber\\
&&\hskip0.5in
\le\const |j_1|M^{j_1}\  |j_2|M^{j_2}\ M^{-(2-\et)j_3}
\nonumber
\end{eqnarray}
and
\goodbreak
\begin{eqnarray}
&&
\sum_{j_1\le j_2\le j_3<0}
|j_1|M^{j_1}\  |j_2|M^{j_2}\ M^{-(2-\et)j_3}
\nonumber\\
&&\hskip0.5in\le
\const \sum_{j_1\le j_2<0}
|j_1|M^{j_1}\  |j_2|M^{j_2}\ M^{-(2-\et)j_2}
\nonumber\\
&&\hskip0.5in\le
\const \sum_{j_1\le j_2<0}
|j_1|^2M^{j_1}\  M^{-j_2+\et j_2}
\nonumber\\ 
&&\hskip0.5in\le
\const \sum_{j_1<0}
|j_1|^2M^{j_1}\  M^{-j_1+\et j_1}
\nonumber\\ 
&&\hskip0.5in=
\const \sum_{j_1<0}
|j_1|^2M^{\et j_1}<\infty
\nonumber
\end{eqnarray}

\bigskip\noindent{\bf  Current status:}
It remains to bound
\begin{eqnarray}
&&\sum_{j_1\le j_2\le j_3<0\atop j_3\le (1-\tilde\et)j_1}\!
\bigg|\int d^3k d^3p\ V\ C_{j_1}(k)C_{j_2}(p)\nabla_\bq C_{j_3}(\pm k\pm p\pm q)
\ \Xi_{j_3}(\pm \bk\pm \bp\pm \bq)\bigg|
\nonumber\\
&&\hskip0.1in
\le\const \sum_{j_1\le j_2\le j_3<0\atop j_3\le (1-\tilde\et)j_1}
          \int d^3k d^3p\ f(\sfrac{|ik_0-e(\bk)|^2}{ M^{2j_1}})M^{-j_1}\  
                          f(\sfrac{|ip_0-e(\bp)|^2}{M^{2j_2}})M^{-j_2}
\nonumber\\ 
  &&\hskip5.5cm       M^{-2j_3}\ \Xi_{j_3}(\pm \bk\pm \bp\pm \bq)
                  \chi_{j_3}(\pm \bk\pm \bp\pm \bq)
\nonumber\\
&&\hskip0.1in
\le\const\! \sum_{j_1\le j_2\le j_3<0\atop j_3\le (1-\tilde\et)j_1}M^{-2j_3}
          \int d^2\bk d^2\bp\ \chi_{j_1}(\bk)\  \chi_{j_2}(\bp)\; 
                (\chi_{j_3}\Xi_{j_3})\ (\pm \bk\pm \bp\pm \bq)
\nonumber\\
&&\hskip0.1in
\le\const\! \sum_{j_1,j_2, j_3<0\atop {j_3\over 1-\tilde\et}\le j_1,j_2\le j_3}
          \hskip-10pt 
          M^{-2j_3}
          \int d^2\bk d^2\bp\ \chi_{j_1}(\bk)\  \chi_{j_2}(\bp)\  
                (\chi_{j_3}\Xi_{j_3})\ (\pm \bk\pm \bp\pm \bq)
\nonumber
\end{eqnarray}
where $\chi_j(\bk)$ is the characteristic function of the set of $\bk$'s
with $|e(\bk)|\le M^j$. 

Make a change of variables with $\pm\bk\pm \bp\pm \bq$ becoming the new 
$\bk$ integration variable. This gives
\begin{eqnarray}
\const \sum_{j_1,j_2, j_3<0\atop {j_3\over 1-\tilde\et}\le j_1,j_2\le j_3}
     M^{-2j_3}\int d^2\bk d^2\bp\ \chi_{j_1}(\pm\bk\pm\bp\pm\bq)\  
           \chi_{j_2}(\bp)\ \chi_{j_3}(\bk)\,\Xi_{j_3}(\bk)
\nonumber
\end{eqnarray}\goodbreak
\bigskip\noindent{\bf  Reduction 3:} 
It suffices to show that, for each fixed $\tilde\bk$, $\tilde\bp$ and  
$\tilde\bq$ in $\bR^2$, there are (possibly $\tilde\bk$, $\tilde\bp$, 
$\tilde\bq$ dependent, but $j_i$ independent)  constants $c$ and $C$ such that
\begin{eqnarray}\label{eqnPtwise}
\int_{|\bk-\tilde \bk|\le c}\hskip-20pt d^2\bk\hskip5pt
\int_{|\bp-\tilde\bp|\le c}  \hskip-20 pt d^2\bp\ 
    \chi_{j_1}(\pm\bk\pm\bp\pm\bq)\  \chi_{j_2}(\bp)\  
                \chi_{j_3}(\bk)\,\Xi_{j_3}(\bk)\hskip1in \nonumber \\
  \hskip2.5in  \le C|j_2| M^{j_2}M^{(1-\et)j_3}M^{\veps j_3-\et' j_3}
\end{eqnarray}
for all $\bq$ obeying $|\bq-\tilde \bq|\le c$. The constant $\et'$ will be
chosen later and will obey $\veps-\et-\et'>0$.
Since 
$$
\set{(\bk,\bp,\bq)\in\bR^6}{|e(\bk)|\le 1, |e(\bp)|\le 1,\ |e(\pm\bk\pm\bp\pm\bq)|\le 1}
$$ 
is compact, once this is proven
we will have the bound
\begin{eqnarray}
&&\hskip-0.3in
\const \sum_{j_1,j_2, j_3<0\atop {j_3\over 1-\tilde\et}\le j_1,j_2\le j_3}
         M^{-2j_3}\int d^2\bk d^2\bp\ 
         \chi_{j_1}(\pm\bk\pm\bp\pm\bq)\  \chi_{j_2}(\bp)\  
                \chi_{j_3}(\bk)\,\Xi_{j_3}(\bk)
\nonumber\\
&&\le 
\const \sum_{j_1,j_2, j_3<0\atop {j_3\over 1-\tilde\et}\le j_1,j_2\le j_3}
M^{-2j_3} |j_2| M^{j_2}M^{(1-\et)j_3}M^{\veps j_3-\et' j_3}
\nonumber\\
&&\le\const \sum_{j_3<0}\big(\sfrac{|j_3|}{1-\tilde\et}\big)^3 
           M^{(\veps-\et-\et') j_3}
\nonumber\\
&&\le\const
\nonumber
\end{eqnarray}
since $\veps>\et+\et'$. Furthermore, if $\tilde \bk$ or $\tilde \bp$ or 
$\pm\tilde\bk\pm\tilde\bp\pm\tilde \bq$ does not lie on $\cF$, we can 
choose $c$ sufficiently
small that the integral of  \Ref{eqnPtwise} vanishes whenever 
$|\bq-\tilde\bq|\le c$ and $|j_1|,|j_2|,|j_3|$ are large enough. So it 
suffices to require that $\tilde \bk$, $\tilde \bp$ and 
$\pm\tilde \bk\pm\tilde\bp\pm\tilde \bq$ all lie on $\cF$.

\bigskip\noindent{\bf  Reduction 4:} 
If $\tilde\bk$ is not a singular point, make a change of variables to
$\rho=e(\bk)$ and an ``angular'' variable $\th$. So the 
condition $\chi_{j_3}(\bk)\ne 0$ forces $|\rho|\le \const M^{j_3}$. 
If $\tilde \bk$ is a singular
point, the condition $\Xi_{j_3}(\bk)\ne 0$ forces $|\bk-\tilde\bk|\ge\const
M^{\et j_3}$ and this, in conjunction with the condition that 
$\chi_{j_3}(\bk)\ne 0$, forces $\bk$ to lie fairly near one of the two 
branches of $\cF$ at $\tilde\bk$ at least a distance $\const M^{\et j_3}$ 
from $\tilde \bk$.  Using Lemma \ref{le:normalform}, we can make a change 
of variables such that $e\big(\bk(\rho,\th)\big)=\rho\th$ and either 
$|\th|\ge\const M^{\et j_3}$ or $|\rho|\ge\const M^{\et j_3}$.
Possibly exchanging the roles of $\rho$ and $\th$, we may, without loss 
of generality assume the former.  Then the 
condition $\chi_{j_3}(\bk)\ne 0$ forces $|\rho|\le \const M^{(1-\et)j_3}$. Thus,
regardless of whether $\tilde\bk$ is singular or not,
\begin{eqnarray}
&&\hskip-0.3in\int_{|\bk-\tilde \bk|\le c}\hskip-20pt d^2\bk\hskip5pt
\int_{|\bp-\tilde\bp|\le c}  \hskip-20 pt d^2\bp\ 
    \chi_{j_1}(\pm\bk\pm\bp\pm\bq)\  \chi_{j_2}(\bp)\  
                \chi_{j_3}(\bk)\,\Xi_{j_3}(\bk)
\nonumber\\
&&\hskip0.2in\le
\const \int_{|\th|\le 1\atop |\rho|\le {\rm const}\, M^{(1-\et)j_3}}
\hskip-20pt d\rho d\th\hskip5pt
\int_{|\bp-\tilde\bp|\le c}  \hskip-20 pt d^2\bp\ 
    \chi_{j_1}(\pm\bk(\rho,\th)\pm\bp\pm\bq)\  \chi_{j_2}(\bp)
\nonumber\\
&&\hskip0.2in\le
\const \int_{|\th|\le 1\atop |\rho|\le {\rm const}\, M^{(1-\et)j_3}}
\hskip-20pt d\rho d\th\hskip5pt
\int_{|\bp-\tilde\bp|\le c}  \hskip-20 pt d^2\bp\ 
    \chi_{j'}(\pm\bk(0,\th)\pm\bp\pm\bq)\  \chi_{j_2}(\bp)
\nonumber\\
&&\hskip0.2in\le
\const M^{(1-\et)j_3}\int_{|\th|\le 1 }\hskip-10pt d\th\hskip5pt
\int_{|\bp-\tilde\bp|\le c}  \hskip-20 pt d^2\bp\ 
    \chi_{j'}(\pm\bk(0,\th)\pm\bp\pm\bq)\  \chi_{j_2}(\bp)
\nonumber
\end{eqnarray}
where $M^{j'}=M^{j_1}+\const M^{(1-\et)j_3}\le \const M^{(1-\et)j_3}$.
Thus it suffices to prove that
$$
\int_{|\bp-\tilde\bp|\le c} d^2\bp\ \int_{|\th|\le 1 }d\th\ 
    \chi_{j'}(\pm\bk(0,\th)\pm\bp\pm\bq)\  \chi_{j_2}(\bp) 
\le C|j_2| M^{j_2}M^{\veps j_3-\et' j_3}
$$
for all $\bq$ obeying $|\bq-\tilde \bq|\le c$.

\medskip\noindent{\bf  The Final Step:}
We apply Proposition \ref{le:length2} with $j=j'$ and 
$\de= M^{(1+\veps)j_3/2}$. Denote by $\tilde
\chi(\bp)$ the characteristic function of the set of $\bp$'s with
$$
\mu\Big(\set{-1\le \th\le 1}
            {\big|e\big(\pm\bp\pm\bq\pm\bk(0,\th)\big)\big|\le M^{j'}}\Big)
\ge c_1\big(\sfrac{M^{j'}}{M^{(1+\veps)j_3/2}}\big)^{1/n_0}
$$ 
where  $c_1$ is the supremum of $\sfrac{d\th}{ds}$ ($s$ is arc length).
If $\pm\bp\pm\bq$ is not in the set of measure $D\de^2$ specified in
Proposition \ref{le:length2}, then 
$$
\ell\Big(\set{\bk\in\cF}
        {\big|e\big(\pm\bp\pm\bq\pm\bk\big)\big)\big|\le M^{j'}}\Big)
< \big(\sfrac{M^{j'}}{\de}\big)^{1/n_0} 
$$
and hence\goodbreak
\begin{eqnarray}
&&\mu\Big(\set{-1\le \th\le 1}
  {\big|e\big(\pm\bp\pm\bq\pm\bk(0,\th)\big)\big|\le  M^{j'}}\Big)
\nonumber\\
&&\hskip1in=\int_{-1}^1 d\th \
\chi\Big(\big|e\big(\pm\bp\pm\bq\pm\bk(0,\th)\big)\big|\le  M^{j'}\Big)
\nonumber\\
&&\hskip1in=\int ds \sfrac{d\th}{ds} \
\chi\Big(\big|e\big(\pm\bp\pm\bq\pm\bk\big)\big|\le  M^{j'}\Big)
\nonumber\\
&&\hskip1in\le c_1 \ell\Big(\set{\bk\in\cF}
            {\big|e\big(\pm\bp\pm\bq\pm\bk\big)\big)\big|\le M^{j'}}\Big)
\nonumber\\
&&\hskip1in<c_1\big(\sfrac{ M^{j'}}
                      {M^{(1+\veps)j_3/2}}\big)^{1/n_0}
\nonumber
\end{eqnarray}
so that $\tilde\chi(\bp)=0$.
Thus $\tilde \chi(\bp)$ vanishes except on a set of measure $DM^{(1+\veps)j_3}$
and
\begin{eqnarray}
&&\int d^2\bp\ \int_{|\th|\le 1 }d\th\ 
    \chi_{j'}(\pm\bk(0,\th)\pm\bp\pm\bq)\  \chi_{j_2}(\bp) 
\nonumber\\
&&\hskip0.5in\le \int d^2\bp\ \tilde \chi(\bp)\chi_{j_2}(\bp)
    \int_{|\th|\le 1 }d\th\ \chi_{j'}(\pm\bk(0,\th)\pm\bp\pm\bq)
\nonumber\\
&&\hskip1in+\int d^2\bp\ \big(1-\tilde \chi(\bp)\big)\chi_{j_2}(\bp)
    \int_{|\th|\le 1 }d\th\ 
    \chi_{j'}(\pm\bk(0,\th)\pm\bp\pm\bq)
\nonumber\\
&&\hskip0.5in\le 2\int d^2\bp\ \tilde \chi(\bp)
    +\const\int d^2\bp\ \chi_{j_2}(\bp)
    \big(\sfrac{ M^{j'}}{M^{(1+\veps)j_3/2}}\big)^{1/n_0}
\nonumber\\
&&\hskip0.5in\le\const M^{(1+\veps)j_3}
    +\const |j_2|M^{j_2}
    \big(\sfrac{ M^{j'}}{M^{(1+\veps)j_3/2}}\big)^{1/n_0}
\nonumber\\
&&\hskip0.5in\le\const M^{j_3}M^{\veps j_3}
    +\const |j_2|M^{j_2}
    \big( M^{(1-\et)j_3-{1+\veps\over 2}j_3}\big)^{1/n_0}
\nonumber\\
&&\hskip0.5in\le\const M^{j_2}M^{(\veps-{\tilde\et\over1-\tilde\et}) j_3}
    +\const |j_2|M^{j_2}
    \big( M^{({1\over 2}-\et-{\veps\over 2})j_3}\big)^{1/n_0}
\nonumber\\
&&\hskip1.5in\hbox{since $j_3\le(1-\tilde\et)j_2
        \le j_2-\sfrac{\tilde\et}{1-\tilde\et}j_3$}
\nonumber\\
&&\hskip0.5in\le \const|j_2| M^{j_2}M^{\veps j_3-\et' j_3}
\nonumber
\end{eqnarray}
provided $\veps\le\sfrac{1}{2n_0}$, $\et'\ge\sfrac{\tilde\et}{1-\tilde\et}$
and $\et'\ge\sfrac{\et}{n_0}+\sfrac{\veps}{2n_0}$.

\smallskip
We  choose $\et'= \sfrac{\et}{n_0}+\sfrac{\veps}{2n_0}$. Then the remaining
conditions
\begin{equation}\label{eq:chooseet}
\veps>\et+\et'
\iff \veps>\et+\sfrac{\et}{n_0}+\sfrac{\veps}{2n_0}
\iff\sfrac{2n_0-1}{2n_0}\veps>\sfrac{n_0+1}{n_0}\et
\iff\et<\sfrac{2n_0-1}{2n_0+2}\veps
\end{equation}
and 
$$
\et'\ge\sfrac{\tilde\et}{1-\tilde\et}
\iff
\tilde\et\le\sfrac{\et'}{1+\et'}=\sfrac{2\et+\veps}{2n_0+2\et+\veps}
$$
are satisfied because of the conditions we imposed when we chose $\et,\tilde\et$
and $\veps$ just before Reduction 1.

We have now verified that $\nabla_\bq F(q)$ is a uniformly convergent sum (over
$j_1$, $j_2$, and $j_3$) of the continuous functions 
$\nabla_\bq\int d^3k d^3p\ V\ C_{j_1}(k)C_{j_2}(p) C_{j_3}(\pm k\pm p\pm q)$.
Hence $F(q)$ is $C^1$ in $\bq$.
\end{proof}

\subsubsection{The general diagram}
The argument of the last section applies equally well to general diagrams.
\begin{lemma}\label{le:spatderiv}
Let $G(q)$ be the value of any two--legged 1PI graph with external momentum
$q$. Then $G(q)$ is $C^1$ with respect to the spatial components $\bq$.
\end{lemma}
\begin{proof}
We shall simply merge the argument of the last section with the general
bounding argument of \cite[Appendix \ref{FS1-ap:A}]{SFS1}. 
This is a good time to read that
Appendix, since we shall just explain the modifications to be made to
it.  In addition to the small constants 
$\et,\et',\tilde\et,\veps>0$ of Lemma \ref{le:secondorder}, we choose a
small constant $\dimp>0$ and require\footnote{The first three conditions
as well as the condition that $\sfrac{\tilde\et}{1-\tilde\et}\le\et'$
were already present in Lemma \ref{le:secondorder}. The other conditions
are new.}  that
$$
0<\veps\le\sfrac{1}{2n_0},\ 
0<\et<\sfrac{2n_0-1}{2n_0+2}\veps,\ 
\et'=\sfrac{\et}{n_0}+\sfrac{\veps}{2n_0}\  {\rm and}\ 
\sfrac{\tilde\et}{1-\tilde\et}<\min\big\{\et',\veps-\et-\et'\big\}
$$ 
and
$$
\dimp\le\min\{\et\,,\,\tilde\et\,,\,(1+\veps-\et-\et')(1-\tilde\et)-1\}
$$
with $n_0$ being the integer in  Proposition \ref{le:length2}. All of these conditions
may be satisfied by
\begin{itemize}
\item[$\circ$] choosing $0<\veps\le\sfrac{1}{2n_0}$ and then 
\item[$\circ$] choosing $0<\et<\sfrac{2n_0-1}{2n_0+2}\veps$
(by \Ref{eq:chooseet}, this ensures that $\veps-\et-\et'>0$) and then
\item[$\circ$] choosing $0<\tilde\et<1$ so that
$\sfrac{\tilde\et}{1-\tilde\et}<\min\big\{\et',\veps-\et-\et'\big\}$
(this ensures that the expression $(1+\veps-\et-\et')(1-\tilde\et)-1>0$) and then
\item[$\circ$] choosing $\dimp>0$ so that 
$\dimp\le\min\{\et\,,\,\tilde\et\,,\,(1+\veps-\et-\et')(1-\tilde\et)-1\}$
\end{itemize} 

As in \cite[Appendix \ref{FS1-ap:A}]{SFS1}, use 
\cite[\Ref{FS1-eq:Cexpn}]{SFS1} to introduce a scale expansion
for each propagator and express $G(q)$ in terms of a renormalized 
tree expansion \cite[\Ref{FS1-eq:Gren}]{SFS1}. We shall prove, by 
induction on the depth, $D$, of $G^J$, the bound
\begin{equation}\label{eq:spatderivindhyp}
\sum_{J\in \cJ(j,t,R,G)}\sup_{q}\big|\partial_{q_0}^{s_0}\partial_\bq^{s_1}G^J(q)\big|
\le\cst{}{n} |j|^{3n-2}M^{j}M^{-s_0j}M^{-s_1(1-\dimp)j}
\end{equation}
for $s_0,s_1\in\{0,1\}$.  The notation is as in
\cite[Appendix \ref{FS1-ap:A}]{SFS1}: $n$ is the number of vertices in $G$ 
and $\cJ(j,t,R,G)$ is the set of all assignments $J$ of scales to the lines of 
$G$ that have root scale $j$, that give forest $t$ and that are compatible
with the assignment $R$ of renormalization labels to the two--legged forks
of $t$. (This is explained in more detail just before 
\cite[\Ref{FS1-eq:Gren}]{SFS1}.) If $s_0=0$ and $s_1=1$, the right hand side 
becomes $\cst{}{n}|j|^{3n-2}M^{\dimp j}$, which is summable over $j<0$, implying
that $G(q)$ is $C^1$ with respect to the spatial components $\bq$. If $s_1=0$,
\Ref{eq:spatderivindhyp} is contained in  \cite[Proposition \ref{FS1-le:fixedrootscale}]{SFS1},
so it suffices to consider $s_1=1$.

As in \cite[Appendix \ref{FS1-ap:A}]{SFS1}, if $D>0$, decompose the tree $t$ 
into a pruned tree $\tilde t$ and insertion subtrees $\tau^1,\cdots,\tau^m$ by 
cutting the branches beneath  all minimal $E_f=2$ forks $f_1,\cdots,f_m$. In 
other words each of the forks $f_1,\cdots,f_m$ is an $E_f=2$ fork having
no $E_f=2$ forks, except $\phi$, below it in $t$. Each $\tau_i$
consists of the fork $f_i$ and all of $t$ that is above $f_i$. It has depth
at most $D-1$ so the corresponding subgraph $G_{f_i}$ obeys
\Ref{eq:spatderivindhyp}. Think of each subgraph $G_{f_i}$ as
a generalized vertex in the graph $\tilde G=G/\{G_{f_1},\cdots,G_{f_m}\}$.
Thus $\tilde G$ now has two as well as four--legged vertices. These
two--legged vertices have kernels of the form
$
T_i(k)=\sum_{j_{f_i}\le j_{\pi(f_i)}}\ell G_{f_i}(k)
$
when $f_i$ is a $c$--fork and of the form
$
T_i(k)=\sum_{j_{f_i}> j_{\pi(f_i)}}(\bbbone-\ell)G_{f_i}(k)
$
when $f_i$ is an $r$--fork. At least one of the external lines of $G_{f_i}$ 
must be of scale precisely $j_{\pi(f_i)}$ so
the momentum $k$ passing through $G_{f_i}$ lies in the support of 
$C_{j_{\pi(f_i)}}$. In the case of a $c$--fork $f=f_i$ we have, as in
\cite[\Ref{FS1-eq:cforkA}]{SFS1} and using the same notation, by the inductive hypothesis,
\begin{eqnarray}
&&\hskip-23pt\sum_{j_{f}\le j_{\pi(f)}}\sum_{J_f\in\cJ(j_f,t_f,R_f,G_f)}
\hskip-6pt\sup_{k}\Big|\partial_\bk^{s_1}\ell G_{f}^{J_f}(k)\Big|
\le\sum_{j_{f}\le j_{\pi(f)}}\hskip-7pt\cst{}{n_f}|j_f|^{3n_f-2}M^{j_f}
         M^{-s_1(1-\dimp)j_f}\nonumber\\
&&\hskip70pt
\le \cst{}{n_f}|j_{\pi(f)}|^{3n_f-2}M^{j_{\pi(f)}}M^{-s_1(1-\dimp)j_{\pi(f)}}
\label{eq:cforkS} 
\end{eqnarray}
for $s_1=0,1$.
As $\ell G_{f}^{J_f}(k)$ is independent of $k_0$ derivatives with respect
to $k_0$ may not act on it.
In the case of an $r$--fork $f=f_i$, we have, as in 
\cite[\Ref{FS1-eq:rforkA}]{SFS1},
\begin{eqnarray}
&&\hskip-20pt\sum_{j_{f}> j_{\pi(f)}}\sum_{J_f\in\cJ(j_f,t_f,R_f,G_f)}
\sup_{k}\bbbone\big(C_{j_{\pi(f)}}(k)\ne 0\big)
\Big|\partial_{k_0}^{s_0}\partial_\bk^{s_1}
(\bbbone-\ell)G_{f}^{J_f}(k)\Big|\nonumber\\
&&\hskip70pt
\le \sum_{j_{f}> j_{\pi(f)}}\ \sum_{J_f\in\cJ(j_f,t_f,R_f,G_f)}
M^{(1-s_0)j_{\pi(f)}}\sup_{k}\Big|\partial_{k_0}\partial_\bk^{s_1}
    G_{f}^{J_f}(k)\Big|\nonumber\\
&&\hskip70pt
\le\cst{}{n_f}M^{(1-s_0)j_{\pi(f)}}
\sum_{j_{f}> j_{\pi(f)}}|j_{f}|^{3n_f-2}M^{-s_1(1-\dimp)j_f}\nonumber\\
&&\hskip70pt
\le\cst{}{n_f}|j_{\pi(f)}|^{3n_f-1}M^{j_{\pi(f)}}M^{-s_0j_{\pi(f)}}
M^{-s_1(1-\dimp)j_{\pi(f)}}\label{eq:rforkS}
\end{eqnarray}

Denote by $\tilde J$ the restriction to $\tilde G$ of the scale assignment
$J$. We bound $\tilde G^{\tilde J}$, which again is of the form 
\cite[\Ref{FS1-eq:tildeGtildeJform}]{SFS1}, by a variant of the six step 
procedure followed in \cite[Appendix \ref{FS1-ap:A}]{SFS1}. In fact the first 
four steps are identical.
\begin{enumerate} 
\item
Choose a spanning tree $\tilde T$ for $\tilde G$ with the property that 
$\tilde T\cap \tilde G^{\tilde J}_f$ is a connected tree for every 
$f\in t(\tilde G^{\tilde J})$. 
\item
 Apply any $q$--derivatives. By the product rule each derivative may act on any 
line or vertex on the ``external momentum path''. It suffices to consider 
any one such action.
\item
 Bound each two--legged renormalized subgraph (i.e.
$r$--fork) by \Ref{eq:rforkS} and each two--legged 
counterterm (i.e. $c$--fork) by \Ref{eq:cforkS}. 
Observe that when $s'_0$ $k_0$--derivatives and $s'_1$ $\bk$--derivatives act 
on the vertex, the bound is no worse than $M^{-s'_0j}M^{-s_1'(1-\dimp)j}$ 
times the bound with no derivatives, because we necessarily have 
$j\le j_{\pi(f)}<0$. 
\item
 Bound all remaining vertex functions, $u_\rv$, (suitably differentiated)
by their suprema in momentum space. 
\end{enumerate} 
We have already observed that if $s_1=0$, the bound \Ref{eq:spatderivindhyp} 
is contained in \cite[Proposition \ref{FS1-le:fixedrootscale}]{SFS1}, 
with $s=0$. In the event that $s_1=1$, but the spatial gradient acts on a 
vertex, \cite[Proposition \ref{FS1-le:fixedrootscale}]{SFS1}, again with 
$s=0$ but with either one $v$ replaced by its gradient or with an extra factor 
of $M^{-(1-\dimp)j}$ coming from Step 3, again gives \Ref{eq:spatderivindhyp}.

So it suffices to consider the case that $s_1=1$ and the spatial gradient
acts on a propagator of the ``external momentum path''. It is in this case
that we apply the arguments of Lemma \ref{le:spatderiv}. The heart of 
those arguments was the observation that, when the gradient acted on a line 
$\ell_3$ of scale $j_3$, the line $\ell_3$ also lay on distinct momentum
loops, $\La_{\ell_1}$ and $\La_{\ell_2}$, generated by lines, $\ell_1$ and
$\ell_2$ of scales $j_{\ell_1}$ and $j_{\ell_2}$ with $j_{\ell_1},
j_{\ell_2}\le j_{\ell_3}$. This is still the case and is proven in Lemmas
\ref{le:overlap} and \ref{le:overlapScale} below. So we may now apply
the procedure of Lemma \ref{le:secondorder}.

\bigskip\noindent{\bf  Reduction 1:} 
It suffices to consider $j\le j_{\ell_3}\le (1-\tilde\et)j$. For the 
remaining terms, we simply bound the differentiated propagator,
as in the argument following \Ref{eq:diffprop}, by
\begin{eqnarray*}
\big|\partial_{q_0}^{s_0'}\partial_\bq C_{j_{\ell_3}}\big(k_{\ell_3}(q)\big)\big|
    &\le&\const   M^{-2j_{\ell_3}}M^{-s_0'j_{\ell_3}}\\
    &\le&\const   M^{-j_{\ell_3}}M^{-s_0'j_{\ell_3}}M^{-(1-\tilde\et)j}\\
    &\le&\const   M^{-j_{\ell_3}}M^{-s_0'j_{\ell_3}}M^{-(1-\dimp)j}
\end{eqnarray*}
So for the terms with $j_{\ell_3}> (1-\tilde\et)j$, the effect of the spatial
gradient is to degrade the $s_1=0$ bound by at most a factor of 
$\const M^{-(1-\dimp)j}$ and we may apply the rest of \cite[Proposition 
\ref{FS1-le:fixedrootscale}]{SFS1}, starting with step 5,  without further 
modification. 

\goodbreak
\bigskip\noindent{\bf  Reduction 2:} It suffices to consider
loop momenta $\big(k_\ell\big)_{\ell\in\tilde G\setminus\tilde T}$,
in the domain of integration of \cite[\Ref{FS1-eq:tildeGtildeJform}]{SFS1}, 
for which the momentum $\bk_{\ell_3}$ flowing through $\ell_3$ (which
is a linear combination of $\bq$ and various loop momenta) remains a distance
at least $M^{\et j_{\ell_3}}$ from all singular points. 
If $\bk_{\ell_3}$ is at most a distance $M^{\et j_{\ell_3}}$ from some 
singular point then $\big|\partial_\bk e(\bk_{\ell_3})\big|
\le \const M^{\et j_{\ell_3}}$
and we may use the $\nabla e$ in the numerator of \Ref{eq:diffprop}
to improve the bound on the $\ell_3$ propagator to
\begin{eqnarray*}
\big|\partial_{q_0}^{s_0'}\partial_\bq C_{j_{\ell_3}}\big(k_{\ell_3}(q)\big)\big|
    &\le&\const   M^{-2j_{\ell_3}}M^{-s_0'j_{\ell_3}}M^{\et j_{\ell_3}}\\
    &\le&\const   M^{-j_{\ell_3}}M^{-s_0'j_{\ell_3}}M^{-(1-\et)j}\\
    &\le&\const   M^{-j_{\ell_3}}M^{-s_0'j_{\ell_3}}M^{-(1-\dimp)j}
\end{eqnarray*}
Once again, in this case, the effect of the spatial
gradient is to degrade the $s_1=0$ bound by at most a factor of 
$\const M^{-(1-\dimp)j}$ and we may apply the rest of the proof 
\cite[Proposition \ref{FS1-le:fixedrootscale}]{SFS1}, starting with step 5,  
without further modification.

\bigskip\noindent{\bf  Reduction 3:} 
Now apply step 5, that is, bound every propagator.
The extra $\partial_\bq$ acting on $C_{j_{\ell_3}}$ gives a  factor
of $M^{-s_1j_{\ell_3}}\le M^{-s_1j}$ worse than the bound achieved in step 5
of  \cite[Proposition \ref{FS1-le:fixedrootscale}]{SFS1}.
Prepare for the application of step 6, the integration over loop
momenta, by ordering the integrals in such a way that the two integrals
executed first (that is the two innermost integrals) are those over 
$k_{\ell_1}$ and $k_{\ell_2}$. The momentum flowing through $\ell_3$
is of the form $k_{\ell_3}=\pm k_{\ell_1}\pm k_{\ell_2}+ q'$, where $q'$
is some linear combination of the external momentum $q$ and possibly
various other loop momenta. Make a change of variables from $\bk_{\ell_1}$ 
and $\bk_{\ell_2}$ to $\bk = \bk_{\ell_3}
=\pm \bk_{\ell_1}\pm \bk_{\ell_2}+ \bq'$ and $\bp=\bk_{\ell_2}$. 
It now suffices to show that, for each fixed $\tilde\bk$, $\tilde\bp$ and 
$\tilde\bq'$ in $\bR^2$, there are (possibly $\tilde\bk$, $\tilde\bp$, 
$\tilde\bq'$ dependent, but $j_{\ell_i}$ independent)  constants $c$ and $C$ 
such that
\begin{eqnarray}\label{eqnPtwiseSpatial}
\int_{|\bk-\tilde \bk|\le c}\hskip-20pt d^2\bk\hskip5pt
\int_{|\bp-\tilde\bp|\le c}  \hskip-20 pt d^2\bp\ 
    \chi_{j_{\ell_1}}(\pm\bk\pm\bp\pm\bq')\  \chi_{j_{\ell_2}}(\bp)\  
            \chi_{j_{\ell_3}}(\bk)\,\Xi_{j_{\ell_3}}(\bk)\hskip1in \nonumber \\
  \hskip2.25in  \le C|j_{\ell_2}| M^{j_{\ell_2}}
               M^{(1-\et)j_{\ell_3}}M^{\veps j_{\ell_3}-\et' j_{\ell_3}}
\end{eqnarray}
for all $\bq'$ obeying $|\bq'-\tilde \bq'|\le c$. Recall that 
$\chi_j(\bk)$ is the characteristic function of the set of $\bk$'s
with $|e(\bk)|\le M^j$ and $\Xi_j(\bk)$ is the characteristic
function of the set 
$$
\set{\bk\in\bR^2}{|\bk-\tilde \bq|\ge M^{\et j}\hbox{ for all singular points
} \tilde \bq}
$$ 
Once proven, the bound \Ref{eqnPtwiseSpatial} (together with the usual
compactness argument) replaces the bound
\begin{eqnarray*}
\int d^2\bk_{\ell_1}
\int d^2\bk_{\ell_2}\ 
    \chi_{j_{\ell_1}}(\bk_{\ell_1})\  \chi_{j_{\ell_2}}(\bk_{\ell_2}) 
          \le \const |j_{\ell_1}| M^{j_{\ell_1}}\ 
          |j_{\ell_2}| M^{j_{\ell_2}}
\end{eqnarray*}
used in step 6 of \cite[Proposition \ref{FS1-le:fixedrootscale}]{SFS1}. Since
$j\le j_{\ell_1},j_{\ell_2}\le j_{\ell_3}\le (1-\tilde \et) j$, the 
bound \Ref{eqnPtwiseSpatial} constitutes an improvement by a factor of
\begin{eqnarray*}
&&\hskip-20pt\const M^{j_{\ell_2}}
               M^{(1-\et)j_{\ell_3}}M^{\veps j_{\ell_3}-\et' j_{\ell_3}}
              M^{-j_{\ell_1}}\  M^{-j_{\ell_2}}
\le\const  M^{(1-\et-\et'+\veps)j_{\ell_3}} M^{-j_{\ell_1}}\\
&&\hskip1in\le \const  M^{[(1-\et-\et'+\veps)(1-\tilde \et)-1]j}\\
&&\hskip1in\le \const  M^{\dimp j}
\end{eqnarray*}
So once we proven the bound \Ref{eqnPtwiseSpatial}, we may continue with 
the rest of the proof of  
\cite[Proposition \ref{FS1-le:fixedrootscale}]{SFS1}, without further 
modification, and show that the inductive hypothesis \Ref{eq:spatderivindhyp} 
is indeed preserved.

In fact the bound \Ref{eqnPtwiseSpatial} has already been proven
in Reduction 4 and The Final Step of the proof of Lemma \ref{le:secondorder}.
 \end{proof}

\begin{lemma}\label{le:overlap}Let $G$  be any two--legged 1PI graph with
each vertex having an even number of legs.
Let $T$ be a spanning tree for $G$. Assume that the two external legs of
$G$ are hooked to two distinct vertices and that $\ell_3$ is a line of $G$ 
that is in the linear subtree of $T$ joining the external legs. Recall
that any line $\ell$ not in $T$ is associated with a loop
$\La_\ell$ that consists of $\ell$ and the linear subtree of $T$ joining
the vertices at the ends of $\ell$. 
There exist two lines $\ell_1$ and $\ell_2$, not in $T$
such that $\ell_3\in\La_{\ell_1}\cap\La_{\ell_2}$.
\end{lemma}

\begin{proof} 
Since $T$ is a tree, $T\setminus\{\ell_3\}$ necessarily contains exactly 
two connected components $T_1$ and $T_2$ (though one could consist of just 
a single vertex). 
On the other hand, since $G$ is 1PI, $G\setminus\{\ell_3\}$ must be connected.
So there must be a path in $G\setminus T$ that connects the two components
of $T\setminus\{\ell_3\}$. Since $T$ is a spanning tree, every line of
$G\setminus T$ joins two vertices of $T$, so we may alway choose the connecting
path to consist of a single line. Let $\ell_1$ be any such line. Then
$\ell_3\in\La_{\ell_1}$.
If $G\setminus\{\ell_1,\ell_3\}$ is still connected, then there must be
a second line $\ell_2\ne\ell_1$ of $G\setminus T$ that connects the 
two components of $T\setminus\{\ell_3\}$. Again $\ell_2\in\La_{\ell_1}$.
If $G\setminus\{\ell_1,\ell_3\}$ is not connected, it consists of two connected
components $G_1$ and $G_2$ with $G_1$ containing $T_1$ and $G_2$ containing
$T_2$. Each of $T_1$ and $T_2$ must contain exactly one external vertex
of $G$. So each of $G_1$ and $G_2$ must have exactly one external leg that
is also an external leg of $G$. As $\ell_1$ and $\ell_3$ are the remaining 
external legs of both $G_1$ and $G_2$, each has three external legs, which
is impossible.
\end{proof}

\begin{lemma}\label{le:overlapScale}
Let $G$ and $T$ be as in Lemma \ref{le:overlap}. Let $J$ be an assignment
of scales to the lines of $G$ such that $T\cap G^J_f$ is a connected tree for
every  fork $f\in t(G^J)$. (See step 1 of 
\cite[Proposition \ref{FS1-le:fixedrootscale}]{SFS1}.)
Let $\ell_3\in T$. Let $\ell\in G\setminus T$ connect the two
components of $T\setminus\{\ell_3\}$, then $j_\ell\le j_{\ell_3}$.
\end{lemma}
\begin{proof}
Let $G'$ be the connected component containing $\ell$ of the subgraph of 
$G$ consisting of lines having scales $j\ge j_\ell$. By hypothesis, 
$G'\cap T$ is a spanning tree for $G'$. So the linear subtree of $T$ joining 
the vertices of $\ell$ is completely contained in $G'$. But 
$\ell_3$ is a member of that line's subtree and so has scale 
$j_{\ell_3}\ge j_\ell$.
\end{proof} 
 
\subsection{The frequency derivative of the self--energy}

We show to all orders that singularities in this derivative
can only occur at the closure of the lattice generated by the  
Van Hove points. 
That the singularities really occur is shown by explicit calculations
in model cases in the following sections.   

\begin{lemma}\label{le:q0deriv}
Let $G(q)$ be the value of any two--legged 1PI graph with external momentum
$q$. Let $B$ be the closure of the set of momenta of the form $(0,\bq)$ 
with $\bq=\sum_{i=1}^n(-1)^{s_i}\tilde\bq_i$ where $n\in\bN$ and, for 
each $1\le i\le n$, $s_i\in\{0,1\}$ and $\tilde\bq_i$ is a singular point. 
Then $G(q)$ is $C^1$ with respect to $q_0$ on $\bR^3\setminus B$.
\end{lemma}

\begin{proof}
The proof is similar to that of Lemma \ref{le:spatderiv}.
Introduce scales in the standard way and denote the root scale $j$. 
Choose a spanning tree $T$ in the standard way. View $c$-- and $r$--forks as
vertices. The external momentum is always routed through the spanning tree so 
the derivative may only act on vertices and on lines of the spanning tree.
The cases in which the derivative acts on an interaction vertex or $c$--fork
are trivial. If the derivative acts on an $r$--fork, the effect on the
bound \cite[\Ref{FS1-eq:rforkA}]{SFS1}, namely a factor of 
$M^{-j_{\pi(f)}}$, is the same as the effect on the bound 
\cite[\Ref{FS1-eq:propbnd}]{SFS1} when the derivative acts on
a propagator attached to the $r$--fork. So suppose that the derivative acts 
on a line $\ell_3$ of the spanning tree that has scale $j_3$.
We know from Lemma \ref{le:overlap} in the last section that there exist 
two different lines $\ell_1$ and $\ell_2$, not in $T$ such that $\ell_3$ lies 
on the loops associated to $\ell_1$ and $\ell_2$. We also know, from Lemma
 \ref{le:overlapScale}, that the scales $j'$ of all loop momenta running 
through $\ell_3$, including the scales $j_1$ and $j_2$ of the two lines chosen, 
obey $j'\le j_3$.

\bigskip\noindent{\bf  Reduction 1:} 
It suffices to consider $j\le j_1, j_2\le j_3\le (1-\tilde\et)j$.
For the remaining terms, we simply bound
\begin{eqnarray}
\big|\sfrac{d\hfill}{dq_0} C_{j_3}(\pm q\pm\hbox{internal momenta})\big|
\le\const M^{-2j_3}\le\const M^{-j_3}M^{-(1-\tilde\et)j}
\nonumber
\end{eqnarray}
After one sums over all scales except the root scale $j$,
one ends up with 
$$
\const |j|^n M^{j}M^{-(1-\tilde\et)j}
=\const |j|^n M^{\tilde\et j}
$$
which is still summable.

\bigskip\noindent{\bf  Reduction 2:} 
Denote by $k$, $p$ and $\pm k\pm p + q'$ the momenta flowing in the
lines $\ell_1$, $\ell_2$ and $\ell_3$, respectively. Here $q'$ is plus
or minus the external momentum $q$ possibly plus or minus some some other loop 
momenta. In this reduction we prove that it suffices to consider 
$(\bk,\bp)$ with $|\bk-\tilde\bk|\le M^{\et j}$ 
and $|\bp-\tilde\bp|\le M^{\et j}$ and 
$|\pm \bk\pm \bp\pm \bq'-\tilde \bq|\le M^{\et j}$ for some singular points 
$\tilde\bk$, $\tilde\bp$ and $\tilde \bq$. 

Suppose that at least one of $\bk$, $\bp$, $\pm \bk\pm \bp\pm \bq'$ is 
farther than $M^{\et j}$ from all singular points. We can make a change
of variables (just for the purposes of computing the volume of the domain
of integration) such that $\bk$ is at least a distance $M^{\et j}$ from
all singular points. After the change of variables, the indices $j_1$, $j_2$ 
and $j_3$ are no longer ordered, but all are still between $j$ and 
$(1-\tilde\et)j$. Let $\Xi_j(\bk)$ be the characteristic function of the set 
$$
\set{\bk\in\bR^2}{|\bk-\tilde \bk|\ge M^{\et j}\hbox{ for all singular points
} \tilde \bk}
$$
We claim that
\begin{eqnarray}
\vol\set{(\bk,\bp)
&\in&\bR^4}{
 \chi_{j_1}(\bk)\Xi_{j}(\bk)\chi_{j_2}(\bp)\chi_{j_3}(\pm\bk\pm\bp\pm\bq')\ne 0}
\nonumber\\
     &\le& 
C|j_2|M^{j_1-\et j}M^{j_2}M^{\veps j-\et j}
\nonumber
\end{eqnarray}
This would constitute a volume improvement of $M^{(\veps-2\et)j}$
and would provide summability if $2\et<\veps$.
By the usual compactness arguments, it suffices to show that, for each 
fixed $\tilde\bk$, $\tilde\bp$ and $\tilde\bq$
in $\bR^2$, there are (possibly $\tilde\bk$, $\tilde\bp$, $\tilde\bq$ dependent,
but $j,j_i$ independent)  constants $c$ and $C$ such that
$$
\int_{|\bk-\tilde \bk|\le c}\hskip-22pt d^2\bk\hskip5pt
\int_{|\bp-\tilde\bp|\le c}  \hskip-20 pt d^2\bp\ 
    \chi_{j_1}(\bk)\Xi_{j}(\bk)\chi_{j_2}(\bp)\chi_{j_3}(\pm\bk\pm\bp\pm\bq')
     \le C|j_2|M^{j_1-\et j}M^{j_2}M^{\veps j-\et j}
$$
for all $\bq'$ obeying $|\bq'-\tilde \bq|\le c$. Furthermore, if 
$\tilde \bk$ or $\tilde \bp$ or $\pm\tilde \bk\pm\tilde\bp\pm\tilde \bq'$ does 
not lie on $\cF$, we can choose $c$ sufficiently small that the integral 
vanishes whenever 
$|\bq'-\tilde\bq|\le c$ and $|j_1|,|j_2|,|j_3|$ are large enough. So it 
suffices to require that $\tilde \bk$, $\tilde \bp$ and 
$\pm\tilde \bk\pm\tilde\bp\pm\tilde \bq'$ all lie on $\cF$.

If $\tilde\bk$ is not a singular point, make a change of variable to
$\rho=e(\bk)$ and an ``angular'' variable $\th$. If $\tilde \bk$ is a singular
point, the condition $\Xi_j(\bk)\ne 0$ forces $|\bk-\tilde\bk|\ge M^{\et j}$. 
We can make a change of variables such that 
$e\big(\bk(\rho,\th)\big)=\rho\th$ 
and either $|\th|\ge\const M^{\et j}$ or $|\rho|\ge\const M^{\et j}$.
Possibly exchanging the roles of $\rho$ and $\th$, we may, without loss 
of generality assume the former. In all of these cases, the 
condition $\chi_{j_1}(\bk)\ne 0$ forces $|\rho|\le \const M^{j_1-\et j}$. 
Thus
\begin{eqnarray}
&&\hskip-0.3in\int_{|\bk-\tilde \bk|\le c}\hskip-20pt d^2\bk\hskip5pt
  \int_{|\bp-\tilde\bp|\le c}  \hskip-20 pt d^2\bp\ 
    \chi_{j_1}(\bk)\Xi_{j}(\bk)\chi_{j_2}(\bp)\chi_{j_3}(\pm\bk\pm\bp\pm\bq')
\nonumber\\
&&\hskip0.2in\le
\const \int_{|\th|\le 1\atop |\rho|\le {\rm const}  M^{j_1-\et j}}
              \hskip-20pt d\rho d\th\hskip5pt
\int_{|\bp-\tilde\bp|\le c}  \hskip-20 pt d^2\bp\ 
    \chi_{j_2}(\bp)\chi_{j_3}(\pm\bk(\rho,\th)\pm\bp\pm\bq')
\nonumber\\
&&\hskip0.2in\le
\const \int_{|\th|\le 1\atop |\rho|\le {\rm const} M^{j_1-\et j}}
            \hskip-20pt d\rho d\th\hskip5pt
\int_{|\bp-\tilde\bp|\le c}  \hskip-20 pt d^2\bp\ 
    \chi_{j'}(\pm\bk(0,\th)\pm\bp\pm\bq')\  \chi_{j_2}(\bp)
\nonumber\\
&&\hskip0.2in\le
\const M^{j_1-\et j}\int_{|\th|\le 1 }\hskip-10pt
d\th\hskip5pt
\int_{|\bp-\tilde\bp|\le c}  \hskip-20 pt d^2\bp\ 
    \chi_{j'}(\pm\bk(0,\th)\pm\bp\pm\bq')\  \chi_{j_2}(\bp)
\nonumber
\end{eqnarray}
where $M^{j'}=M^{j_3}+\const M^{j_1-\et j}\le \const M^{(1-\et-\tilde\et)j}$.
Thus it suffices to prove that
$$
\int_{|\bp-\tilde\bp|\le c} d^2\bp\ \int_{|\th|\le 1 }d\th\ 
    \chi_{j'}(\pm\bk(0,\th)\pm\bp\pm\bq')\  \chi_{j_2}(\bp) 
\le C|j_2| M^{j_2}M^{\veps j-\et j}
$$
for all $\bq'$ obeying $|\bq'-\tilde \bq|\le c$.

We again apply Proposition \ref{le:length2}
with $j=j'$ and $\de= M^{(1+\veps)j_2/2}$. If we denote by $\tilde
\chi(\bp)$ the characteristic function of the set of $\bp$'s with
$$
\mu\Big(\set{-1\le \th\le 1}
            {\big|e\big(\pm\bp\pm\bq'\pm\bk(0,\th)\big)\big|\le  M^{j'}}\Big)
\ge c_1\big(\sfrac{ M^{j'}}{M^{(1+\veps)j_2/2}}\big)^{1/n_0}
$$ 
where $c_1$ is the supremum of $\sfrac{d\th}{ds}$ ($s$ is arc length),
then $\tilde \chi(\bp)$ vanishes except on a set of measure $DM^{(1+\veps)j_2}$
and
\begin{eqnarray}
&&\int d^2\bp\ \int_{|\th|\le 1 }d\th\ 
    \chi_{j'}(\pm\bp\pm\bq'\pm\bk(0,\th))\  \chi_{j_2}(\bp) 
\nonumber\\
&&\hskip0.5in\le \int d^2\bp\ \tilde \chi(\bp)\chi_{j_2}(\bp)
    \int_{|\th|\le 1 }d\th\ \chi_{j'}(\pm\bp\pm\bq'\pm\bk(0,\th))
\nonumber\\
&&\hskip1in+\int d^2\bp\ \big(1-\tilde \chi(\bp)\big)\chi_{j_2}(\bp)
    \int_{|\th|\le 1 }d\th\ 
    \chi_{j'}(\pm\bp\pm\bq'\pm\bk(0,\th))
\nonumber\\
&&\hskip0.5in\le 2\int d^2\bp\ \tilde \chi(\bp)
    +\const\int d^2\bp\ \chi_{j_2}(\bp)
    \big(\sfrac{ M^{j'}}{M^{(1+\veps)j_2/2}}\big)^{1/n_0}
\nonumber\\
&&\hskip0.5in\le\const M^{(1+\veps)j_2}
    +\const |j_2|M^{j_2}
    \big(\sfrac{ M^{j'}}{M^{(1+\veps)j_2/2}}\big)^{1/n_0}
\nonumber\\
&&\hskip0.5in\le\const M^{j_2}M^{\veps(1-\tilde\et) j}
    +\const |j_2|M^{j_2}
    \big( M^{(1-\et-\tilde\et)j-{1+\veps\over 2}j}\big)^{1/n_0}
\nonumber\\
&&\hskip0.5in\le\const M^{j_2}M^{\veps(1-\tilde\et) j}
    +\const |j_2|M^{j_2}M^{{1\over n_0}({1\over 3}-\et)j}
\nonumber\\
&&\hskip0.5in\le \const|j_2| M^{j_2}M^{\veps j-\et j}
\nonumber
\end{eqnarray}
provided $\tilde\et\le\et\le\sfrac{1}{12}$ and $\veps\le\min\big\{\sfrac{1}{6},\sfrac{1}{3n_0}\big\}$.

\medskip\noindent{\bf  End stage of proof}
The momentum flowing through the differentiated line is of the form
$\pm q\pm k_1\pm\ldots\pm k_n$ where $q$ is the external momentum and each
of the $k_i$'s is a loop momentum of a line not in the spanning tree
whose scale is no closer to zero than the scale of the differentiated line.
If $q\notin B$, then either $q_0$ is nonzero or the infimum of 
$\big|\pm \bq\pm\tilde\bk_1\pm\ldots\pm \tilde\bk_n-\tilde\bk_{n+1}\big|$, 
with the $\tilde \bk_i$'s running over all singular points, is nonzero. So 
there exists a $j_0$ such
that when the root scale obeys $j\le j_0$, either the zero component
of one of $k_1$, $\ldots$, $k_n$, $\pm q\pm k_1\pm\ldots\pm k_n$ has magnitude 
larger than $\const M^j$, in which case the corresponding covariance vanishes,
or the distance of one of $\bk_1$, $\ldots$, 
$\bk_n$, $\pm \bq\pm \bk_1\pm\ldots\pm \bk_n$ to the nearest singular point
is at least $M^{\et j}$, in which case we can apply the argument of
reduction 2. (If it is one of $\bk_1$, $\ldots$, $\bk_n$ whose distance
to the nearest singular point is at least $M^{\et j}$, we may choose as
the $\ell_1$ of Lemma \ref{le:overlap} the line initiating that loop
momentum.)
\end{proof}

\section{Singularities} 
In this section, we do a two--loop calculation for a typical case to show 
that the frequency derivative of the self--energy is indeed divergent 
in typical situations, and to calculate the second spatial derivative at 
the singular points. 

By the Morse lemma, there are coordinates $(x,y)$ such that in a neighbourhood
of the Van Hove singularity, the dispersion relation becomes 
$$
e(\bk ) = \tilde e (x,y) = x\; y .
$$
Here we consider the case of a Van Hove singularity at $k=0$ with 
$e(\bk) = \sfrac{k_1k_2}{(2\pi)^2}$. (In particular the Van Hove singularity 
is on the Fermi surface for $\mu=0$.) The nonlinearities induced by the changes 
of variables are absent in this example, and moreover, the curvature is zero 
on the Fermi surface. We rescale to $x=\sfrac{k_1}{2\pi}$, 
$y=\sfrac{k_2}{2\pi}$ and, for definiteness,  take the integration 
region for each variable to be $[-1,1]$. 

For this case, we determine the asymptotics of derivatives of the 
two--loop contribution to the self--energy as a function of $q_0$ 
for small $q_0$. We find that again, the gradient of the self--energy 
is bounded (in fact, the correction is zero in that case) but that 
the $q_0$--derivative is indeed divergent.  Power counting by standard
scales suggests that this derivative diverges at zero temperature 
like $\abs{\log q_0}^3$. However, there is a cancellation of the leading
singularity which is not seen when taking absolute values, so that 
the true behaviour is only $(\log q_0)^2$. We then also calculate 
the asymptotics of the second spatial derivative and find that it is
of the same order as the first frequency derivative. 
Finally, we do the calculation for the one--loop contributions to 
the four--point function, to compare the coefficients of different
divergences in perturbation theory. 

The physical significance of these results will be discussed in Section \ref{discussion}.

\subsection{Preparations} 
We restrict to a local potential. Since we have only considered 
short--range interactions, the potential is smooth in momentum space. 
For differentiability  questions, a momentum dependence could only 
make a difference if the potential vanished at the singular points or other
special points, so the restriction to a local potential, which is 
constant in momentum space, is not a loss of generality. 

There are two graphs contributing, one of vertex correction type and 
the other of vacuum polarization type (the graphs with 
insertion of first--order self--energy graphs have been eliminated 
by renormalization through a shift in $\mu$).
$$
\figplace{vertcor}{0 in}{0.0in}\qquad
\figplace{vacpol}{0 in}{0.0in}
$$
The latter gets a $(-1)$ from the 
fermion loop and a $2$ from the spin sum. Thus the total contribution is 
\begin{eqnarray}
\Sigma_2 (q_0,q)
&=&
-\sfrac{1}{\be^2} \sum_{\om_1,\om_2}
\sfrac{1}{(2\pi)^4}\int d^2 k_1 d^2 k_2 d^2 k_3\ \de (q -k_1 + k_2 - k_3) \;
\nonumber\\
&&\hskip1in
C(\om_1, e(k_1))\;C(\om_2, e(k_2))\;C(q_0-\om_1+\om_2, e(k_3))
\nonumber
\end{eqnarray}
with 
$$
C(\om,E) = \frac{1}{\I \om -E} .
$$
Call $E_i = e(k_i)$ and 
$$
\langle F\rangle_q = \int \dd \rh_q (\ve{k}) \; F(\ve{k})
$$
where $\ve{k} =(k_1,k_2,k_3)$ and
$\dd \rh_q (\ve{k}) = \sfrac{1}{(2\pi)^4}\int d^2k_1 d^2k_2 d^2k_3\ \de (q -k_1 + k_2 - k_3) $. The frequency summation gives
\begin{equation}\label{eq:Si2}
\Sigma_2 (q_0,q)
=
-
\left\langle
\frac{%
(\Feb (E_1) + \Beb(E_2 - E_3)) \; (\Feb(E_2) - \Feb(E_3))
}{\I q_0 + (E_2 -E_3 -E_1)}
\right\rangle_q
\end{equation}
where $\Feb (E) = (1+\E^{\be E})^{-1}$ is the Fermi function
and $\Beb (E) = (\E^{\be E}-1)^{-1}$ is the Bose function.
Since 
$$
\Beb(E_2 - E_3) \; \big[ \Feb(E_2) - \Feb(E_3) \big]
=\Feb(E_2)\big[\Feb(E_3)-1\big]
$$
the numerator is bounded in magnitude by 2.
The denominator is bounded below in magnitude by $|q_0|$, so the integrand 
is $C^\infty$ in $q_0$, for all $q_0\ne 0$ and all $\be \ge 0$. Because the 
integral is over a compact region of momenta, the same holds for the integral.  
The limit $\be \to \infty$ exists and has the same property. 
The structure of the denominator may suggest that  for it to almost
vanish requires only the combination $E_2 -E_3 -E_1$ to get small, 
but a closer look reveals that each $E_i$ has to be small: at $T=0$, the 
Fermi functions become step functions, $\Feb (E ) \to \Th (-E)$ and 
$\Beb(E) \to - \Th (-E)$, and then the factors in the numerator
imply that all summands in $E_2 -E_3 -E_1$ really have the same sign, 
i.e.\ $\abs{E_2 -E_3 -E_1} = |E_2| + |E_3| + |E_1|$,
so all $|E_i|$ must be small for the energy difference to be small. 
At finite $\be$, when the $E_i$'s are ``of the wrong sign'' the exponential 
suppression provided by the numerator compensates for the 
$|q_0|\ge\sfrac{\pi}{\be}$ in the denominator.

Because 
\begin{equation}\label{q0Exp}
\I q_0 - e(q) - \la^2 \Si_2(q_0,q)
=
\I q_0 (1 + \I \la^2 \del_0 \Si_2 (0,q)) - (e(q) + \la^2 \Si_2 (0,q)) + \ldots 
\end{equation}
and because $\del_0 \Si_2 (0,q) $ is purely imaginary, 
we are interested in 
\begin{equation}\label{Zq}
Z_2 (q) 
= 
\left(
1 - \la^2 \mbox{ Im }\del_0 \Si_2 (0,q)
\right)^{-1} \;.
\end{equation}
The value $q_0=0$ is not an allowed fermionic Matsubara frequency at $T>0$. 
We shall keep $q_0 \ne 0$ in the calculations. In discussions about temperature
dependence, we shall replace $|q_0|$ by $\pi/\be$. 

The second order contribution to $\mbox{ Im }\del_0 \Si$ is 
\begin{eqnarray}\label{Imderiv}
&&\hskip-1in\mbox{ Im } \del_0\Sigma_2 (q_0,q)
=
\Big\langle
\Phi_{q_0} (E_2 -E_3 -E_1)
\nonumber\\
&&\hskip0.5in
[\Feb (E_1) + \Beb(E_2 - E_3)] \; [\Feb(E_2) - \Feb(E_3)]
\Big\rangle_q
\end{eqnarray}
with 
\begin{equation}\label{eq:ReId}
\Phi_{q_0} (\veps) 
=
\mbox{ Re } \frac{1}{(\I q_0 + \veps)^2} 
=
\frac{\veps^2 - q_0^2}{(\veps^2+q_0^2)^2}
.
\end{equation}

\subsection{$q_0$--derivative}
The above expression \Ref{Imderiv} for $\mbox{ Im }\del_0 \Si_2 (q_0,q) $ 
shows that it is an even function of $q_0$ and that $q_0$ serves as a regulator 
so that even at zero temperature, a singularity can develop only in the limit
$q_0 \to 0$. We therefore calculate it as a function of $q_0$ at zero 
temperature. In the following, we take $q_0 > 0$. 

As $\beta \to \infty$, the Fermi function $\Feb (E) \to \Th (-E)$ and for 
$E \ne 0$, the Bose function $\Beb (E) \to - \Th (-E)$. 
In this limit, the integrand vanishes except when $E_2 E_3 <0$ and
$E_1 (E_2-E_3) < 0$. This reduces to the two cases
\begin{equation}\label{restri}
E_1>0, E_2<0, E_3 >0 \mbox{ and } E_1 < 0, E_2 > 0 , E_3 < 0.
\end{equation}
In both cases, the combination of $\Feb$,$\Beb$'s in the numerator is $-1$. 
Thus 
\begin{eqnarray}
&&\hskip-25pt\mbox{ Im }\del_0 \Si_2 (q_0,q) 
=
-
\Big\langle
\Phi_{q_0} (E_2-E_3-E_1) \; 
\Big\lbrack
\True{E_1> 0 \wedge E_2 < 0 \wedge E_3 > 0} 
\nonumber\\
&&\hskip2.5in+
\True{E_1< 0 \wedge E_2 > 0 \wedge E_3 < 0} 
\Big\rbrack
\Big\rangle_q
\nonumber
\end{eqnarray}

Recall that we are considering the case of a Van Hove point at 
$k=(x,y)=0$ for the dispersion relation $e(k) = xy$ with 
$x = k_1/2\pi $ and $y=k_2/2\pi$, so that $\frac{\dd^2 k}{(2\pi)^2} = \dd x\; \dd y$. 
Moreover we set 
$q=0$, and use the delta function to fix $E_1$ in terms of $E_2$ and 
$E_3$, so that 
$$
E_2 = xy \qquad E_3 = x'y' \qquad E_1= (x-x')(y-y') 
$$
and 
$$
\veps = E_2-E_3-E_1 = xy' - 2x'y' + x'y .
$$
Recall also that, for definiteness, we are taking the integration region 
for each variable to be $[-1,1]$. The sign conditions 
on the $E_i$ impose conditions on the variables $x, \ldots$, which are
listed in Appendix \ref{bloodysigns}, and which we use to transform the 
integration region to $[0,1]^4$. 

At $T=0$, only $n \in \cM = \{ 1,2,3,4,9,10,11,12\}$ from the table in
Appendix \ref{bloodysigns} contribute. Thus 
$$
\mbox{ Im } \del_0 \Sigma_2 (q_0,0)
=
-
\int_{[0,1]^4} \dd \rx \dd \ry \dd \rx' \dd \ry'  
\sum_{n\in \cM}
\True{\rho_n}
\Phi_{q_0} (\veps_n (\rx,\ry,\rx',\ry'))
$$
with $\veps_n$ and $\rho_n$ given in the table in Appendix \ref{bloodysigns}.
By \Ref{eq:n+8}, \Ref{eq:1-3} and \Ref{eq:1-2}, 
and because $\Phi_{q_0}$ is even in $\veps$, 
all eight terms give the same contribution. 
Hence 
\begin{equation}\label{eq:SiI}
\mbox{ Im } \del_0 \Sigma_2 (q_0,0)
= 
- 
2 I(q_0)
\end{equation}
where
\begin{equation}\label{Idef}
I(q_0) = 4 \int_0^1 \dd y \int_0^1 \dd y' \int_0^1 \dd x' \int_0^{x'} \dd x\;
\Phi_{q_0} \left((2x'-x) y' + y x'\right) .
\end{equation}

\begin{lemma}\label{lowerbound}
Let $I$ be defined as in \Ref{Idef} and $0 < q_0 < \frac12$. Then
$$
\mbox{ Im } \del_0 \Sigma_2 (q_0,0)
 = 
- 4 \log 2 \; \abs{\log{q_0}}^2 - 2 C_1 \; \abs{\log{q_0}}
+ B(q_0)
$$
where $B$ is a bounded function of $q_0$, and 
$$
C_1
=
2 (\log 2)^2
-
4
\int_0^1 \frac{\dd x}{x}\; 
\log \left(
\frac{1+2x}{1+x}
\right)
$$ 
\end{lemma}

\begin{proof}
By \Ref{eq:SiI}, it suffices to show that 
$$
I(q_0) = 
2 \log 2 \; \abs{\log{q_0}}^2 + C_1 \; \abs{\log{q_0}}
+ \tilde B(q_0)
$$
with bounded $\tilde B_0$. 
We rewrite the argument $\veps$ of $\Phi_{q_0}$ as 
$\veps = x' (2y'+y) - xy'$. We first bound the contribution of $y' \le q_0$. 
To do this, bound 
$$
\abs{\Phi_{q_0} (\veps ) } \le \frac{1}{q_0^2 + \veps^2}
$$
Use that this bound is decreasing in $\veps$ and that $\veps \ge x'y$. 
Therefore 
\begin{eqnarray}
&&
4 \int_0^1 \dd y \int_0^{q_0} \dd y' \int_0^1 \dd x' \int_0^{x'} \dd x\;
\Phi_{q_0} \left((2x'-x) y' + y x'\right) 
\nonumber\\
&&\hskip0.5in\le
4 \int_0^1 \dd y \int_0^{q_0} \dd y' \int_0^1 \dd x' \int_0^{x'} \dd x\;
\frac{1}{q_0^2 + (x'y)^2}
\nonumber\\
&&\hskip0.5in=
4 q_0 \int_0^1 \dd x' \int_0^1 \dd y \;
\frac{x'}{q_0^2 + (x'y)^2}
\nonumber\\
&&\hskip0.5in=
4 q_0 \int_0^1 \dd x' \; 
\frac{1}{q_0} \arctan \frac{x'}{q_0} 
\nonumber\\
&&\hskip0.5in\le
2 \pi .
\nonumber
\end{eqnarray}
Thus it suffices to calculate the asymptotic behaviour of 
$$
\tilde I (q_0) 
=
4 \int_0^1 \dd y \int_{q_0}^1 \dd y' \int_0^1 \dd x' \int_0^{x'} \dd x\;
\Phi_{q_0} \left(x' (2y'+y) - xy'\right)  
$$
for small $q_0 >0$. 
By \Ref{eq:ReId}
\begin{equation}\label{ibpbasis}
\Phi_{q_0} (\veps) 
=
- \frac{\del }{\del \veps} \; \frac{\veps}{q_0^2+\veps^2} .
\end{equation}
Because $\sfrac{\partial\veps}{\partial x}=-y'$, 
$$
\int_0^{x'} \dd x\; \Phi_{q_0} (\veps (x)) 
=
\frac{1}{y'}
\left\lbrack
\frac{\veps}{\veps^2+q_0^2}
\right\rbrack_{(2y'+y)x'}^{(y'+y)x'} .
$$
The integral over $x'$ can now be done, using 
$$
\int_0^1 \dd x'\; 
\frac{\al x'}{(\al x')^2+q_0^2}
=
\frac{1}{2\al} \log \left(1+ \frac{\al^2}{q_0^2}\right) .
$$
Thus 
$$
\tilde I (q_0) = \tilde I_1 - \tilde I_2
$$
with 
\begin{eqnarray}
\tilde I_j
&=&
4 \int_{q_0}^1 \frac{\dd y'}{y'} \int_0^1 \dd y  \ \  
\frac{1}{2(j y'+y)} \log \left(1+ \frac{(j y'+y)^2}{q_0^2}\right) 
\nonumber\\
&=&
2 \int_1^{\frac{1}{q_0}} \frac{\dd \eta}{\eta} 
\int_{j \eta}^{\frac{1}{q_0}+j\eta} 
\frac{\dd \xi  }{\xi} \;\log \left(1+ \xi^2 \right) 
\nonumber
\end{eqnarray}
where we have made the change of variables $\xi=\sfrac{j y'+y}{q_0}$,
$d\xi=\sfrac{1}{q_0}dy$ followed by the change of variables 
$\et=\sfrac{y'}{q_0}$, $d\et=\sfrac{1}{q_0}dy'$.
Thus 
\begin{eqnarray*}
\tilde I (q_0) 
&=&
\int_1^{{q_0^{-1}}} \frac{\dd \eta}{\eta} \;
\left(
J_{[\eta,q_0^{-1}+\eta]}  -
J_{[2\eta,q_0^{-1}+2\eta]} 
\right)\\
&=&
\int_1^{{q_0^{-1}}} \frac{\dd \eta}{\eta} \;
\left(
J_{[\eta,2\eta]}  -
J_{[q_0^{-1}+\eta,q_0^{-1}+2\eta]} 
\right)
\end{eqnarray*}
with 
$$
J_A = 2 \int_A \frac{\dd \xi  }{\xi} \;\log \left(1+ \xi^2 \right) \ge 0 
\mbox{ for } A \subset [0,\infty)
$$
We have
\begin{eqnarray}\label{eq:Jdef}
J_{[a,b]} 
=
2 \int_a^b \frac{\dd \xi  }{\xi} \;\log \left(1+ \xi^2 \right)
=
\int_{a^2}^{b^2} 
\frac{\dd t}{t} \; \log (1+t) .
\end{eqnarray}
In our case both integration intervals for $J$ are subsets of $[1,\infty)$, 
so we can expand the logarithm, to get
\begin{eqnarray}\label{eq:Jcalc}
J_{[a,b]} 
&=&
\int_{a^2}^{b^2} 
\frac{\dd t}{t} \; \Big[\log (t) + \log \Big(1+ \frac{1}{t}\Big)\Big]
\nonumber\\
&=&
\frac12 
\left\lbrack
(\log b^2)^2 - (\log a^2)^2
\right\rbrack
-
\sum_{n \ge 1}
\frac{(-1)^n}{n^2}\;
(a^{-2n} - b^{-2n})
\nonumber\\
&=&
2 
\left\lbrack
(\log b)^2 - (\log a)^2
\right\rbrack
-
\sum_{n \ge 1}
\frac{(-1)^n}{n^2}\;
(a^{-2n} - b^{-2n})
\nonumber\\
&=&
2 
\log (ab) \; 
\log \frac{b}{a}
-
\sum_{n \ge 1}
\frac{(-1)^n}{n^2}\;
(a^{-2n} - b^{-2n})
\end{eqnarray}
The final integral over $\eta$ gives, for the first term,
\begin{eqnarray}
\int_1^{{q_0^{-1}}} \frac{\dd \eta}{\eta} \;
J_{[\eta,2\eta]} 
&=& 
2 \log 2 \; |\log q_0|^2 + 2 (\log 2)^2 \; |\log q_0|
\nonumber\\
&&-
\frac12 \sum_{n \ge 1} \frac{(-1)^n}{n^3} (1-4^{-n}) (1-q_0^{2n})
\nonumber
\end{eqnarray}
with the last term analytic, and hence bounded, for $|q_0| < 1$. 
The second term gives two contributions:
\begin{eqnarray}
\int_1^{{q_0^{-1}}} \frac{\dd \eta}{\eta} \;
J_{[q_0^{-1}+\eta,q_0^{-1}+2\eta]} 
&=&
W-R
\nonumber
\end{eqnarray}
with (here $M=q_0^{-1}$)
$$
R= 
\int_1^M \frac{\dd \eta}{\eta}\; X(\eta),
\qquad
X(\eta) =
\sum_{n \ge 1}
\frac{(-1)^n}{n^2}\;
\left[
(M+\eta)^{-2n}
-
(M+2\eta)^{-2n}
\right]
$$
and 
$$
W = 
2\int_1^M \frac{\dd \eta}{\eta}\;
\left[
\left(\log (M+\al \eta)\right)^2
\right]_{\al=1}^{\al=2} .
$$
For $|q_0| \le 1$ the series for $X$ converges absolutely and gives
$$
X= \sum_{n \ge 1}
\frac{(-1)^n}{n^2}\;
q_0^{2n}
\left[
(1+\eta q_0)^{-2n}
-
(1+2\eta q_0)^{-2n}
\right]
$$
so $|X| \le \sum_{n \ge 1} q_0^{2n}/n^2 \le 2 q_0^2$,
hence for $|q_0| \le 1$,
$$
R \le
2q_0^2 |\ln q_0| \le q_0 .
$$
In $W$, scaling back to $x=q_0\eta$ gives 
\begin{eqnarray}
W
&=&
2\int_{q_0}^1 \frac{\dd x}{x}\;
\left[
\left(\log q_0^{-1} + \log (1+ \al x) \right)^2
\right]_{\al=1}^{\al=2} 
\nonumber\\
&=&
2\int_{q_0}^1 \frac{\dd x}{x}\;
\left[
2 \log q_0^{-1}\, \log (1+ \al x) + \left(\log (1+ \al x) \right)^2
\right]_{\al=1}^{\al=2} .
\nonumber
\end{eqnarray}
Because $\al x \ge 0$, $0 \le \log (1+ \alpha x) \le \alpha x$, 
so $W$ is bounded by 
a constant times $\log q_0^{-1}$. The integral of the same function 
from $0$ to $q_0$ is of order $q_0 |\log q_0|$, therefore
$$
W = 4 \log q_0^{-1} \; 
\int_0^1 \frac{\dd x}{x}\; 
\log \left(
\frac{1+2x}{1+x}
\right)
+ \tilde B(q_0)
$$
where $\tilde B$ is a bounded function. 
\end{proof}

\subsection{First spatial derivatives}
In our model case, we have $E_2=xy$ and $E_3=x'y'$. Moreover we take 
$\sfrac{\bq}{2\pi}=(\xi,\eta)$ so that, 
fixing $k_1$ by momentum conservation, 
$E_1 = (\xi + x-x')(\eta+ y-y')$. It is clear from \Ref{eq:Si2} that 
the spatial derivatives also act on the Fermi function $f_\be(E_1)$, so we 
get two terms. The derivative of the Fermi function $\Feb (x)$ is minus 
the approximate delta $\Deb (x) = \beta/(4\cosh^2(\beta x/2))$ so that
\begin{equation}\label{fderiv}
-\del_i \Sigma_2 (q_0,q) 
=
\left\langle
(\del_i e)
\left[
-\Deb(E_1) \frac{f_2-f_3}{\I q_0 + \veps} 
+
\frac{(f_1+b_{23})(f_2-f_3)}{(\I q_0 +\veps)^{2}}
\right]
\right\rangle_q
\end{equation}
where $\del_i e$ has argument $q+k_2-k_3$, $f_i=\Feb(E_i)$, 
$b_{23}=\Beb(E_2-E_3)$ and, as before, 
$
\veps = E_2 - E_3 - E_1 .
$
At $q=0$ 
$$
-\frac{\del}{\del \xi} \Sigma_2 (q_0,0)
=
S_1(\beta,q_0) + S_2 (\beta,q_0)
$$
with 
\begin{eqnarray}
S_1(\beta,q_0) 
&=&
\int_{[-1,1]^4} \dd x \dd y \dd x' \dd y' \;\; (y-y')\;
\frac{\Feb(xy)-\Feb(x'y')}{\I q_0 + \veps}
\nonumber\\
&&
\hskip1.7in(-\Deb ) ((x-x')(y-y') ) 
\nonumber\\
S_2(\beta,q_0) 
&=&
\int_{[-1,1]^4} \dd x \dd y \dd x' \dd y' \;\;
 (y-y') \;
\frac{\Feb(xy)-\Feb(x'y') }{(\I q_0 + \veps)^2}
\nonumber\\
&&
\hskip1.7in[\Feb((x-x')(y-y'))+\Beb(xy-x'y')]
\nonumber
\end{eqnarray}
Here $\epsilon = xy-x'y'-(x-x')(y-y')
= xy'+x'y-2x'y'$.

Now consider $S_1$ and apply the reflection $(x,y,x',y')\to(-x,-y,-x',-y')$ 
to the integration variables. 
The domain of integration is invariant. 
The only noninvariant factor is $y-y'$
and it changes its sign. Thus $S_1$ vanishes.
By the same argument, $S_2$ vanishes as well.
By symmetry, the same holds for the $\eta$--derivative. Thus 
$$
\nabla \Si_2 (q_0,0) = 0 .
$$

\subsection{The second spatial derivatives}

Let $\sfrac{\bq}{2\pi}=(\xi,\eta)$. 
The real part of $\Si_2 (\xi,\eta)$ is a correction to $e(\xi,\eta) = \xi\eta$.
Since $\del_\xi\del_\et e (\xi,\eta) = 1$, we calculate  the correction that
$\Si_2$ gives to that quantity. 

\begin{lemma}\label{xietderiv}
For small $q_0 \ne 0$
\begin{equation}\label{eq:soerprais}
 \lim\limits_{\beta \to \infty}
\mbox{\rm  Re }
\frac{\del^2}{\del\xi\del\eta} \Si_2 (q_0,0)
 = 
(2+4 \log 2) \; (\log |q_0|)^2 + O (|\log|q_0||)
\end{equation}
\end{lemma}

\begin{proof}
By symmetry it suffices to consider $q_0 > 0$.  
In general, the second order spatial derivatives are 
$$
- \del_i \del_j \Sigma_2 (q_0,q) 
=
Z_1^{(i,j)} + Z_2^{(i,j)} + Z_3^{(i,j)}
$$
where 
\begin{eqnarray*}
Z_1^{(i,j)}
&=&
\left\langle
\left(
(\del_i \del_j e)
+ 
2 \frac{(\del_i e) (\del_j e)}{\I q_0 + \veps}
\right)
\frac{(f_1+ b_{23})(f_2-f_3)}{(\I q_0 +\veps)^{2}}
\right\rangle_q
\\
Z_2^{(i,j)}
&=&
\left\langle
\Big((\del_i e) (\del_j e) \; (-\delta'_\beta) (E_1)
+(\del_i \del_j e)(-\Deb) (E_1)\Big)\; 
\frac{f_2-f_3}{\I q_0 + \veps}
\right\rangle_q
\\
Z_3^{(i,j)}
&=&
\left\langle
2 (\del_i e) (\del_j e)(-\Deb) (E_1)\; 
\frac{f_2-f_3}{(\I q_0 + \veps)^2}
\right\rangle_q
\end{eqnarray*}
Here $e=e(k_2-k_3 + q) = E_1$ and $\veps = E_2-E_3 -E_1$. Denote 
$$
\ze(q_0) 
=
- \frac{\del^2}{\del\xi\del\eta} \Si_2 (q_0,0,0)
$$
In the $xy$ case, the two integration momenta are denoted by 
$(x,y)$ and $(x',y')$, and $E_2 =xy$, $E_3=x'y'$, and 
$E_1 = e( k_2-k_3 + q)$ $=(\xi + x-x')(\eta+ y-y')$.
Since  $\del_\xi e\mid_{\xi=\et=0} = (y-y')$, $\del_\eta e\mid_{\xi=\et=0} = (x-x')$, 
and $\del^2_{\xi\eta} e = 1$,
$$
\ze = \ze_1+\ze_2+\ze_3
$$
with  
\begin{eqnarray*}
\ze_1
&=&
\int _{[-1,1]^4} \dd^4 X \; 
\left(
1
+ 
2 \sfrac{(x-x')(y-y')}{\I q_0 + \veps(x,y,x',y')}
\right)
\sfrac{(f_1+b_{23})(f_2-f_3)}{(\I q_0 +\veps(x,y,x',y'))^{2}}
\\
\ze_2
&=&
\int _{[-1,1]^4} \dd^4 X \; 
\left(
-E_1 \Deb' (E_1) - \Deb(E_1)
\right)
\; 
\sfrac{f_2-f_3}{\I q_0 + \veps(x,y,x',y')}
\\
\ze_3
&=&
\int _{[-1,1]^4} \dd^4 X \; 
2(-E_1) \Deb(E_1) 
\; 
\sfrac{f_2-f_3}{(\I q_0 + \veps(x,y,x',y'))^2}
\end{eqnarray*}
Here $X=(x,y,x',y')$ and $\dd^4 X = \dd x\; \dd y \; \dd x' \; \dd y'$,
and 
$$
\veps(x,y,x',y')
=
xy -x'y' -E_1
=
xy -x'y' - (x-x') (y-y') .
$$
Using the decomposition given in Appendix \ref{bloodysigns}, 
$$
\ze_1 \gtoas{\beta \to \infty}
- 2 \int_{[0,1]^4} \dd \rx \dd \ry \dd \rx' \dd \ry' \;
\sum_{j=1}^4
\left(
1 + 2 \frac{F_j}{\I q_0 + \veps_j}
\right)
\frac{1}{(\I q_0 + \veps_j)^2}
\True{\rho_j}
$$
(The limit can be taken under the integral. Decomposing according to the signs
of $x, x' \ldots$, only $j \in \{ 1,2,3,4,9,10,11,12\}$ can contribute. Using
the symmetry \Ref{eq:n+8}, one obtains the above. The minus sign arises from 
the combination of Fermi functions $f_\beta$, as discussed around \Ref{restri}.) 
We write $\zeta_1 = - ( \zeta_{11} + \zeta_{12} )$. 
The term involving the $1$ is
$$
\ze_{11}
=
2 \int_{[0,1]^4} \dd \rx \dd \ry \dd \rx' \dd \ry' \;
\sum_{i=1}^4
\frac{1}{(\I q_0 + \veps_i)^2}
\True{\rho_i}
$$
By \Ref{eq:1-3} and \Ref{eq:1-2},  
$$
\ze_{11}
=
8 \int_{[0,1]^4} \dd \rx \dd \ry \dd \rx' \dd \ry' \;
\mbox{ Re }
\frac{1}{(\I q_0 + \veps_1)^2}
\True{\rho_1} .
$$
This is the same (up to a sign) as \Ref{eq:SiI}, \Ref{Idef}, hence 
$\ze_{11} = 4 \log 2 \; (\log|q_0|)^2  + O(\abs{\log |q_0|}) $. 

The term $\zeta_{12}$ involving the $F_j$ gives another contribution, 
$$
\zeta_{12}
=
2 \int_{[0,1]^4} \dd \rx \dd \ry \dd \rx' \dd \ry' \;
\sum_{j=1}^4
 2 \frac{F_j}{(\I q_0 + \veps_j)^3}
\True{\rho_j}
$$
The summand for $j=2$ gives the same integral as that for $j=1$
because the integrand is related by the exchange $\rx \leftrightarrow \ry$
and $\rx' \leftrightarrow \ry'$. Ditto for $j=4$ and $j=3$. 
Thus 
$$
\zeta_{12}
=
8 \int_{[0,1]^4} \dd \rx \dd \ry \dd \rx' \dd \ry' \;
\sum_{j=1,3}
\frac{F_j}{(\I q_0 + \veps_j)^3}
\True{\rho_j}
$$
Because $F_1=-F_3$ and $\veps_1=-\veps_3$, and because 
$\rho_1 \Leftrightarrow \rho_3$,  
$$
\sum_{j=1,3} \frac{F_j}{(\I q_0 + \veps_j)^3}
\True{\rho_j}
=
\True{\rho_1} \; F_1 \;
2 \mbox{ Re } 
\frac{1}{(\I q_0 + \veps_1)^3} .
$$
Thus 
$$
\zeta_{12}
=
16 \mbox{ Re }
\int_0^1 \dd \ry \int_0^1\dd \ry' \int_0^1\dd \rx' \int_0^{\rx'}\dd \rx \;
\frac{(\rx-\rx')(\ry+\ry')}{(\I q_0 + \rx'(\ry+\ry') + \ry'(\rx'-\rx))^3}
$$
Let $b=\rx'(\ry+\ry')$, change integration variables from $\rx$ to
$u = \rx'-\rx \in [0,\rx']$, and use
$$
\int_0^{\rx'} \dd u \; 
\frac{u}{(\I q_0 + b + \ry' u)^3}
=
\frac{(\rx')^2}{2(\I q_0 + b + \rx'\ry')^2 (\I q_0 + b)}
$$ 
Renaming to $z=\rx'$, we have
$$
\zeta_{12}
=
- 8 
\int_0^1 \dd \ry \int_0^1\dd \ry' \int_0^1\dd z
\mbox{ Re }
\frac{(\ry+\ry') \; z^2}{(\I q_0 + z (\ry+\ry'))\; (\I q_0 + z (\ry+ 2\ry'))^2}
$$
The bound 
$$
\abs{
\frac{(\ry+\ry') \; z^2}{(\I q_0 + z (\ry+\ry'))\; (\I q_0 + z (\ry+ 2\ry'))^2}
}
\le
\frac{z}{q_0^2 + z^2 (\ry+ 2\ry')^2}
$$
for the integrand implies that
$$
\abs{\zeta_{12}} 
\le 
4
\int_0^1 \dd \ry \int_0^1\dd \ry' 
\frac{1}{(\ry+2\ry')^2} \;
\ln \left(
1+ \frac{(\ry+2\ry')^2}{q_0^2}
\right)
$$
which shows that $\abs{\zeta_{12}} \le \const [\log |q_0|]^2 $ 
for small $|q_0|$. 

We now calculate the coefficient of the $\log^2$. 
First rewrite 
$$
\frac{(\ry+\ry') \; z^2}{(\I q_0 + z (\ry+\ry'))\; (\I q_0 + z (\ry+ 2\ry'))^2}
=
\frac{1}{(\ry+2\ry')^2} \;
\frac{z^2}{(z+\I A)(z+ \I B)^2}
$$
with 
$$
A = \frac{q_0}{\ry+\ry'}
\qquad
B = \frac{q_0}{\ry+2\ry'} .
$$
Partial fractions give
$$
\frac{z^2}{(z+\I A)(z+\I B)^2}
=
\frac{\alpha}{z+\I A}
+
\frac{\beta}{z+ \I B}
+
\I \; 
\frac{b_2}{(z+ \I B )^2}
$$
with 
\begin{eqnarray*}
\alpha
&=&
\frac{A^2}{(B-A)^2}
=
\left(
\frac{\ry+2\ry'}{\ry'}
\right)^2
\\
\beta
&=&
\frac{B(B-2A)}{(B-A)^2}
=
1 
-
\left(
\frac{\ry+2\ry'}{\ry'}
\right)^2
\\
b_2 
&=&
\frac{B^2}{A-B} 
=
B \; \frac{\ry+\ry'}{\ry'}
.
\end{eqnarray*}
$\alpha, \beta$ and $b_2$ are real, so we need
$$
\mbox{ Re } 
\int_0^1 \frac{\dd z}{z + \I A}
=
\frac12 \ln (1 + A^{-2})
$$
and
$$
\mbox{ Re } 
\I b_2
\int_0^1 \frac{\dd z}{(z + \I B)^2}
=
\frac{b_2}{B(1+B^2)} .
$$
With this, we have
\begin{eqnarray*}
\zeta_{12} 
&=&
- 4
\int_0^1 \dd \ry \int_0^1\dd \ry'\  
\frac{1}{(\ry+2\ry')^2} \;
\nonumber\\
&&\qquad\qquad\qquad
\left\lbrack
\alpha \ln (1 + A^{-2})
+
\beta \ln (1 + B^{-2})
+ 
\frac{2 b_2}{B(1+B^2)} 
\right\rbrack.
\end{eqnarray*}
Collecting terms and renaming $\ry'=\eta$ gives
\begin{eqnarray*}
&&\zeta_{12} 
=
- 4
\int_0^1 \dd \ry \int_0^1\dd \eta \ \ 
\Big[
\frac{1}{\eta^2} \; \ln \frac{q_0^2 + (\ry+\eta)^2}{q_0^2 + (\ry+2\eta)^2}
\nonumber\\
&&\hskip1.7in+
2 \frac{\ry+\eta}{\eta} \; \frac{1}{q_0^2 + (\ry+2\eta)^2}
\nonumber\\
&&\hskip1.7in+
\frac{1}{(\ry+2\eta)^2} \;
\ln 
\left(
1+ \frac{(\ry+2\eta)^2}{q_0^2}
\right)
\Big]
\end{eqnarray*}
Although the first summands individually contain nonintegrable singularities 
at $\eta =0$, these singularities cancel in the sum. 
A convenient way to implement this is to use that 
\begin{eqnarray}\label{eieiei}
\frac{1}{\eta^2} \; \ln \frac{q_0^2 + (\ry+\eta)^2}{q_0^2 + (\ry+2\eta)^2}
&+&
2 \frac{\ry}{\eta} \; \frac{1}{q_0^2 + (\ry+2\eta)^2}
\nonumber\\
&=&
\frac{\del}{\del \eta}
\left\lbrack
- \frac{1}{\eta}\;
\ln \frac{q_0^2 + (\ry+\eta)^2}{q_0^2 + (\ry+2\eta)^2}
\right\rbrack
\nonumber\\
&&+
\frac{1}{\eta}
\left\lbrack
\frac{2(\ry+\eta)}{q_0^2+(\ry+\eta)^2}
-
\frac{2(\ry+4\eta)}{q_0^2+(\ry+2\eta)^2}
\right\rbrack .
\end{eqnarray}
Moreover
\begin{eqnarray*}
&&
\frac{1}{(\ry+2\eta)^2} \;
\ln 
\left(
1+ \frac{(\ry+2\eta)^2}{q_0^2}
\right)
\nonumber\\
&&\hskip0.5in=
\frac12 \;
\frac{\del}{\del \eta}
\left\lbrack
- \frac{1}{y+2\eta}\;
\ln
\left(
1+ \frac{(\ry+2\eta)^2}{q_0^2}
\right)
\right\rbrack
+ \frac{2}{q_0^2 + (\ry+2\eta)^2}
\end{eqnarray*}
Thus 
\begin{eqnarray*}
&&\hskip-13pt\zeta_{12} 
=
- 4
\int_0^1 \dd \ry 
\left\lbrack
- \frac{1}{\eta}\;
\ln \frac{q_0^2 + (\ry+\eta)^2}{q_0^2 + (\ry+2\eta)^2}
- \frac12
\frac{1}{y+2\eta}\;
\ln
\left(
1+ \frac{(\ry+2\eta)^2}{q_0^2}
\right)
\right\rbrack_{\eta=0}^{\eta=1}
\nonumber\\
&&\hskip0.25in-
4
\int_0^1 \dd \ry 
\int_0^1\dd \eta 
\;
\frac{2}{\eta}
\left\lbrack
\frac{\ry+\eta}{q_0^2+(\ry+\eta)^2}
-
\frac{\ry+2\eta}{q_0^2+(\ry+2\eta)^2}
\right\rbrack
\end{eqnarray*}
Evaluation at $\eta =1$ gives one bounded term and one term of order
$\log |q_0|$. The terms at $\eta =0$ give
$$
- 4
\int_0^1 \dd \ry 
\left\lbrack
\frac{2\ry}{q_0^2+\ry^2}
-
\frac{1}{2\ry} \ln \left(1+\frac{\ry^2}{q_0^2}\right)
\right\rbrack
$$
The first summand integrates to $O(\ln |q_0|)$. 
By \Ref{eq:Jcalc}, the second term gives 
$$
2 (\ln |q_0|)^2 + \mbox{ bounded }
$$
The remaining integrals are 
\begin{eqnarray*}
&&- 8
\int_0^1 \dd \ry 
\int_0^1\dd \eta 
\;
\left\lbrack
\frac{1}{q_0^2+(\ry+\eta)^2}
-
\frac{2}{q_0^2+(\ry+2\eta)^2}
\right\rbrack 
\\
&&= 
8 \int_0^1 \dd \ry 
\int_1^2\dd \eta 
\;
\frac{1}{q_0^2+(\ry+\eta)^2}
\le 8
\end{eqnarray*}
since the integrand is bounded by $1$, and 
$$
-4
\int_0^1 \dd \ry 
\int_0^1\dd \eta 
\ 
\frac{2}{\eta}
\left\lbrack
\frac{\ry}{q_0^2+(\ry+\eta)^2}
-
\frac{\ry}{q_0^2+(\ry+2\eta)^2}
\right\rbrack
$$
The integrand is bounded by a constant times $\frac{1}{q_0^2 + \ry^2+\eta^2}$
so the integral is of order $\log |q_0|$. 
Thus $\zeta_{12} = 2 (\log |q_0|)^2 + $ bounded terms, so that
$\zeta_1 = - (\zeta_{11} + \zeta_{12}) = - (4 \log 2 + 2)  (\log |q_0|)^2 +$
less singular terms. 

We now show that $\ze_2$ and $\ze_3$ vanish 
in the limit $\be \to \infty$. 
Consider $\ze_2$ first. Let $G_\be(E_1) = -E_1 \Deb' (E_1) - \Deb(E_1)$.
(Note that $G_\be(E)= \be G_1 (\be E)$.)
At $\xi = \et =0$,  $E_1 = (x-x') (y-y')$
is invariant under the exchange $(x,y) \leftrightarrow (x',y')$, so
\begin{eqnarray}
\ze_2
=
\int _{[-1,1]^4}\hskip-10pt \dd^4 X \; 
G_\be(E_1) \Feb (xy)
\left(
\frac{1}{\I q_0 + \veps(x,y,x',y')}
-
\frac{1}{\I q_0 + \veps(x',y',x,y)}
\right)
\nonumber
\end{eqnarray}
Because $\veps (x,y,x',y') = xy'+x'y - 2 x'y'$, 
$$
 \veps(x',y',x,y) - \veps(x,y,x',y')
=
2 (x'y' - xy) 
=
2 [ x (y'-y) + (x'-x) y']
$$
Hence
$$
\ze_2
=
2
\int _{[-1,1]^4} \dd^4 X \; 
G_\be(E_1)\; \Feb (xy)
\;
\frac{ x (y'-y) + (x'-x) y'}%
{(\I q_0 + \veps(x,y,x',y')) \; (\I q_0 + \veps(x',y',x,y)) }
$$
Consider first 
$$
T_1
=
2
\int _{[-1,1]^4} \dd^4 X \; 
G_\be(E_1)\; \Feb (xy)
\;
\frac{x (y'-y)}%
{(\I q_0 + \veps(x,y,x',y')) \; (\I q_0 + \veps(x',y',x,y)) }
$$
Because $[-1,1]^4$ and the integrand are invariant under the reflection 
$R$ which maps $(x,y,x',y')$ to $(-x,-y,-x',-y')$ and 
$$
R\{ (x,y,x',y')\ |\ y \ge y'\} = \{ (x,y,x',y')\ |\  y \le y'\}, 
$$
we can put in a factor $2 \True{y > y'}$. 
Changing variables from $x'$ to $E_1$, so that
$$
x' = x- \frac{E_1}{y-y'}\ ,
$$
gives
\begin{eqnarray}
T_1 
&=&
4
\int_{[-1,1]^3} \dd x \dd y \dd y' \;  \True{ y - y' > 0} 
 (-x)\;\Feb (xy)
\nonumber\\
&&\hskip0.5in
\int_{(x-1)(y-y')}^{(x+1)(y-y')} \dd E_1
\;
\frac{G_\be(E_1) }{(\I q_0 + \veps(x,y,x',y')) \; (\I q_0 + \veps(x',y',x,y)) }
\nonumber
\end{eqnarray}
Change variables one last time, from $E_1 $ to $u=\beta E_1$, to get
\begin{eqnarray}
T_1 
&=&
4\int_{[-1,1]^3} \dd x \dd y \dd y' \; \int_{\bR} \dd u
\nonumber\\
&&\hskip0.5in
 \True{ y - y' > 0} 
\True{\be (x-1)(y-y') \le u \le \be (x+1)(y-y')}
\nonumber\\
&&\hskip0.5in
 (-x) \Feb (xy)
\;
\frac{G_1(u) }{(\I q_0 + \veps(x,y,x',y')) \; (\I q_0 + \veps(x',y',x,y)) }
\nonumber
\end{eqnarray}
where $x' = x- \frac{u}{\be (y-y')}\rightarrow x$ as $\be\rightarrow\infty$. 
The integrand is bounded in magnitude by the $L^1$ function
$$
(x,y,y',u) \mapsto G_1(u) \frac{1}{q_0^2} 
$$
and it converges almost everywhere (namely, for $x\ne  \pm 1$, $xy \ne 0$) to 
$$
 \True{ y - y' > 0}  x  \; \Th (-xy) 
\frac{1}{q_0^2 + x^2 (y-y')^2}
\big[(- u) \de_1' (u)-\de_1(u)\big] .
$$
By dominated convergence, the limit $\be \to \infty$ can be taken 
under  the integral. The limiting range of integration for $u$ is $\bR$. 
By the fundamental theorem of calculus
$$
\int_\bR  \dd u \; G_1(u) 
=\int_\bR \dd u \;\sfrac{d\hfill}{du}\big[-u\de_1(u)\big]= 0,
$$
so $T_1$ vanishes as $\be \to \infty$. 
The calculation of the limit $\be \to \infty$ of 
$$
T_2
=
2
\int _{[-1,1]^4} \dd^4 X \; 
G_\be(E_1)\; \Feb (xy)
\;
\frac{(x'-x) y'}%
{(\I q_0 + \veps(x,y,x',y')) \; (\I q_0 + \veps(x',y',x,y)) }
$$
is similar and gives 0 as well.  Thus $\zeta_2 \to 0$ as $\beta \to \infty$. 
By the same arguments
$$
\ze_3
=
\int _{[-1,1]^4} \dd^4 X \; 
2(-E_1)\Deb (E_1) \; \Feb(xy)
\frac{2\I q_0 (\tilde\veps -\veps) + (\tilde \veps - \veps) (\tilde \veps + \veps)}%
{(\I q_0 + \veps)^2 \; (\I q_0 + \tilde \veps)^2 }
$$
where $\veps = \veps(x,y,x',y')$ and $\tilde \veps = \veps(x',y',x,y)$.
After the same limiting argument as before, the $u$--integral is now
$$
\int_\bR u \de_1 (u)\; \dd u = 0
$$
because the integrand is odd. Thus $\ze_3 \to 0$ as $\be \to \infty$, too.
\end{proof}

\begin{lemma}
The real part of the second derivative of the self--energy
with respect to $\xi$ grows at most logarithmically as $q_0 \to 0$: 
there are constants $A$ and $B$ such that for all $|q_0| < 1$
$$
\abs{
\mbox{ Re }
\frac{\del^2}{\del \xi^2} \Sigma_2  (q_0,0) }
\le
A + B \log \frac{1}{|q_0|}
$$
The same holds by symmetry for the second derivative with respect to $\eta$. 
\end{lemma}

\begin{proof}
We need to bound 
$$
\del^2_\xi
\left\langle
\frac{(\Feb(E_1) + \Beb(E_2-E_3)) \; (\Feb(E_2) - \Feb(E_3))}%
{\I q_0 + \veps}
\right\rangle
$$
at $\xi=\et=0$ where all $\xi$--dependence is in $E_1 = (\xi + x-x') (\et+y-y')$
and in $\veps = E_2 - E_3 - E_1=xy'+x'y-2x'y'$.  We proceed as for the mixed 
derivative $\del^2/\del \xi \del \eta$, but now some terms are different because
$\del^2_\xi E_1 = 0$. We obtain 
$$
\frac{\del^2 \Sigma_2}{\del\xi^2} (q_0,0,0)
=
\bX_1 + \bX_2 + \bX_3
$$
with
\begin{eqnarray}
\bX_1 
&=&
2
\left\langle
(y-y')^2 \; 
\frac{(\Feb(E_1) + \Beb(E_2-E_3)) \; (\Feb(E_2) - \Feb(E_3))}
{(\I q_0 + \veps)^3}
\right\rangle
\nonumber\\
\bX_2
&=&
\left\langle
\frac{\Deb'(E_1) \; (y-y')^2}{\I q_0 + \veps} \;
(\Feb(E_2) - \Feb(E_3))
\right\rangle
\nonumber\\
\bX_3
&=&
2
\left\langle
\frac{-\Deb(E_1) \; (y-y')^2}{(\I q_0 + \veps)^2} \;
(\Feb(E_2) - \Feb(E_3))
\right\rangle
\nonumber
\end{eqnarray}
We first calculate the zero--temperature limit of $\bX_1$. 
Using the notations of Appendix \ref{bloodysigns}, 
$$
\bX_1^{(0)}
=
\lim\limits_{\beta\to\infty} \bX_1
=
- 4 \int_{[0,1]^4} \dd \rx \dd \rx' \dd \ry \dd \ry' \;
\sum_{n=1}^4
\frac{D_n^2}{(\I q_0 + \veps_n)^3} \;
\True{\rho_n}
$$
We use that $D_1^2=D_3^2$ and $\veps_3 = -\veps_1$,
to combine $n=1$ and $n=3$ in one term, and $n=2$ and $n=4$ 
in another, so that $\bX_1^{(0)} = \bX_{1,1}^{(0)} + \bX_{1,2}^{(0)}$
with
\begin{eqnarray}
\bX_{1,1}^{(0)}
&=&
- 8 \I
\int_0^1 \dd\ry \int_0^1 \dd \ry' \;(\ry'+\ry)^2 \;
\int_0^1 \dd\rx' \int_0^{\rx'} \dd\rx\;
\mbox{Im } \sfrac{1}{(\I q_0 + \rx' \ry +\ry' (2\rx'-\rx))^3}
\nonumber\\
\bX_{1,2}^{(0)}
&=&
- 8 \I
\int_0^1 \dd\ry' \int_0^{\ry'} \dd \ry \;(\ry'-\ry)^2 \;
\int_0^1 \dd\rx' \int_0^1 \dd\rx\;
\mbox{Im } \sfrac{1}{(\I q_0 + \rx\ry'+\rx'(2\ry'-\ry))^3}
\nonumber
\end{eqnarray}
Thus  $\bX_1$ does not contribute to Re $\del^2_\xi \Sigma_2$. 

We now show that $\bX_3$ is imaginary in the limit $\beta \to \infty$, 
because after a rewriting of terms, the difference $\Feb(E_2) - \Feb(E_3)$ 
effectively implies taking the imaginary part. 
Because $(y-y')^2 \Deb(E_1) $ is invariant under $(x,y) \leftrightarrow (x',y')$
(and denoting $\tilde\veps (x,y,x',y') = \veps(x',y',x,y)$)
\begin{eqnarray}
\bX_3 
&=&
-2
\int \dd^4 X \; 
(y-y')^2 \; \Deb (E_1) \; \Feb(xy) \;
\left[
\frac{1}{(\I q_0 + \veps)^2}
-
\frac{1}{(\I q_0 + \tilde\veps)^2}
\right]
\nonumber\\
&=&
-4 \int \dd y \int \dd y' \True{y-y' > 0}
\int \dd x \; \Feb (xy) \; (y-y')
\nonumber\\
&&\hskip1cm
\int \dd x' \; 
(y-y') \Deb(E_1) 
\left[
\frac{1}{(\I q_0 + \veps)^2}
-
\frac{1}{(\I q_0 + \tilde\veps)^2}
\right]
\nonumber
\end{eqnarray}
For the second equality, we used invariance under the reflection 
$(x,y,x',y')\rightarrow (-x,-y,-x',-y')$.
The convergence argument used in the analysis of the $T_1$ contribution to 
$\del^2_{\xi\eta} \Sigma_2$  can be summarized in the following Lemma.

\begin{lemma}\label{1001}
Let $\beta_0 \ge 0$ and $F : [\beta_0,\infty)\times[-1,1]^4 \to \bC$ be bounded, and 
$$
\lim\limits_{\beta \to \infty} F(\beta , x,y,x-\sfrac{u}{\be(y-y')},y' ) 
= f(x,y,y')
$$ 
a.e. in $(x,y,u,y')$. 
Let $ E_1 = (x-x') (y-y')$. Then 
\begin{eqnarray}
&&
\lim\limits_{\beta \to \infty} 
\int F (\beta, X) \; 
\Deb (E_1) \; (y-y') \;
\True{y>y'}
\dd^4 X
\nonumber\\
&=&
\int \dd y \int \dd y' \True{y>y'} \int \dd x \; 
f(x,y,y')\nonumber
\end{eqnarray}
\end{lemma}

\bigskip\noindent
Applying Lemma \ref{1001}, and using that $\veps (x,y,x,y') = x (y-y') = -\tilde\veps (x,y,x'y')$, 
we get
\begin{eqnarray}
\lim\limits_{\beta\to\infty} \bX_3 
&=&
-4 
\int\limits_{-1}^{1}\dd y \int\limits_{-1}^{y} \dd y'  \; (y-y') 
\int\limits_{-1}^{1} \dd x \; \Theta_{\frac12} (- xy) 
\; 
\cI (x,y,y') 
\nonumber
\end{eqnarray}
where $\Theta_{\frac12} (x) = \lim\limits_{\beta \to \infty} \Feb (-x)$
is the Heaviside function, except that $\Theta_{\frac12} (0)= 1/2$,
\begin{eqnarray*}
\cI (x,y,y') 
&=&
\left[
\frac{1}{(\I q_0 + x (y-y'))^2}
-
\frac{1}{(\I q_0 - x (y-y'))^2}
\right]\\
&=&
2 \I 
\mbox{ Im }
\frac{1}{(\I q_0 + x (y-y'))^2} .
\end{eqnarray*}
Thus, in the limit $\beta \to \infty$, $\bX_3$ does not contribute to 
the real part of $\del^2_\xi \Sigma_2 $ either. 

\bigskip\noindent
It remains to bound $\bX_2$. Here we integrate by parts, 
to remove the derivative from the approximate delta function, 
and then take $\beta \to \infty$. 

We first do the standard rewriting
\begin{eqnarray}
\bX_2
&=&
2 
\int \dd^4 X \True{y-y' > 0} (y-y')^2
\nonumber\\
&&\hskip1cm
\Feb (xy)\;
\Deb'((x-x')(y-y') ) \; 
\left[
\frac{1}{\I q_0 + \veps} 
-
\frac{1}{\I q_0 + \tilde\veps} 
\right]
\nonumber
\end{eqnarray}
We integrate by parts in $x$, using that 
$
(y-y') \Deb'((x-x')(y-y') )
=
\frac{\del}{\del x} \Deb ((x-x')(y-y') )
$. 
This gives the boundary term 
\begin{eqnarray*}
B_{1/\beta}
&=&
2 \sum_{x = \pm 1} x \;
\int \dd y \int \dd y' \True{y-y' > 0} \int \dd x' \; \Feb (xy)
\\
&&\hskip1cm
(y-y') \Deb((x-x')(y-y') ) \; 
\left[
\frac{1}{\I q_0 + \veps} 
-
\frac{1}{\I q_0 + \tilde\veps} 
\right]
\end{eqnarray*}
and two integral terms, namely 
\begin{eqnarray*}
I_1
&=&
2
\int \dd y \int \dd y' \True{y-y' > 0} \int \dd x \int \dd x' \; 
 y \Deb (xy)
\\
&&\hskip1cm
(y-y') \Deb((x-x')(y-y') ) \; 
\left[
\frac{1}{\I q_0 + \veps} 
-
\frac{1}{\I q_0 + \tilde\veps} 
\right]
\end{eqnarray*}
and 
\begin{eqnarray*}
I_2
&=&
 2
\int \dd y \int \dd y' \True{y-y' > 0} \int \dd x \int \dd x' \; 
\Feb (xy)
\\
&&\hskip1cm
(y-y') \Deb((x-x')(y-y') ) \; 
\left[
\frac{y'}{(\I q_0 + \veps)^2} 
-
\frac{y'-2y}{(\I q_0 + \tilde\veps)^2} 
\right]
\end{eqnarray*}
An obvious variant of Lemma \ref{1001}, where $x$ is summed over $\pm 1$ 
instead of integrated, applies to the boundary term and gives 
\begin{eqnarray*}
B_0
=
\lim\limits_{\beta \to \infty} B_{1/\beta}
&=&
2 \sum_{x = \pm 1} x \;
\int\limits_{-1}^{1} \dd y \int\limits_{-1}^{y} \dd y' \;
\Theta_{\frac12} (-xy)
\\
&&\hskip1cm
\left[
\frac{1}{\I q_0 + x(y-y')} 
-
\frac{1}{\I q_0 - x(y-y')} 
\right]
\end{eqnarray*}
Thus $B_0$ is real. Because $\Theta_{\frac12} (y) + \Theta_{\frac12} (-y) =1$, 
we get
\begin{eqnarray*}
B_0
&=&
4 
\int\limits_{-1}^{1} \dd y \int\limits_{-1}^{y} \dd y' \; 
\frac{y-y'}{q_0^2 + (y-y')^2}
\end{eqnarray*}
With $z=y-y'$, we have
\begin{eqnarray*}
B_0 
&=&
4
\int\limits_{-1}^{1} \dd y 
\int\limits_0^{1+y} \dd z \; \frac{z}{q_0^2 + z^2}
=
4
\int\limits_{-1}^{1} \dd y \;
\frac 12 
\ln (q_0^2 + z^2 ) \Big\vert_{0}^{1+y}
\nonumber\\
&=&
2
\int\limits_{-1}^{1} \dd y \;
\ln \left( 1 + \frac{(1+y)^2}{q_0^2} \right)
=
2 \int\limits_0^2 \dd \eta \;
\ln \left( 1 + \frac{\eta^2}{q_0^2} \right)
\nonumber\\
&=&
O(\log |q_0|) .
\end{eqnarray*}
In $I_1$, we change variables from $(x',x)$ to $(u,v)$ where
$u = \beta (x-x') (y-y')$ and $v=\beta x|y|$ and get
\begin{eqnarray}
I_1 
&=&
2\int\limits_{-1}^{1} \dd y \int\limits_{-1}^{y} \dd y'   \mbox{ sgn} (y) \;
\int\limits_{-\beta |y|}^{\beta |y|} \dd v\; \delta_1 (v)
\int\limits_{(\frac{v}{|y|}-\beta)(y-y')}^{(\frac{v}{|y|}+\beta) (y-y')}\dd u\; \delta_1 (u)
\nonumber\\
&&\hskip0.3cm
\left[
\frac{1}{\I q_0 + \veps(\frac{v}{\beta |y|}, y, \frac{v}{\beta |y|} - \frac{u}{\beta(y-y')},y') } 
-
\frac{1}{\I q_0 + \veps(\frac{v}{\beta |y|} - \frac{u}{\beta(y-y')},y',\frac{v}{\beta |y|}, y)} 
\right]
\nonumber
\end{eqnarray}
The integral converges in the limit $\beta \to \infty$ by dominated convergence. 
The last factor vanishes in that limit, so $I_1 \to 0$ as $\beta \to \infty$.

Finally, Lemma \ref{1001} implies that  $I_2^{(0)} = \lim_{\beta\to\infty} I_2$
exists and equals
\begin{eqnarray*}
I_2^{(0)}
&=&
2
\int\limits_{-1}^{1} \dd y \int\limits_{-1}^{y}\dd y' \int\limits_{-1}^{1} \dd x \; \Theta_{\frac12}(-xy)
\nonumber\\
&&\hskip1cm
\left[
\frac{y'}{(\I q_0 + x(y-y'))^2} 
-
\frac{y'-2y}{(\I q_0 - x(y-y'))^2} 
\right]
\end{eqnarray*}
In the real part, the terms with $y'$ in the numerator cancel, 
so that 
\begin{eqnarray}
\mbox{Re } I_2^{(0)}
&=&
2
\int\limits_{-1}^1\dd y \int\limits_{-1}^y \dd y' \int\limits_{-1}^1 \dd x \; \Theta_{\frac12}(-xy) 
\mbox{ Re }
\frac{2y}{(\I q_0 - x (y-y'))^2}
\nonumber
\\
&=&
2
\int\limits_{-1}^1\dd y  \int\limits_{-1}^1 \dd x \; 
\Theta_{\frac12}(-xy) 
\mbox{ Re } 
\int_{0}^{1+y} \dd z
\frac{2y}{(\I q_0 - x z)^2}
\nonumber\\
&=&
-4 \int\limits_{-1}^1 y \;\dd y  \int\limits_{-1}^1 \dd x \; 
\Theta_{\frac12}(-xy) 
\frac{1+y}{q_0^2+x^2(1+y)^2}
\nonumber\\
&=&
-4 \int\limits_{-1}^1 y \;\dd y  \int\limits_{-(1+y)}^{1+y} \dd v \; 
\Theta_{\frac12}
\left(
- \frac{v y }{1+ y}
\right) \;
\frac{1}{q_0^2+v^2}
\nonumber\\
&=&
-4
\int\limits_{0}^1 y \;\dd y  \int\limits_{1-y}^{1+y} 
\frac{\dd v }{q_0^2+v^2}
\nonumber
\end{eqnarray}
The contribution from the region $v \ge 1$ is obviously bounded. 
The remaining integral is 
\begin{eqnarray*}
\int\limits_{0}^1 y \;\dd y  \int\limits_{1-y}^{1} 
\frac{\dd v }{q_0^2+v^2}
&=&
\int\limits_{0}^1 
\frac{\dd v }{q_0^2+v^2}
\int\limits_{1-v}^{1} y \;\dd y 
\nonumber\\
&=&
\int\limits_{0}^1 
\frac{\dd v }{q_0^2+v^2}
\left(v - \frac{v^2}{2}\right)
=
O(|\log |q_0|\,|)
\end{eqnarray*}
\end{proof}

\subsection{One--loop integrals for the xy case}
For the discussion in Section \ref{discussion}, it is useful 
to calculate the lowest order contributions to the four--point function, 
the so--called bubble integrals. Again we restrict to the $xy$ 
type singularity. Because the fermionic bubble integrals
are not continuous at zero temperature, it is best to calculate
them by setting $q_0=0$ first, then letting the spatial part $q$ tend to $0$, 
all at a fixed inverse temperature $\beta$, and then calculate
the asymptotics as $\beta \to \infty$. 
$$
\figplace{phbubble}{0 in}{0.0in}\qquad\qquad\qquad
\figplace{ppbubble}{0 in}{0.0in}
$$

\subsubsection{The particle--hole bubble}
Write $x'=x+\xi$, $y'=y+\et$.
 The bubble is 
\begin{eqnarray}\label{Bph}
B_{\rm ph}(q_0,\xi,\et) 
&=&
\int_{[-1,1]^2} \dd x \dd y \; 
\frac{1}{\be} \sum_\omega
\frac{1}{\I \omega - xy}
\frac{1}{\I (q_0 +\omega) - x'y'}
\nonumber\\
&=&
\int_{[-1,1]^2} \dd x \dd y \; 
\frac{\Feb (xy) - \Feb(x'y')}{\I q_0 + xy - x'y'}
\nonumber\\
&=&
\int_0^1 \dd t
\int_{[-1,1]^2} \dd x \dd y \; 
\frac{xy-x'y'}{\I q_0 + xy - x'y'}
(-\Deb) (t xy  + (1-t) x'y')
\nonumber\\
\end{eqnarray}
Obviously, at $q_0 \ne 0$, $B_{\rm ph}(q_0,\xi,\et) \to 0$ as $(\xi,\et) \to 0$. 
So set $q_0 =0$, i.e.\ consider
$$
B_{\rm ph}^{(0)} = \lim\limits_{(\xi,\et) \to 0} B_{\rm ph}(0,\xi,\et) .
$$
\begin{lemma}
The large $\beta $ asymptotics of $B_{\rm ph}^{(0)}$ is
$$
B_{\rm ph}^{(0)} =
-2 \ln \be  + 2 K + O(\E^{-\be})
$$
where $K = \int_0^\infty \frac{\dd u}{2 \cosh^2 \frac{u}{2}} \ln u$.
\end{lemma}

\begin{proof}
By \Ref{Bph},
\begin{eqnarray}
B_{\rm ph}^{(0)}
&=&
\int_{[-1,1]^2} \dd x \dd y \; 
(-\Deb) (xy)
=
4
\int_{[0,1]^2} \dd x \dd y \; 
(-\Deb) (xy)
\nonumber\\
&=&
-4
\int_0^1 \frac{\dd x}{x}\;
\int_0^{\be x} \frac{\dd u}{4 \cosh^2 \frac{u}{2}}
=
-4
\int_0^{\be } \frac{\dd u}{4 \cosh^2 \frac{u}{2}}
\; \ln \frac{\be}{u}
\nonumber\\
&=&
-2 \ln \be
\int_0^\infty \frac{\dd u}{2 \cosh^2 \frac{u}{2}}
+ 2 K 
- \int_\be^\infty
\frac{\dd u}{\cosh^2 \frac{u}{2}}\;
\ln \frac{u}{\be} .
\nonumber
\end{eqnarray}
The last integral is exponentially small in $\be $ because of the decay of 
$1/\cosh^2$. 
\end{proof}

\subsubsection{The particle--particle bubble}
This time write $x'=x-\xi$, $y'=y-\et$. The bubble is 
\begin{eqnarray}
B_{\rm pp}(q_0,\xi,\et) 
&=&
\int_{[-1,1]^2} \dd x \dd y \; 
\frac{1}{\be} \sum_\omega
\frac{1}{\I \omega - xy}
\frac{1}{\I (q_0 -\omega) - x'y'}
\nonumber\\
&=&
\int_{[-1,1]^2} \dd x \dd y \; 
\frac{\Feb (-xy) - \Feb(x'y')}{-\I q_0 + xy + x'y'}
\nonumber\\
&\hskip-15pt\gtoas{(\xi,\et) \to 0}\hskip-15pt&
\int_{[-1,1]^2} \dd x \dd y \; 
\frac{1}{-\I q_0 + 2 xy}\;
\tanh \left(\frac{\be}{2} xy\right)
\nonumber
\end{eqnarray}
Again, we set $q_0 =0$ and keep $\be < \infty$. Then 
\begin{eqnarray}
B_{\rm pp}^{(0)}
&=&
B_{\rm pp}(0,0,0)
\nonumber\\
&=&
\int_{[-1,1]^2} \dd x \dd y \; 
\frac{1}{2 xy}\;
\tanh \left(\frac{\be}{2} xy\right)
=
2
\int_{[0,1]^2} \dd x \dd y \; 
\frac{1}{xy}\;
\tanh \left(\frac{\be}{2} xy\right)
\nonumber\\
&=&
2
\int_0^1 \frac{\dd x}{x}
\int_0^x \frac{\dd E}{E}\; \tanh \frac{\be}{2} E
=
-2 \int_0^1 \dd E\; \frac{\ln E }{E}\; \tanh \frac{\be}{2} E
\nonumber\\
&=&
\left[
- (\ln E)^2 \; \tanh \frac{\be}{2} E
\right]_0^1
+
 \int_0^1 \dd E\; (\ln E )^2\; \frac{\be}{2 \cosh^2 \frac{\be}{2} E}
\nonumber\\
&=&
\int_0^{\be/2} \dd v \; 
\left(
\ln \frac{2v}{\be}
\right)^2
\frac{1}{ \cosh^2 v}
\nonumber\\
&=&
\int_0^{\infty} \dd v \; 
\left(
\ln \frac{2v}{\be}
\right)^2
\frac{1}{ \cosh^2 v}
-
\int_{\be/2}^\infty \dd v \; 
\left(
\ln \frac{2v}{\be}
\right)^2
\frac{1}{ \cosh^2 v}
\nonumber
\end{eqnarray}
Thus we have

\begin{lemma}
$$
B_{\rm pp}^{(0)}
=
(\ln \be)^2 -2K\ln \be+K' 
+
O\left(\E^{-\be}\right)
$$
where $K=\int_0^\infty \sfrac{\ln(2v)}{\cosh^2 v}dv$ 
and $K'=\int_0^\infty \sfrac{(\ln(2v))^2}{\cosh^2 v}dv$.
\end{lemma}

\section{Interpretation} 
\label{discussion}

\noindent
Let us discuss these results a bit more informally.
The above calculations for the $xy$ case can be summarized as follows. 
Evidently, the derivatives we were looking at diverge in the limit $q_0 \to 0$
at zero temperature. To proceed, we discuss positive temperatures and 
replace $q_0 $ by $\pi/\beta$, and thus translate 
everything into $\beta$--dependent quantities.
We have proven that to all orders $r$ in $\la$,
\begin{equation}
|\nabla \Sigma_r|
\le  \const
\end{equation}
(where the constant depends on the order $r$ in $\la$). 
In the model 
computations of the last section, we have seen
that to second order in the coupling constant $\la$
\begin{eqnarray*}
\mbox{ Im } \del_0 \Sigma 
&\sim& 
-4 \ln 2 \; (\la \ln \beta)^2
\\
\mbox{ Re }\sfrac{\del^2}{\del\xi\del\et} \Sigma
&\sim& 
(2 + 4 \ln 2) \; (\la \ln \beta)^2
\\
\mbox{ Re }\sfrac{\del^2}{\del\xi^2} \Sigma
&\sim& 
O (\la^2 \ln \beta)
\\
B_{\rm ph} (0)
&\sim& 
-2 \; \la \ln \beta
\\
B_{\rm pp} (0)
&\sim& 
 \la \; (\ln \beta)^2 
- \frac{\al}{2} (\la \ln \beta)^2 
\end{eqnarray*}
We have redefined $B_{\rm ph}$ and $B_{\rm pp}$ to include the appropriate
coupling constant dependence.
The line for $B_{\rm pp} $ includes second order corrections 
to the superconducting vertex, where the coefficient given by the
loop integral is $\al$. If $\la $ is negative (attractive interaction),
the superconducting instability driven by the $\la (\ln \beta)^2$ term
is always strongest, but if $\la >0$ (repulsive bare interaction), 
the $\la (\ln \beta)^2$ term suppresses the leading order Cooper 
pair interaction,  
and the higher order term proportional to $\alpha$ is only of order
$(\la \ln \beta)^2$. In this case, all terms that can drive instabilities
are linear or quadratic in $\la \ln \beta$. 
A first attempt to weigh the relative strength of these 
divergences is to look at the prefactors. 
Here it seems 
that the second derivative gets the largest contribution. The asymmetry between 
this logarithmic divergence and the boundedness of the gradient 
is striking. 

In the following we discuss the physical significance of the above
findings. This discussion is not rigorous, but it reveals some interesting 
possibilities that should be studied further. 

We first note that the above results were achieved in renormalized 
perturbation theory, that is, the above all--order results are in the 
context of an expansion where a counterterm is used to fix the Fermi 
surface. (The second order explicit calculations assume only 
that the first order corrections can be taken into account by a shift 
in $\mu$, which is true because the interaction is local.) 
Using a counterterm is the only way to get an all--order expansion 
that is well--defined in the limit of zero temperature. A scheme where 
the scale decomposition adapts to the Fermi surface movement was outlined 
in \cite{PeSaICMP} and developed mathematically in \cite{PePhD} and 
\cite{PeSa}; but in any such scheme the expansions have to be done iteratively 
and cannot be cast in the form of a single renormalized expansion, because the 
singularity moves in every adjustment of the Fermi surface \cite{PeSaICMP}.
The meaning of the counterterm was explained in detail in 
\cite{FST1,FST2, FST3,FST4}. In brief, the dispersion relation we are using in 
our propagators is not the bare one, but the renormalized one, whose zero 
level set is the Fermi surface of the interacting system. 
Thus, we have in fact made the assumption that

\begin{quote}
the Fermi surface of the {\it interacting} system has the properties \Ha--\Hf.
In particular it  contains singular points and these singular points are 
nondegenerate. 
\end{quote}

\noindent
We shall first discuss our results under this assumption and then 
speculate about how commonly it will be valid. 
In the case of a  nonsingular, strictly convex, curved Fermi surface, 
there was a similar assumption, which was, however, mathematically justified 
by our proof of an inversion theorem \cite{FST4} that gives a bijective 
relation between the free and interacting dispersion relation. 

\subsection{Asymmetry and Fermi velocity suppression}
First note that the regularized (discrete--time) 
functional integral for many--body systems
has a symmetry that allows one to make an arbitrary {\em nonzero} rescaling
of the field variables. This is based on the behaviour of the measure under 
$\psi_i \to g_i \psi_i$ and $\psq_i \to \tilde g_i \psq_i$. The result is
\begin{equation}\label{ggt}
\int
\prod_i \dd \psq_i \dd \ps_i \; \E^{- \cA (\psq,\ps)}
=
\left(\prod_i g_i \tilde g_i \right)^{-1}
\prod_i \dd \psq_i \dd \ps_i \; \E^{- \cA (\tilde g^{-1}\psq, g^{-1}\ps)}
\end{equation}
Source terms get a similar rescaling. This can, of course, be used to 
remove a factor $Z^{-1}$  from the quadratic term in the fields,
but the factor $Z$ (see \Ref{Zq}) will then reappear in the interaction and source terms.
Note that $Z$ depends on momentum. 
In the limit of very small $Z$, some terms
may get greatly enhanced or suppressed. 

The Fermi velocity is defined as $v_F(p) = Z \nabla \Si$ on the Fermi surface. 
If one extrapolates the above formula for $\del_0 \Si_2$ to $\beta \to \infty$,
the $Z$ factor, defined in \Ref{Zq}, becomes 
zero at the Van Hove points. 

There is a crucial difference between the one-- and two--dimensional
cases. In dimension one, both $\del_0 \Sigma_2$ {\em and} $\nabla \Sigma_2$ 
behave like $\la^2 \log \beta$, so that, after extracting the 
field strength renormalization, the Fermi velocity retains its original value. 
In our two--dimensional situation, however, 
$\nabla \Sigma$ remains bounded, and thus the Fermi velocity gets 
suppressed in a neighbourhood of the Van Hove singularity because it contains 
a factor of $Z$.  
Because the time derivative term in the action is 
$Z^{-1} \psq \I k_0 \ps$ and $k_0$ is an odd multiple of the temperature 
$T= \beta^{-1}$, 
one can also interpret $Z(k)^{-1} T =  T(k)$ heuristically as a 
``momentum--dependent temperature'' that varies over the Fermi surface
and that increases as one approaches the Van Hove points (``hot spots'').
This behaviour is illustrated in Figure \ref{fig1}.  

\begin{figure}

\epsfxsize=14cm
\centerline{\epsffile{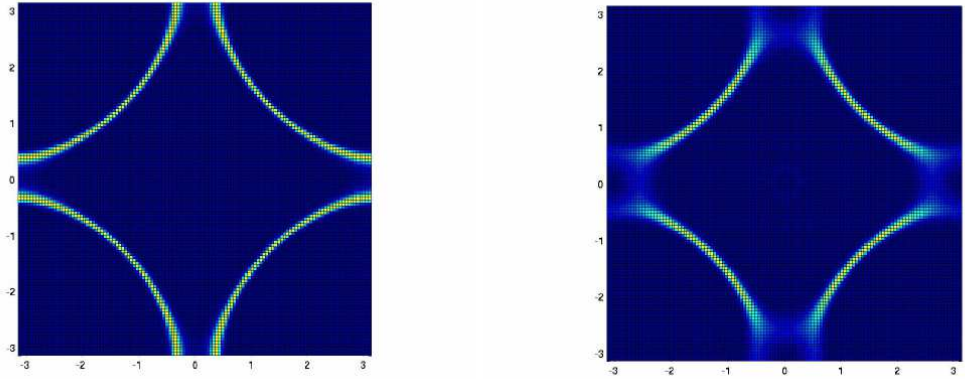}}

\caption{Comparison of the derivative of the Fermi function, 
where on the right, the inverse temperature $\beta$ is scaled 
by an angle--dependent factor that vanishes
near the points $(\pi,0)$ and $(0,\pi)$. The dark regions correspond 
to small derivatives, the bright ones to large derivatives.
}
\label{fig1}
\end{figure}

\subsection{Inversion problem}
There is, however, also a more basic problem that is exhibited 
by these results. It arises when one starts 
questioning the assumption that the {\it interacting} Fermi surface
contains singularities. 
The second derivative of the self--energy 
in spatial directions is divergent at zero temperature. 
Thus even the lowest nontrivial correction may change the structure
of $e$ significantly. This is related to the inversion problem, 
which was solved for strictly convex Fermi surfaces in \cite{FST4}.

In all of the following discussion, we concentrate on two dimensions and
assume that by adjusting $\mu$, we can arrange things such that the 
interacting Fermi surface still contains a point where the gradient 
of  the dispersion relation vanishes. 

In the general theory of expansions for many--fermion systems, 
the most singular terms are those created by the movement
of the Fermi surface. Counterterms need to be used to make
perturbation theory well-defined, or an adaptive scale decomposition
has to be chosen. To show that the model with counterterms 
corresponds to a  bare model in some desired class, one needs to prove 
an inversion theorem, which requires regularity estimates. 
To get an idea, it is instructive to consider a neighbourhood 
of a  regular point on the Fermi surface.  By an affine transformation
to coordinates $(u,v)$ in momentum space, 
i.e. shifting the origin to that point and rotating so that  the 
tangent plane to the Fermi surface  is given by $u=0$, the 
function $e$ can be transformed to a function 
$\tilde e (u,v) = u + \frac{\kappa}{2} v^2 + \ldots$, 
where $\kappa$ denotes the curvature and $\ldots$ the 
higher order terms. (By an additional change of coordinates, 
the error terms can be made to vanish, but this is not important here.) 
Let $\Sigma_0 ( p)$ be the self--energy at frequency zero 
(or $\pm \pi/\beta$).  Then 
$$
(e+\Sigma_0)^{\sim} (u,v)
=
u + \frac{\kappa}{2} v^2 +  \la \si (u,v) .
$$
In order for the correction $\la\si$ not to overwhelm the zeroth 
order contribution $e$, we need $\del_u \si$ to be bounded and 
$\del_v^2 \si$ to be bounded. Then the correction is small when $\la$
is small. Note that different regularity is needed in normal 
and tangential directions. Indeed, even in the absence of Van Hove 
singularities, the function $\si$ is not twice differentiable in $u$ at 
$T=0$, but it is $C^2$ in $v$ if $e$ has the same property \cite{FST2}.

Thus, the mere divergence of a derivative does not tell us about a potential
problem with renormalization; the relevant question is whether the correction
is larger than the zeroth order term.  In the strictly convex, curved case, 
there is no problem.  However, in two dimensions, the normal form for the 
Van Hove singularity, $\tilde e (uv) = uv$ gets changed significantly 
because 
\begin{equation}
\tze_2
=
\sfrac{\del^2\ }{\del u \del v}
\left[
e(u,v) + \tilde \si (u,v)
\right]\big|_{u=v=0}
\end{equation}
diverges at zero temperature. 
To leading order in $q_0$, $u$ and $v$,
our result for the inverse full propagator to second order 
takes the form
$
\ze_2 \; \I q_0  - \tze_2 \; uv
$
with 
\begin{eqnarray*}
\ze_2
&=&
1+ \vth_2 (\la \ln\beta)^2
\\
\tze_2
&=&
1+ \tilde\vth_2  (\la \ln\beta)^2 .
\end{eqnarray*}
Because $\tze_2$ diverges in the zero--temperature limit, 
it is not even in second order consistent to assume that
the type of $(u,v) = (0,0)$ as a critical point  of
$e$ and $E = e + \tilde \sigma$ is the same.
One can attempt to fix this problem by rescaling the fields by $\tze^{1/2}$ .
Then the $\I q_0$ term gets rescaled similarly because $\ze$ depends
on the same combination of $\la $ and $\ln \beta$ as $\tze$. 
It should, however, be noted that,
as in the above--discussed suppression of the Fermi velocity,
a corresponding rescaling of all interaction and source terms also occurs. 
In particular, $\si$, which gives the Fermi surface shift, 
itself gets replaced by $\si \tze^{-1}$. 
Because $\si$ remains bounded (only its derivatives diverge),
this would imply that the Fermi surface shift gets scaled down in 
the rescaling transformation, 
indicating a ``pinning'' of the Fermi surface at the Van Hove points. 

In our case, $\vth_2 = 4 \ln 2$ and $\tilde \vth_2 = 2 + 4 \ln 2$. 
With these numerical values, $\tze_2$ is larger than $\ze_2$, 
so $\tze_2$ appears as the natural factor for the rescaling. 
In the Hubbard model case, the transformations leading to the 
normal form depend on the parameter $\theta$ in \Ref{hubdisp}, 
so that one can expect 
$\ze$ and $\tze$ to depend on $\theta$. The study of these
dependencies, as well as the interplay between the two critical
points that contribute, is left to future work. 

One point of criticism of the Van Hove scenario has always been 
that it appears nongeneric because the logarithmic singularity
gets weakened very fast as one moves away from the Van Hove 
densities. 
If the above speculations can be substantiated by careful studies, 
the singular Fermi surface scenario may turn out to be much more natural 
than one would naively assume. Moreover, there is a natural way how extended
Van Hove singularities may arise by interaction effects. 

\appendix

\section{Interval Lemma}

The following standard result is included for the convenience of the reader.

\begin{lemma}\label{le:interval}
Let $\ep,\ \eta$ be strictly positive real numbers and $k$
be a strictly positive integer. Let $I\subset\bR$ be an interval (not
necessarily compact) and $f$ a $C^k$ function on $I$ obeying
$$
|f^{(k)}(x)|\ge\et\qquad{\rm for\ all\ }x\in I
$$ 
Then
$$
{\rm Vol}\set{x\in I}{|f(x)|\le\ep}\le 2^{k+1}\big(\sfrac{\ep}{\et}\big)^{1/k} \; .
$$
\end{lemma}

\begin{proof}
Denote $\al =\big(\sfrac{\ep}{\et}\big)^{1/k}$. In terms $\al$, we must show
$$
|f^{(k)}(x)|\ge\sfrac{\ep}{\al^k}\quad{\rm for\ all\ }x\in I
\quad\Longrightarrow\quad
{\rm Vol}\set{x\in I}{|f(x)|\le\ep}\le 2^{k+1}\al
$$
Define $c_k$ inductively by $c_1=2$ and $c_k=2+2c_{k-1}$. Because
$b_k=2^{-k}c_k$ obeys $b_1=1$ and $b_k=2^{-k+1}+b_{k-1}$ we have 
$b_k\le 2$ and hence $c_k\le 2^{k+1}$. We shall prove
$$
|f^{(k)}(x)|\ge\sfrac{\ep}{\al^k}\quad{\rm for\ all\ }x\in I
\quad\Longrightarrow\quad
{\rm Vol}\set{x\in I}{|f(x)|\le\ep}\le c_k\al
$$
by induction on $k$.

Suppose that $k=1$ and let $x$ and $y$ be any two points in 
$\set{x\in I}{|f(x)|\le\ep}$. Then
$$
|x-y|=\sfrac{|x-y|}{|f(x)-f(y)|}|f(x)-f(y)|=\sfrac{|f(x)-f(y)|}{|f'(\ze)|}
\le \sfrac{2\ep}{|f'(\ze)|}
$$
for some $\ze\in I$. As $|f'(\ze)|\ge\sfrac{\ep}{\al}$ we have $|x-y|\le2\al$.
Thus $\set{x\in I}{|f(x)|\le\ep}$ is contained in an interval of length
at most $2\al$ as desired.

Now suppose that the induction hypothesis is satisfied for $k-1$ and that
$|f^{(k)}(x)|\ge\sfrac{\ep}{\al^k}$ on $I$. As in the last paragraph the
set $\set{x\in I}{|f^{(k-1)}(x)|\le\sfrac{\ep}{\al^{k-1}}}$ is contained
in a subinterval $I_0$ of $I$ of length at most $2\al$. Then $I$ is the
union of $I_0$ and at most two other intervals $I_+,I_-$ on which 
$|f^{(k-1)}(x)|\ge\sfrac{\ep}{\al^{k-1}}$. By the inductive hypothesis
\begin{eqnarray*}
{\rm Vol}\set{x\in I}{|f(x)|\le\ep}
&\le&
{\rm Vol}(I_0)+\sum_{i=\pm}{\rm Vol}\set{x\in I_i}{|f(x)|\le\ep}
\nonumber\\
&\le& 
2\al+2c_{k-1}\al=c_k\al
\nonumber
\end{eqnarray*}
\end{proof}

\section{Signs etc.} 
\label{bloodysigns}
The restrictions \Ref{restri} are summarized in the following table 
\begin{equation}\label{restri2}
\begin{array}{ccc|c}
E_2=xy&E_3=x'y'&E_1=(x-x')(y-y')&s_f\\
\hline
+&-&-& (-1) \\
-&+&+& (-1)
\end{array}
\end{equation}
In both cases, the product of indicator functions resulting from 
the limit of Fermi functions is $-1$. 

Let $R$ be the reflection at zero, 
\begin{equation}\label{reflat0}
R(x,y,x',y') = (-x,-y,-x',-y'). 
\end{equation}
The function
$$
\veps 
=
\veps (x,y,x',y')
=
xy'+x'y - 2x'y'
$$
satisfies
$$
\veps (x,y,x',y')
=
\veps (-x,-y,-x',-y'). 
$$
The function $D(x,y,x',y') = y-y'$ satisfies
$$
D(R(x,y,x',y')) = - D (x,y,x',y')
$$
The function $F(x,y,x',y') = (x-x') (y-y')$ is invariant under $R$. 

In the following we list all cases of signs for $x$, $x'$, $y$ and $y'$, 
together with $\veps, D, F$ written as functions of 
$$
\rx=|x|, \ry=|y|, \rx'=|x'|, \ry'=|y'|. 
$$
to be able to restrict the integrals to $[0,1]$ whenever this is 
convenient (by transforming to $\rx, \ldots, \ry'$ as integration variables),
and to exhibit some important sign changes. 
In the last column, we list the condition $\rho_n$ obtained from the 
restriction on the sign of $(x-x')(y-y')$ in \Ref{restri2}.

$$
\begin{array}{rc|l|l|r|c}
n&x\,y\;x'\,y'&\veps_n=xy'+x'y - 2x'y'&D_n=y-y'& F_n & \rho_n\\
\hline
1&+++-& \veps_1=\rx'\ry+(2\rx'-\rx)\ry' & D_1=\ry+\ry'& (\rx-\rx')(\ry+\ry')   & \rx < \rx'\\
2&++-+& \veps_2=\rx \ry'+(2\ry'-\ry) \rx'& D_2=\ry-\ry'& (\rx+\rx')(\ry-\ry')   & \ry < \ry'\\
3&+-++& \veps_3=\rx \ry'-(2\ry'+\ry)\rx'& D_3=-(\ry+\ry')& - (\rx-\rx')(\ry+\ry') & \rx < \rx'\\
4&+---& \veps_4=\rx'\ry-(2\rx'+\rx)\ry'& D_4=-(\ry-\ry')& - (\rx+\rx')(\ry-\ry') & \ry < \ry'\\
\hline
5&++++& \veps_5=   \rx\ry'-\rx'(2\ry'-\ry) & & & \\
6&++--& \veps_6= -(\rx\ry'+\rx'(2\ry'+\ry)) & & & \\
7&+-+-& \veps_7= -(\rx\ry'-\rx'(2\ry'-\ry)) & & & \\
8&+--+& \veps_8=   \rx\ry'+\rx'(2\ry'+\ry) & & & \\ 
\hline
9 &---+&    \veps_1  & - D_1 & & \rho_1 \\ 
10&--+-&   \veps_2 & - D_2 &  & \rho_2 \\
11&-+--&   \veps_3 & - D_3 &  & \rho_3 \\
12&-+++& \veps_4 & - D_4 &  & \rho_4 \\
\hline
13&----& \veps_5 & & & \\
14&--++& \veps_6 & & & \\
15&-+-+& \veps_7 & & & \\
16&-++-& \veps_8 & & & \\
\end{array}
$$
Cases 1-4 and 9-12 obey the restrictions \Ref{restri2}. 
Cases 5-8 and 13-16 do not because there, the signs of $xy$ and $x'y'$ are the same. 
They are used to discuss some terms at finite $\beta$. Since the first two restrictions
are not satisfied, the column for the last restriction, $\rho$, is left empty 
in these cases. Case $n+8$ is obtained from $n$ by the reflection $R$. 
Thus 
\begin{equation}\label{eq:n+8}
\veps_{n+8} = \veps_n, \quad D_{n+8} = -D_n, \quad F_{n+8} = F_n, \quad \rho_{n+8} = \rho_n
\end{equation}
Moreover, 
\begin{equation}\label{eq:1-3}
\veps_1= - \veps_3 \mbox{ and } \veps_2= - \veps_4
\end{equation}
and 
\begin{equation}\label{eq:1-2}
\veps_2 (\ry,\rx,\ry',\rx')   
=
\veps_1 (\rx,\ry,\rx',\ry')
\mbox{ and }
\rho_2 (\ry,\rx,\ry',\rx')  
\Leftrightarrow
\rho_1 (\rx,\ry,\rx',\ry') .
\end{equation}

\end{document}